\documentclass{article}

    \usepackage[nonatbib, preprint]{neurips_2023}

\usepackage[utf8]{inputenc} 
\usepackage[T1]{fontenc}    
\usepackage{hyperref}       
\usepackage{url}            
\usepackage{booktabs}       
\usepackage{amsfonts}       
\usepackage{nicefrac}       
\usepackage{microtype}      
\usepackage{xcolor}         

\usepackage{esvect}
\usepackage{dsfont}
\usepackage{amsmath}
\usepackage{amsthm}
\usepackage{amssymb}

\newtheorem{definition}{Definition}
\newtheorem{condition}[definition]{Condition}
\newtheorem{theorem}{Theorem}[section]
\newtheorem{lemma}{Lemma}
\newtheorem{corollary}[lemma]{Corollary}
\newtheorem*{proof*}{Proof}

\newcommand{\mapped}[1]{\overset{#1}{\sim}}
\newcommand{\notmapped}[1]{\overset{#1}{\not\sim}}
\newcommand{\define}{\triangleq}
\newcommand{\bcomment}[1]{}

\newcommand{\blue}[1]{\textcolor{blue}{#1}}
\newcommand{\red}[1]{\textcolor{red}{#1}}

\newcommand{\mtrx}[1]{\mathbf{#1}}
\newcommand{\vctr}[1]{\vv{#1}}

\newcommand{\eql}[1]{\overset{#1}{=}}

\newcommand{\pheq}{\phantom{=}}

\newcommand{\pth}[1]{\left(#1\right)} 
\newcommand{\croc}[1]{\left[#1\right]} 
\newcommand{\acc}[1]{\left\{#1\right\}} 
\newcommand{\ip}[1]{\left\langle#1\right\rangle} 

\newcommand{\bpth}[1]{\big(#1\big)}

\newcommand{\crocMat}[4]{\begin{bmatrix} #1 & #2 \\ #3 & #4\end{bmatrix}}
\newcommand{\crocVec}[2]{\begin{bmatrix} #1 \\ #2\end{bmatrix}}

\newcommand{\diag}{\operatorname{diag}}
\newcommand{\tr}{\operatorname{tr}}
\newcommand{\Var}{\operatorname{Var}}
\newcommand{\Cov}{\operatorname{Cov}}

\newcommand{\calA}{\mathcal{A}}
\newcommand{\calB}{\mathcal{B}}

\newcommand{\calD}{\mathcal{D}}

\newcommand{\calN}{\mathcal{N}}
\newcommand{\calO}{\mathcal{O}}

\newcommand{\calU}{\mathcal{U}}
\newcommand{\calV}{\mathcal{V}}
\newcommand{\calW}{\mathcal{W}}

\newcommand{\bbE}{\mathbb{E}}

\newcommand{\bbN}{\mathbb{N}}

\newcommand{\bbR}{\mathbb{R}}

\newcommand{\zeromat}{\mtrx{0}}

\newcommand{\amat}{\mtrx{a}}

\newcommand{\mmat}{\mtrx{m}}

\newcommand{\Amat}{\mtrx{A}}

\newcommand{\Emat}{\mtrx{E}}

\newcommand{\Gmat}{\mtrx{G}}

\newcommand{\Imat}{\mtrx{I}}

\newcommand{\Lmat}{\mtrx{L}}
\newcommand{\Mmat}{\mtrx{M}}

\newcommand{\Pmat}{\mtrx{P}}
\newcommand{\Qmat}{\mtrx{Q}}

\newcommand{\Umat}{\mtrx{U}}
\newcommand{\Vmat}{\mtrx{V}}
\newcommand{\Wmat}{\mtrx{W}}

\newcommand{\Thetamat}{\mtrx{\Theta}}
\newcommand{\Sigmamat}{\mtrx{\Sigma}}

\newcommand{\onevec}{\vctr{1}}
\newcommand{\zerovec}{\vctr{0}}

\newcommand{\avec}{\vctr{a}}
\newcommand{\bvec}{\vctr{b}}

\newcommand{\fvec}{\vctr{f}}

\newcommand{\mvec}{\vctr{m}}

\newcommand{\tvec}{\vctr{t}}

\newcommand{\xvec}{\vctr{x}}
\newcommand{\yvec}{\vctr{y}}

\newcommand{\Avec}{\vctr{A}}
\newcommand{\Bvec}{\vctr{B}}

\newcommand{\Xvec}{\vctr{X}}
\newcommand{\Yvec}{\vctr{Y}}

\newcommand{\alphavec}{\vctr{\alpha}}
\newcommand{\betavec}{\vctr{\beta}}

\newcommand{\muvec}{\vctr{\mu}}

\newcommand{\rhovec}{\vctr{\rho}}

\newcommand{\TODO}[1]{}
\usepackage{paralist}
\usepackage{tikz}
\usetikzlibrary{matrix}
\usepackage{pgfplots}

\title{Gaussian Database Alignment and Gaussian Planted Matching}

\author{%
  Osman E.~Dai \\
  School of Industrial and Systems Engineering \\
  Georgia Institute of Technology \\
  Atlanta, GA 30313, USA \\
  \texttt{oedai@gatech.edu} \\
  \And
  Daniel~Cullina \\
  School of Electrical Engineering and Computer Science \\
  Pennsylvania State University\\
  State College, PA 16801, USA \\
  \texttt{dqc5596@psu.edu} \\
  \And
  Negar~Kiyavash \\
  School of Computer and Communication Sciences \\
  Ecole Polytechnique Federale de Lausanne\\
  1015 Lausanne, Switzerland \\
  \texttt{negar.kiyavash@epfl.ch} \\
}

\begin{document}

\maketitle

\begin{abstract}
  Database alignment is a variant of the graph alignment problem: Given a pair of anonymized databases containing separate yet correlated features for a set of users, the problem is to identify the correspondence between the features and align the anonymized user sets based on correlation alone. This closely relates to planted matching, where given a bigraph with random weights, the goal is to identify the underlying matching that generated the given weights. We study an instance of the database alignment problem with multivariate Gaussian features and derive results that apply both for database alignment and for planted matching, demonstrating the connection between them. The performance thresholds for database alignment converge to that for planted matching when the dimensionality of the database features is \(\omega(\log n)\), where \(n\) is the size of the alignment, and no individual feature is too strong. The maximum likelihood algorithms for both planted matching and database alignment take the form of a linear program and we study relaxations to better understand the significance of various constraints under various conditions and present achievability and converse bounds. Our results show that the almost-exact alignment threshold for the relaxed algorithms coincide with that of maximum likelihood, while there is a gap between the exact alignment thresholds. Our analysis and results extend to the unbalanced case where one user set is not fully covered by the alignment.
\end{abstract}

\section{Introduction}

The modern ubiquity of data collection has drawn interest to the problem of data alignment, which is described as follows: We have two data sets containing information regarding various anonymized users. Both sets might contain data associated with a particular user, in which case we observe correlation between data of given user. Alignment is the problem of identifying such correlated pairs. This enables data merging (e.g. in the field of computational biology \cite{singh2008global} or computer vision \cite{conte2004thirty}) or de-anonymization (with several high-profile instances, such as the 2006 Netflix Prize incident \cite{narayanan2008robust} or 2016 release of the MBS/PBS healthcare data \cite{culnane2017health}).

Database alignment and graph alignment are two well studied versions of the alignment problem. In the former setting, the data consists of multi-dimensional features, each associated with an individual user. These features are correlated across the two databases only if there are associated with the same user. In the graph setting, features are associated with pairs of users, and they are correlated across the two graphs only if the pairs match.

A line of work has studied alignment of correlated Erd{\H{o}}s-R{\'e}nyi graphs, identifying information theoretic bounds for exact alignment \cite{pedarsani2011privacy,lyzinski2014seeded, cullina2017exact}, for partial alignment \cite{ganassali2021impossibility, hall2022partial}, analyzing various efficient and nearly-efficient algorithms \cite{shirani2017seeded, barak2019nearly, dai2019analysis, ding2021efficient, fan2022spectralA, fan2022spectralB,mao2022random, wu2022settling}. There are other results on variants of this problem, considering alignment with side information \cite{lyzinski2014seeded, shirani2017seeded, mossel2020seeded} or alignment of graphs with Gaussian edge weights \cite{ganassali2022sharp, wu2022settling}.

For database alignment, the earliest result identified the sharp information theoretic condition for the exact alignment of databases with finite-alphabet features \cite{cullina2018fundamental}. A later study identified tight bounds for almost-exact alignment when features are high-dimensional and each dimension is i.i.d with an arbitrary distribution \cite{shirani2019concentration}. Several works focused on Gaussian databases, presenting the sharp conditions of exact and almost-exact alignment of databases with Gaussian features \cite{dai2019database} and identifying the elliptic boundary that show the order of magnitude of errors that is achievable within the almost-exact alignment region \cite{dai2020achievability}. Another work studied the related problem of testing whether Gaussian databases are correlated \cite{zeynep2022detecting}.

Planted matching is a closely related problem. In this setting, we consider a bipartite graph, with random edge weights, over a pair of user sets with an underlying true matching. The weights of edges across true user pairs are different from those across users that are not matched under the true matching. In the original formulation \cite{chertkov2010inference}, it is motivated by the problem of tracking moving particles between two images. With particles in the two images forming the vertex set of the bipartite graph, each edge weight is some calculated measure of likelihood of two particles from different images correspond to each other. This can be considered as variant of the database alignment formulation where the features in each database is the information about the particle in each image. One significant difference between the two formulations is that in planted matching, edge weights are typically considered to be independent, while for database alignment, the actual likelihood measures between pairs are not truly independent. This follows from the observation that, in database alignment, the likelihood of the matched pair \((u,v)\) and that of \((u,v')\) both depend on the value of the feature associated with \(u\).

Earliest work studied planted matching in the case where the non-matched edges have uniform distribution while the distribution matched edges is folded Gaussian \cite{chertkov2010inference, semerjian2020recovery}. Later work studied the case where all edges are exponentially distributed, the matched ones having finite mean while the non-matched edges have mean on the order of the number of vertices in the bipartite graph \cite{moharrami2021planted}, which was then extended to  consider a variant of the problem with side information, where only a subset of the vertex pairs are eligible to be included in the matching, i.e. the matching is planted in a non-complete bipartite graph \cite{ding2021planted}. This last study established the sharp information theoretic conditions for almost-exact alignment in the case where the matched edge has a distribution with fixed density (i.e. not dependent on the size of the vertex set) while that of the non-matched edges is scaled (i.e. stretched) by the average degree in the bipartite graph.

The Gaussian case we study is distinct from the planted matching cases described above. We consider the case where all edges are Gaussian with unit variance and some distance between the means. In the regime of interest, the difference in means scales with the square root of log of the number of vertices. We derive sharp thresholds that dictate the order of magnitude of errors at any signal level. Furthermore, we extend our analysis to the unbalanced case where the two vertex sets differ in size and the true matching is injective but not surjective. We also consider the performance of various relaxations of the maximum likelihood estimator to understand the dynamics of this problem and the significance of various constraints at different signal levels.

\section{Model}

\subsection{Notation}

Random variables are denoted by upper-case Latin letters while their instances are denoted with the corresponding lowercase letter. Vectors and vector-valued functions are expressed by arrows (e.g. \(\avec\)) while bold font is used for matrices (e.g. \(\amat\)). 
The natural numbers and real numbers are denoted \(\bbN\) and \(\bbR\) respectively while calligraphic notation is used for others sets (e.g. \(\calA\)).

\subsection{Correlated Gaussian Databases}
\label{subsec:modelDatabase}

Let \(\calU\) and \(\calV\) denote sets of users. Let \(M\) denote an underlying partial mapping between \(\calU\) and \(\calV\): a bijection between some subset of \(\calU\) and \(\calV\). We write \(u\mapped{M}v\) if \(u\in\calU\) and \(v\in\calV\) are mapped to each other by \(M\). Any user from one database is mapped to at most one user from the other database. We use \(|M|\) to denote the number of pairs mapped by \(M\). Let \(\Mmat\) denote the matrix encoding of the mapping in \(\{0,1\}^{\calU\times\calV}\) such that \(M_{u,v}=1\) if \(u\) and \(v\) are mapped.

Databases are represented by functions that return feature vectors for each user in the relevant user set. \(\Avec:\calU\to\bbR^{\calD_a}\) and \(\Bvec:\calV\to\bbR^{\calD_b}\) are the databases associated with the two sets of users. The features in each database are indexed by elements in the sets \(\calD_a\) and \(\calD_b\) respectively.

We say \(\Avec\) and \(\Bvec\) are a pair of correlated Gaussian databases with covariance \(\crocMat{\Sigmamat_a}{\Sigmamat_{ab}}{\Sigmamat_{ab}^\top}{\Sigmamat_b}\) if
\begin{compactitem}
    \item All entries in \(\Avec\) and \(\Bvec\) are, together, jointly Gaussian.
    \item \(\Avec(u)\) is independent and identically distributed with variance \(\Sigmamat_a\) for every \(u\in\calU\).\\
    Similarly, \(\Bvec(v)\) is independent and identically distributed with variance \(\Sigmamat_b\) for every \(v\in\calV\).
    \item \(\textrm{Cov}\pth{\Avec(u), \Bvec(v)} = \Sigmamat_{ab}\) if \(u\mapped{M}v\) and
    \(\textrm{Cov}\pth{\Avec(u), \Bvec(v)} = 0\) if \(u\notmapped{M}v\).
\end{compactitem}

Under this model, features in \(\Avec\) and \(\Bvec\) may have arbitrary dimension \(|\calD_a|\) and \(|\calD_b|\) respectively. However, as shown in \hyperref[sec:canonical]{Appendix \ref*{sec:canonical}}, knowledge of the statistic \(\Sigmamat\) can be used to perform linear transformations on features from each database and eliminate degrees of freedom of the features that are not correlated with the other database.

\textbf{\underline{Problem setting:}} We consider the scenario where we observe a pair of correlated Gaussian databases \(\Avec\) and \(\Bvec\) with an unknown partial mapping \(M\) between the sets of users \(\calU\) and \(\calV\). The statistics \(\muvec\) and \(\Sigmamat\) of the i.i.d. distribution of correlated feature pairs are known. We have no prior knowledge of the mapping \(M\) beyond its size. We say the problem is unbalanced if \(|\calU|\neq |\calV|\).

\subsection{Planted Matching on Gaussian Bigraph}
\label{subsec:modelPlanted}

\(\calU\), \(\calV\) and \(M\) are defined as in \hyperref[subsec:modelDatabase]{Subsection  \ref*{subsec:modelDatabase}}. Given parameter \(\mu\in\bbR\), let \(\Wmat\), taking values in \(\bbR^{\calU\times\calV}\), denote the weight matrix of bipartite graph over \(\calU\) and \(\calV\) such that, given \(M\), \(\Wmat\) has independent Gaussian entries with unit variance and mean \(\bbE[W_{u,v}]=\mu\) if \(u\mapped{M}v\) and \(\bbE[W_{u,v}]=0\) otherwise. Without loss of generality, we assume \(\mu>0\).

\textbf{\underline{Problem setting:}} We observe \(\Wmat\) and want to identify the proper matching \(M\). The parameter \(\mu\) is known. We have no prior knowledge of \(M\) beyond its size.

\subsection{Algorithms}
\label{subsec:model-algo}

For database alignment, let \(\Gmat\) be the matrix such that \(G_{u,v}\) is the log-likelihood ratio of hypotheses \(u\mapped{M}v\) vs. \(u\notmapped{M}v\) for any \((u,v)\in\calU\times\calV\). (\(\Gmat\) can be calculated in polynomial time.) For planted matching, let \(\Wmat\) denote the random weight matrix.
\begin{compactitem}
    \item \textbf{Maximum likelihood estimation}: For both problems, the maximum likelihood estimator can be expressed as a linear program. Given \(|\calU|=|M|\), for arbitrary \(\tau\in\bbR\), find the maximizer \(\mmat\in\bbR^{\calU\times\calV}\) that maximizes \(\ip{\Gmat-\tau,\mmat}\) or \(\ip{\mu\Wmat-\tau,\mmat}\) under the constraints:
    \begin{align*}
        & \textrm{(a)}\,\, \sum_{u\in\calU} m_{u,v} \leq 1, \,\,\forall v\in\calV\\
        & \textrm{(b)}\,\, \sum_{v\in\calV} m_{u,v} = 1, \,\,\forall u\in\calU\\
        & \textrm{(c)}\,\, m_{u,v} \in [0,1], \,\,\,\,\, \forall (u,v)\in\calU\times\calV
    \end{align*}
    This is equivalent to the linear assignment problem and can be solved in polynomial time.
    \item \textbf{Maximum row estimation}: Removing constraint (a) gives us an algorithm that individually considers each user in \(\calU\), blind to all other users in the set. This relaxation is relevant when the mapping of a small subset of users is of interest, and not that of the entire set.
    \item \textbf{Threshold testing}: Removing (a) and (b) gives an algorithm that performs a likelihood ratio tests for each pair \((u,v)\in\calU\times\calV\) to decide whether the given pair is part of the true mapping.
\end{compactitem}
For each case, the constraint matrix is totally unimodular and therefore has an integer valued solution \(\mmat\in\{0,1\}^{\calU\times\calV}\). So the solutions always give us a mapping \(\calU\to\calV\) although not necessarily injective for maximum row estimation and not necessarily injective nor single-valued for threshold testing. A detailed description of algorithms as well as their computational complexity is given in Appendix A.

\section{Results}

\newcommand{\datA}{\Avec}
\newcommand{\datB}{\Bvec}

In this section, we summarize our main results. For database alignment, let \(I_{XY}\) denote the total mutual information between a pair of correlated features from the two databases. For planted matching, let \(\mu\) denote the difference of means. All the results are written in terms of `signal strength' \(\zeta\) which refers to \(I_{XY}\) for database alignment, and to \(\mu^2/2\) for planted matching.
\bcomment{
We analyse the performance of the maximum likelihood estimator for \(M\), which can be expressed as a linear program:

For database alignment, let \(I_{XY}\) denote the total mutual information between a pair of correlated features from the two databases. For planted matching, let \(\mu\) denote the difference of means. We consider the following three algorithms' performance in this work:
\begin{compactitem}
    \item \textbf{Maximum likelihood estimation}: Finds the complete mapping that maximizes the likelihood function.
    \item \textbf{Maximum row estimation}: For each user \(u\in\calU\), finds the user in \(v\in\calV\) that maximizes the likelihood of pair \((u,v)\). It may return a `mapping' where users in \(\calV\) are mapped to multiple users.
    \item \textbf{Threshold testing}: Performs a likelihood ratio test on every pair \((u,v)\in\calU\times\calV\) to determine whether they exceed a threshold in which case they are matched. It may return a `mapping' where users in both \(\calU\) and \(\calV\) are mapped to multiple users.
\end{compactitem}
As we show in \blue{ADD REF} \hyperref[sec:algo]{Section \ref*{sec:algo}}, maximum likelihood estimation can be expressed as a linear program, and the other two algorithms as relaxations of this linear program.
}

For database alignment, when per-feature correlation is low, measures of correlation relevant for our analysis can always be expressed in terms of mutual information \(I_{XY}\). Therefore the statements of some of our results are only accurate in this setting. We formally define this regime, where a large number of dimensions each carry infinitessimally small information:
\begin{condition}[Low per-feature correlation in database alignment]
    \label{cond:highDimensional}
    The covariance matrix \(\Sigmamat = \crocMat{\Sigmamat_a}{\Sigmamat_{ab}}{\Sigmamat_{ab}^\top}{\Sigmamat_b}\) is said to satisfy the low per-feature correlation condition if \(\left|\left|\Sigmamat_a^{-1/2}\Sigmamat_{ab}\Sigmamat_b^{-1/2}\right|\right|_2 \leq o(1)\), where \(||\cdot||_2\) denotes the \(\ell_2\) operator norm, i.e. largest singular value.
\end{condition}
Under this condition, mutual information is \(I_{XY}=\frac{1}{2}||\Sigmamat_a^{-1/2}\Sigmamat_{ab}\Sigmamat_b^{-1/2}||_F^2(1+o(1))\). The squared Frobenius norm is the sum of the singular values squared, so the condition implies \(I_{XY} = o(d)\) where \(d\) is the number of dimensions of features. Since \(I_{XY} = \Omega(\log n)\) in the regime where alignment is feasible, \hyperref[cond:highDimensional]{Condition \ref*{cond:highDimensional}} implies dimensionality \(d\geq \omega(\log n)\). The connection between this condition and the low per-feature correlation setting is shown in Appendix B.

\subsection{Achievability}

We say an algorithm achieves \underline{exact alignment} if the expected number of misaligned users is \(o(1)\), and \underline{almost-exact alignment} if the expected number of misaligned users is \(o(n)\) where \(n\) is the number of users covered by the true alignment.

\begin{figure}
    \centering
    \begin{tikzpicture}
        \begin{axis}[
            xmin = 0, xmax = 7.2,
            ymin = 0, ymax = 4.2,
            xtick distance = 1,
            ytick distance = 1,
            grid = both,
            minor tick num = 4,
            major grid style = {lightgray},
            minor grid style = {lightgray!25},
            width = \textwidth,
            height = 0.5\textwidth,
            xlabel = {\(\frac{\zeta}{\log n}\)},
            ylabel = {\(\frac{\log\pth{|\calV|-n}}{\log n}\)},
            legend cell align = {left},
            legend style={at={(0.02,0.96)},anchor= north west}
        ]
            
            \addplot [mark=none, thick, black] coordinates {(1, 0) (1, 1)};

            \addplot[
                domain = 2:4,
                samples = 100,
                smooth,
                thick,
                blue,
            ]{x/2-1};

            \addplot [mark=none, thick, green] coordinates {(4, 0) (4, 1)};

            \addplot [mark=none, thick, red] coordinates {(5.828, 0) (5.828, 1)};
            
            \addplot[
                domain = 1:8,
                samples = 100,
                smooth,
                thick,
                black,
            ]{x};

            \addplot[
                domain = 4:8,
                samples = 100,
                smooth,
                thick,
                blue,
            ]{(sqrt(x)-1)^2+0.005};
            
            \addplot[
                domain = 4:8,
                samples = 100,
                smooth,
                thick,
                green,
            ]{(sqrt(x)-1)^2-0.015};
            
            \addplot[
                domain = 5.828:8,
                samples = 100,
                smooth,
                thick,
                red,
            ]{(sqrt(x)-1)^2-1};
            
            \legend{almost-exact alignment, exact align. - max likelihood, exact align. - max row, exact align. - thresholding}
        \end{axis}
    \end{tikzpicture}
    \caption{Comparison of boundaries for achievability regions of exact alignment and almost-exact alignment. x-axis corresponds to signal strength and y-axis corresponds to order of magnitude of number of unmatched users in \(\calV\). Achievability regions are areas below/right of curves.}
    \label{fig:threshold}
\end{figure}
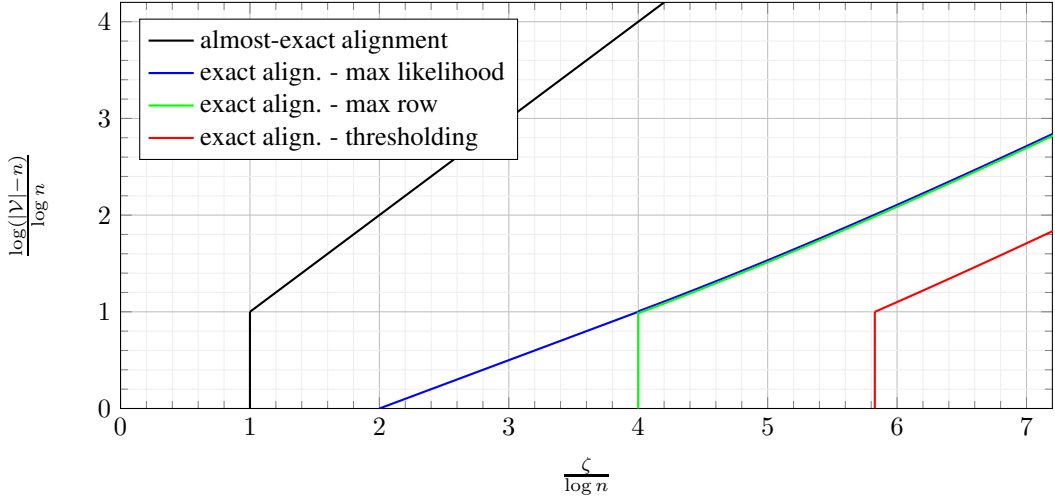

\begin{theorem}[Exact alignment and almost-exact alignment recovery]
    \label{thm:alpha}
    Let \(n = |M| = |\calU|\). Define \(\alpha = \frac{\log\pth{|\calV|-n}}{\log n}\) if \(|\calV|>n\). \(\zeta \geq c\log n + \omega(1)\) is a sufficient condition for \underline{exact alignment}, where the value of \(c\) is given below for the different algorithms and different sizes of \(|\calV|\):\\
    \begin{tabular}{|r|r|r|r|}
        \hline
        Size of \(\calV\)& \(|\calV|=n\) & \(n<|\calV| \leq n + o(n)\) & \(|\calV|\geq n+\Omega(n)\) \\
        \hline
        \hline
        Threshold. &\(\pth{1+\sqrt{2}}^2\) & \(\pth{1+\sqrt{2}}^2\) & \(\pth{1+\sqrt{1+\alpha}}^2\)\\
        \hline
        Max row &\(4\) & \(4\) & \(\pth{1+\sqrt{\alpha}}^2\)\\
        \hline
        Max likelihood & \(2\) & \(2(\alpha+1)\) & \(\pth{1+\sqrt{\alpha}}^2\)\\
        \hline
    \end{tabular}\\
    For \underline{almost-exact alignment}, \(\zeta\geq (\max\{\alpha,1\})\log n + \omega\pth{\sqrt{\log n}}\) is a sufficient condition for all three algorithms.

    \bcomment{
    Then, as \(n\to\infty\), the following are sufficient conditions for \underline{exact alignment}:
    \begin{compactitem}
        \item Threshold testing:\\
        \makebox[8cm]{\(\zeta \geq \pth{1+\sqrt{1+\max\{\alpha,1\}}}^2\log n + \omega(1)\)\hfill}
        \item Maximum row estimation:\\
        \makebox[8cm]{\(\zeta \geq \pth{1+\sqrt{\max\{\alpha,1\}}}^2\log n + \omega(1)\)\hfill}
        \item Maximum likelihood estimation:\\
        \makebox[6cm]{\(\zeta\geq 2(\alpha + 1)\log n + \omega(1)\)\hfill} if \,\(0\leq \alpha \leq 1\)\\
        \makebox[6cm]{\(\zeta \geq \pth{1+\sqrt{\alpha}}^2\log n + \omega(1)\)\hfill} if \,\(1\leq \alpha\)
    \end{compactitem}
    For almost-exact alignment, the condition below holds for all three aforementioned algorithms.
    \begin{compactitem}
        \item \makebox[8cm]{\(I_{XY}\geq (\max\{\alpha,1\})\log n + \omega\pth{\sqrt{\log n}}\)\hfill}
    \end{compactitem}
    }
\end{theorem}
Note that for threshold testing and  maximum row estimation, the exact-alignment thresholds does not change with the size of \(|\calV|\)  as long as \(|\calV|\) is on the order of \(n\). For maximum likelihood estimation, the exact-alignment threshold increases linearly with \(\alpha\) in this regime.
For \(|\calV|\geq \Omega(n)\), the exact-alignment thresholds for all three algorithms increase quadratically with \(\sqrt{\alpha}\). Furthermore, the thresholds for maximum row estimation and maximum likelihood estimation coincide.
The boundaries for exact and almost-exact alignment for the algorithms are illustrated in \hyperref[fig:threshold]{Fig. \ref*{fig:threshold}}. Boundaries for maximum likelihood and maximum row algorithms completely overlap for \(I_{XY}/\log n \geq 4\).

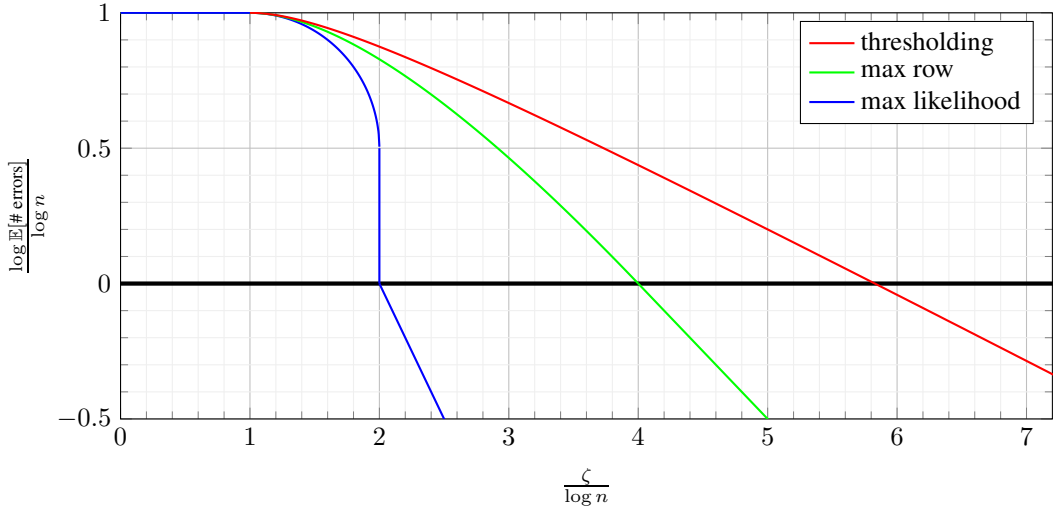
\begin{figure}
    \centering
    \begin{tikzpicture}
        \begin{axis}[
            xmin = 0, xmax = 7.2,
            ymin = -0.5, ymax = 1,
            xtick distance = 1,
            ytick distance = 0.5,
            grid = both,
            minor tick num = 4,
            major grid style = {lightgray},
            minor grid style = {lightgray!25},
            width = \textwidth,
            height = 0.5\textwidth,
            xlabel = {\(\frac{\zeta}{\log n}\)},
            ylabel = {\(\frac{\log\bbE[\textrm{\# errors}]}{\log n}\)},
            legend cell align = {left},
        ]
            
            \addplot[
                domain = 0:1,
                samples = 100,
                smooth,
                thick,
                red,
            ]{1};
            
            \addplot[
                domain = 0:1,
                samples = 100,
                smooth,
                thick,
                green,
            ]{1};
            
            \addplot[
                domain = 0:1,
                samples = 100,
                smooth,
                thick,
                blue,
            ]{1};
            
            
            \addplot[
                domain = 0:8,
                samples = 100,
                smooth,
                ultra thick,
                black,
            ]{0};
            
            \addplot[
                domain = 1:2,
                samples = 200,
                smooth,
                thick,
                blue,
            ]{sqrt(0.25-(x/2-0.5)^2)+0.5};
            
            \addplot[
                domain = 2:8,
                samples = 200,
                smooth,
                thick,
                blue,
            ]{2-x};

            \addplot [mark=none, thick, blue] coordinates {(2, 0) (2, 0.5)};
            
            \addplot[
                domain = 1:4,
                samples = 200,
                smooth,
                thick,
                green,
            ]{2*sqrt(x)-x};
            
            \addplot[
                domain = 4:8,
                samples = 200,
                smooth,
                thick,
                green,
            ]{2-x/2};
            
            \addplot[
                domain = 1:8,
                samples = 200,
                smooth,
                thick,
                red,
            ]{1-(1/x)*(x/2 - 0.5)^2};
            
            \legend{thresholding, max row, max likelihood}
        \end{axis}
    \end{tikzpicture}
    \caption{Comparison of the boundaries for the achievability regions when \(|\calV|=|\calU|=|M|=n\). The x-axis is to signal strength and the y-axis is the order of magnitude of the number of the expected number of mismatched users. Achievability regions are areas above/right of curves.}
    \label{fig:expectation0}
\end{figure}

\begin{theorem}[Expected number of errors in balanced case]
    \label{thm:beta}
    Let \(n = |M| = |\calU| = |\calV|\). The following sufficient conditions guarantee no more than \(n^{1-\beta+o(1)}\)\footnote{Sufficient conditions of linear and vertical form for maximum likelihood estimation both achieve an error bound of \(n^{1-\beta}\). For planted matching, the parabolic and elliptic ones, as well as the linear one for maximum row estimation achieve \(2n^{1-\beta}\). For database alignment, the parabolic and elliptic ones achieve \(2n^{1-\beta+o(1)}\) asymptotically as \(n\to\infty\), while the linear one for maximum row estimation achieves \(2n^{1-\beta}\).} errors in expectation. (The bounds that require \hyperref[cond:highDimensional]{Condition \ref*{cond:highDimensional}} for database alignment are specified.)
    
    \begin{tabular}{|l r c c|}
        \hline
        Sufficient cond. & Range of \(\beta\) & Form of boundary & Req. \hyperref[cond:highDimensional]{Cond. \ref*{cond:highDimensional}}\\
        \hline
        \hline
        \multicolumn{4}{|l|}{Threshold testing}\\
        \(\zeta \geq \pth{\sqrt{\beta}+\sqrt{1+\beta}}^2\log n\)  & \(0<\beta\phantom{\leq1/2.}\) & parabolic & no\\
        \hline
        \multicolumn{4}{|l|}{Maximum row estimation}\\
        \(\zeta \geq \pth{1+\sqrt{\beta}}^2\log n\) & \(0<\beta\leq 1\phantom{/2}\) & parabolic & yes\\
        \(\zeta \geq 2(1+\beta)\log n\) & \(1<\beta\phantom{\leq1/2.}\) & linear & no\\
        \hline
        \multicolumn{4}{|l|}{Maximum likelihood estimation}\\
        \(\zeta \geq \pth{1+2\sqrt{\beta(1-\beta)}}\log n\) & \(0<\beta\leq 1/2\) & elliptic & yes\\
        \(\zeta \geq 2\log n + 2\log\pth{\frac{\sqrt{5}-1}{2}}\) & \(1/2<\beta\leq 1\phantom{/2}\) & vertical & no\\
        \(\zeta \geq (1+\beta)\log n\) & \(1<\beta\phantom{\leq1/2.}\) & linear & no\\
        \hline
    \end{tabular}

    \bcomment{
    \begin{compactitem}
        \item Threshold testing:\\
        \makebox[6cm]{\(I_{XY} \geq \pth{\sqrt{\beta}+\sqrt{1+\beta}}^2\log n\)\hfill}
        \item Maximum row estimation:\\
        \makebox[6cm]{\(I_{XY} \geq \pth{1+\sqrt{\beta}}^2\log n\)\hfill} if\, \(0<\beta\leq1\)\\
        \makebox[6cm]{\(I_{XY} \geq 2(1+\beta)\log n\)\hfill} if\, \(1\leq\beta\)
        \item Maximum likelihood estimation:\\
        \makebox[6cm]{\(I_{XY} \geq \pth{1+2\sqrt{\beta(1-\beta)}}\log n\)\hfill} if\, \(0<\beta\leq1/2\)\\
        \makebox[6cm]{\(I_{XY} \geq 2\log n + \omega(1)\)\hfill} if\, \(1/2\leq\beta\leq1\)\\
        \makebox[6cm]{\(I_{XY}\geq (1+\beta)\log n\)\hfill} if\, \(1\leq\beta\)
    \end{compactitem}
    }
\end{theorem}

The boundaries for the achievability regions of the algorithms are illustrated in \hyperref[fig:expectation0]{Fig. \ref*{fig:expectation0}}. 

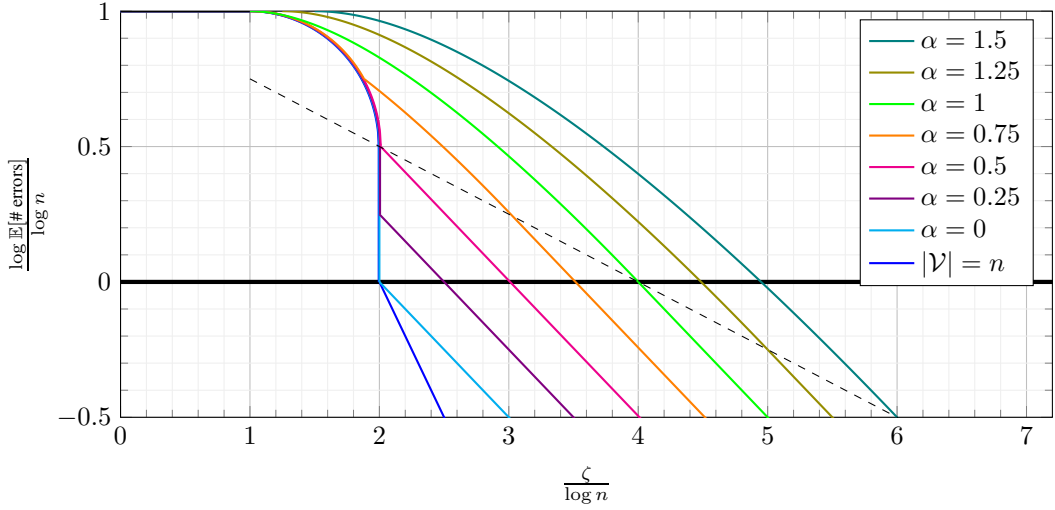
\begin{figure}
    \centering
    \begin{tikzpicture}
        \begin{axis}[
            xmin = 0, xmax = 7.2,
            ymin = -0.5, ymax = 1,
            xtick distance = 1,
            ytick distance = 0.5,
            grid = both,
            minor tick num = 4,
            major grid style = {lightgray},
            minor grid style = {lightgray!25},
            width = \textwidth,
            height = 0.5\textwidth,
            xlabel = {\(\frac{\zeta}{\log n}\)},
            ylabel = {\(\frac{\log\bbE[\textrm{\# errors}]}{\log n}\)},
            legend cell align = {left},
        ]

            \addplot[
                domain = 0:1,
                samples = 100,
                smooth,
                thick,
                teal,
            ]{1};

            \addplot[
                domain = 0:1,
                samples = 100,
                smooth,
                thick,
                olive,
            ]{1};
            
            \addplot[
                domain = 0:1,
                samples = 100,
                smooth,
                thick,
                green,
            ]{1};
            
            \addplot[
                domain = 0:1,
                samples = 100,
                smooth,
                thick,
                orange,
            ]{1};

            \addplot[
                domain = 0:1,
                samples = 100,
                smooth,
                thick,
                magenta,
            ]{1};

            \addplot[
                domain = 0:1,
                samples = 100,
                smooth,
                thick,
                violet,
            ]{1};

            \addplot[
                domain = 0:1,
                samples = 100,
                smooth,
                thick,
                cyan,
            ]{1};
            
            \addplot[
                domain = 0:1,
                samples = 100,
                smooth,
                thick,
                blue,
            ]{1};

            \addplot[
                domain = 0:1,
                samples = 100,
                smooth,
                thick,
                black,
            ]{1};
            
            
            \addplot[
                domain = 0:8,
                samples = 100,
                smooth,
                ultra thick,
                black,
            ]{0};

            \addplot[
                domain = 1:2-0.005,
                samples = 200,
                smooth,
                thick,
                blue,
            ]{sqrt(0.25-((x+0.005)/2-0.5)^2)+0.5};

            \addplot[
                domain = 2:8,
                samples = 200,
                smooth,
                thick,
                blue,
            ]{2-x};

            \addplot [mark=none, thick, blue] coordinates {(2-0.005, 0) (2-0.005, 0.51)};

            \addplot[
                domain = 1:2,
                samples = 200,
                smooth,
                thick,
                cyan,
            ]{sqrt(0.25-(x/2-0.5)^2)+0.5};

            \addplot[
                domain = 2:8,
                samples = 200,
                smooth,
                thick,
                cyan,
            ]{1-x/2};

            \addplot [mark=none, thick, cyan] coordinates {(2, 0) (2, 0.51)};

            \addplot[
                domain = 1:2+0.005,
                samples = 200,
                smooth,
                thick,
                violet,
            ]{sqrt(0.25-((x-0.005)/2-0.5)^2)+0.5};

            \addplot[
                domain = 2:8,
                samples = 200,
                smooth,
                thick,
                violet,
            ]{1.25-x/2};

            \addplot [mark=none, thick, violet] coordinates {(2+0.005, 0.25) (2+0.005, 0.51)};

            \addplot[
                domain = 1:2+0.01,
                samples = 200,
                smooth,
                thick,
                magenta,
            ]{sqrt(0.25-((x-0.01)/2-0.5)^2)+0.5};

            \addplot[
                domain = 2+0.01:8,
                samples = 200,
                smooth,
                thick,
                magenta,
            ]{1.5-(x-0.01)/2};

            \addplot [mark=none, thick, magenta] coordinates {(2+0.01, 0.5) (2+0.01, 0.51)};

            \addplot[
                domain = 1:1.866+0.015,
                samples = 200,
                smooth,
                thick,
                orange,
            ]{sqrt(0.25-((x-0.015)/2-0.5)^2)+0.5};

            \addplot[
                domain = 1.866+0.015:3,
                samples = 200,
                smooth,
                thick,
                orange,
            ]{1-(sqrt(x-0.015)-sqrt(0.75))^2};

            \addplot[
                domain = 3:8,
                samples = 200,
                smooth,
                thick,
                orange,
            ]{1.75-(x-0.015)/2};

            \addplot[
                domain = 1:4,
                samples = 200,
                smooth,
                thick,
                green,
            ]{1-(sqrt(x)-1)^2};
            
            \addplot[
                domain = 4:8,
                samples = 200,
                smooth,
                thick,
                green,
            ]{2-x/2};

            \addplot[
                domain = 1.25:5,
                samples = 200,
                smooth,
                thick,
                olive,
            ]{1-(sqrt(x)-sqrt(1.25))^2};
            
            \addplot[
                domain = 5:8,
                samples = 200,
                smooth,
                thick,
                olive,
            ]{2.25-x/2};

            \addplot[
                domain = 1.5:6,
                samples = 200,
                smooth,
                thick,
                teal,
            ]{1-(sqrt(x)-sqrt(1.5))^2};
            
            \addplot[
                domain = 6:8,
                samples = 200,
                smooth,
                thick,
                teal,
            ]{2.5-x/2};

            \addplot[
                domain = 1:8,
                samples = 200,
                thin,
                smooth,
                black,
                dashed
            ]{1-x/4};
            
            \legend{\(\alpha=1.5\), \(\alpha=1.25\), \(\alpha=1\), \(\alpha=0.75\), \(\alpha=0.5\), \(\alpha=0.25\), \(\alpha=0\), \(|\calV|=n\)}
        \end{axis}
    \end{tikzpicture}
    \caption{Comparison of boundaries for achievability regions of \underline{maximum likelihood algorithm} for various values of \(\alpha\). The x-axis is to signal strength and y-axis is the order of magnitude of expected number of mismatched users. Achievability regions are areas above/right of curves. Boundaries are parabolic/elliptic in the region above the dashed black line and consist of linear segments below it.}
    \label{fig:expectationMLE}
\end{figure}

\begin{theorem}[Expected number of errors in unbalanced case]
    \label{thm:alphabeta}
    Let \(n = |M| = |\calU|\) and \(|\calV|>n\). Define \(\alpha = \frac{\log\pth{|\calV|-n}}{\log n}\). The following sufficient conditions guarantee no more than \(n^{1-\beta+o(1)}\)\footnote{The constants are the same as those given in footnote for \hyperref[thm:beta]{Theorem \ref*{thm:beta}}.} errors in expectation. (The bounds that require \hyperref[cond:highDimensional]{Cond. \ref*{cond:highDimensional}} for database alignment are specified.)
    
    \begin{tabular}{|l c c c|}
        \hline
        Sufficient cond. & Range of \(\beta\) & Boundary & Req. \hyperref[cond:highDimensional]{Cond. \ref*{cond:highDimensional}}\\
        \hline
        \hline
        \multicolumn{4}{|l|}{Threshold testing}\\
        \(\zeta \geq \pth{\sqrt{\max\{\alpha,1\}+\beta}+\sqrt{\beta}}^2\log n\)\hspace{-10pt}  & \(0<\beta\) & parabolic & no\\
        \hline
        \multicolumn{4}{|l|}{Maximum row estimation}\\
        \(\zeta \geq \pth{\sqrt{\max\{\alpha,1\}}+\sqrt{\beta}}^2\log n\) &\(0<\beta\leq\max\{\alpha,1\}\) & parabolic & yes\\
        \(\zeta \geq 2\pth{\max\{\alpha,1\}+\beta}\log n\) & \(\max\{\alpha,1\}<\beta\) & linear & no\\
        \hline
        \multicolumn{4}{|l|}{Maximum likelihood estimation}\\
        \(\zeta \geq \pth{1+2\sqrt{\beta(1-\beta)}}\log n\) & \(0<\beta\leq\min\{1-\alpha,1/2\}\) & elliptic & yes\\
        \(\zeta \geq \pth{\sqrt{\alpha}+\sqrt{\beta}}^2\log n\) & \(1-\alpha<\beta\leq\alpha\) & parabolic & yes\\
        \(\zeta \geq 2\log n + 2\log\pth{\frac{3+\sqrt{5}}{2}}\) & \(1/2<\beta\leq 1-\alpha\) & vertical & no\\
        \(\zeta \geq 2(\alpha+\beta)\log n\) & \(\max\{\alpha,1-\alpha\}<\beta\) & linear & no\\
        \hline
    \end{tabular}

    \bcomment{
    The following are sufficient conditions to have expected number of errors bounded by \(2 n^{1-\beta}\):
    \begin{compactitem}
        \item Threshold testing:\\
        \makebox[6cm]{\(I_{XY} \geq \pth{\sqrt{\max\{\alpha,1\}+\beta}+\sqrt{\beta}}^2\log n\)\hfill}
        \item Maximum row estimation:\\
        \makebox[6cm]{\(I_{XY} \geq \pth{\sqrt{\max\{\alpha,1\}}+\sqrt{\beta}}^2\log n\)\hfill} if\, \(0<\beta\leq\max\{\alpha,1\}\)\\
        \makebox[6cm]{\(I_{XY} \geq 2\pth{\max\{\alpha,1\}+\beta}\log n\)\hfill} if\, \(\max\{\alpha,1\}\leq\beta\)
        \item Maximum likelihood estimation:\\
        \makebox[6cm]{\(I_{XY} \geq \pth{1+2\sqrt{\beta(1-\beta)}}\log n\)\hfill} if\, \(0<\beta\leq\min\{1-\alpha,1/2\}\)\\
        \makebox[6cm]{\(I_{XY} \geq 2\log n + \omega(1)\)\hfill} if\, \(1/2\leq\beta\leq\alpha\)\\
        \makebox[6cm]{\(I_{XY} \geq \pth{\sqrt{\alpha}+\sqrt{\beta}}^2\log n\)\hfill} if\, \(1-\alpha\leq\beta\leq\alpha\)\\
        \makebox[6cm]{\(I_{XY}\geq 2(\alpha+\beta)\log n\)\hfill} if\, \(\alpha\leq\beta\)
    \end{compactitem}
    }
\end{theorem}

The boundaries for the achievability region of the maximum likelihood estimator for various values of \(\alpha\) are illustrated in \hyperref[fig:expectationMLE]{Fig. \ref*{fig:expectationMLE}}. Boundaries for \(\alpha=1.5\) and \(\alpha=1.25\) are tangent to \(y=1\) at \(x=1.5\) and \(x=1.25\) respectively, while all other boundaries are tangent to that line at \(x=1\). These match the almost-exact alignment threshold.

The boundaries for the achievability region of maximum likelihood/maximum row estimation, which coincide, and for thresholding at \(\alpha = 1.5\) are illustrated in \hyperref[fig:expectation1]{Fig. \ref*{fig:expectation1}}. Both boundaries are tangent to \(y=1\) at \(x=1.5\). This matches the almost-exact alignment threshold.

\begin{figure}
    \centering
    \begin{tikzpicture}
        \begin{axis}[
            xmin = 0, xmax = 7.2,
            ymin = -0.5, ymax = 1,
            xtick distance = 1,
            ytick distance = 0.5,
            grid = both,
            minor tick num = 4,
            major grid style = {lightgray},
            minor grid style = {lightgray!25},
            width = \textwidth,
            height = 0.5\textwidth,
            xlabel = {\(\frac{\zeta}{\log n}\)},
            ylabel = {\(\frac{\log\bbE[\textrm{\# errors}]}{\log n}\)},
            legend cell align = {left},
        ]
            
            \addplot[
                domain = 0:1,
                samples = 100,
                smooth,
                thick,
                red,
            ]{1};
            
            \addplot[
                domain = 0:1,
                samples = 100,
                smooth,
                thick,
                teal,
            ]{1};
            
            
            \addplot[
                domain = 0:8,
                samples = 100,
                smooth,
                ultra thick,
                black,
            ]{0};
            
            \addplot[
                domain = 1.5:8,
                samples = 200,
                smooth,
                thick,
                teal,
            ]{1-(sqrt(x)-sqrt(1.5))^2};
            
            \addplot[
                domain = 1.5:8,
                samples = 200,
                smooth,
                thick,
                red,
            ]{(-9+28*x-4*x^2)/(16*x)};
            
            \legend{thresholding, max row \& max likelihood}
        \end{axis}
    \end{tikzpicture}
    \caption{Comparison of the boundaries for the achievability regions when \(|\calU|=|M|=n\) and \(|\calV|=n^{1.5}\). The x-axis is signal strength and the y-axis is the order of magnitude of the number of the expected number of mismatched users. Achievability regions are areas above/right of curves.}
    \label{fig:expectation1}
\end{figure}
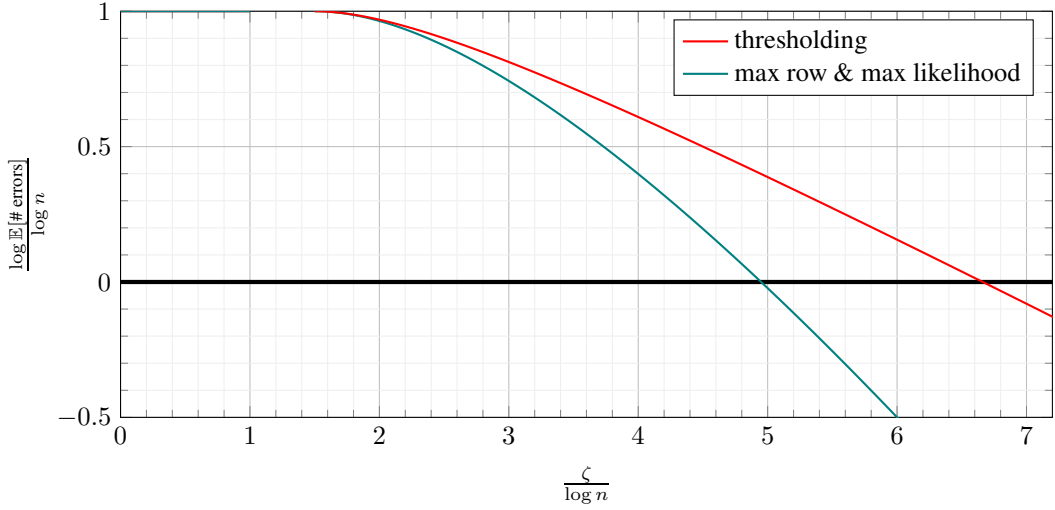

The proofs for Theorems \ref{thm:beta}, \ref{thm:alpha}, \ref{thm:alphabeta} are given in Appendix C.

\subsection{Converse results}
For the planted matching model, we have matching converse results in multiple regimes.
\begin{theorem}
    Let \(n = |M| = |\calU|\) and \(\alpha = \frac{\log\pth{|\calV|-n}}{\log n}\) if \(|\calV|>n\). Each of the following conditions guarantee that any estimator makes at least $\Omega(n^{1-\beta})$ errors with probability $1-o(1)$:
    \bcomment{
    \begin{tabular}{|r l c|}
        \hline
        Necessary cond. & Range of \(\beta\) & Boundary\\
        \hline
        \(\zeta \leq \pth{1+2\sqrt{\beta(1-\beta)}}\log n - \omega(\log \log n)\) & \(0<\beta\leq\min\{1-\alpha,1/2\}\) & elliptic\\
        \(\zeta \leq \pth{\sqrt{\alpha}+\sqrt{\beta}}^2 \log n - \omega(\log\log n)\) & \(1-\alpha<\beta\) & parabolic\\
        \(\zeta \leq 2 \log n - \omega(\log \log n) \) & \(1/2<\beta\) & vertical\\
        \hline
    \end{tabular}
    }
    \begin{tabular}{l c}
      \(\zeta \leq \pth{1+2\sqrt{\beta(1-\beta)}}\log n - \omega(\log \log n)\) & \(0 < \beta \leq \frac{1}{2}\)\\
      \(\zeta \leq \pth{\sqrt{\alpha}+\sqrt{\beta}}^2 \log n - \omega(\log\log n)\) & \(0 < \beta \leq 1\)\\
      \(\zeta \leq 2 \log n - \omega(\log \log n) \) & \(\frac{1}{2} \leq \beta \leq 1\).
    \end{tabular}
\end{theorem}

The proof is given in Appendix D.

The maximum likelihood estimator is also the maximum a posteriori estimator and thus is the optimal estimator for exact recovery.
It is not necessarily optimal for partial recover: it does not necessarily maximize the probability that it makes at most $n^{1-\beta}$ errors for $\beta < 1$.
However, this converse shows that the maximum likelihood estimator is asymptotically optimal for partial recovery: the conditions needed to ensure partial correctness match the converse conditions in the leading term.

\subsection{Interpretation and intuition behind results}

We present the intuition behind the various phase transitions of the achievability boundaries.

\subsubsection{Merging of boundary for maximum likelihood estimation and maximum row estimation}

As shown by \hyperref[thm:alphabeta]{Theorem \ref*{thm:alphabeta}} and \hyperref[fig:expectationMLE]{Fig. \ref*{fig:expectationMLE}}, the boundaries for the achievability regions of maximum likelihood estimation and maximum row estimation fully coincide if the number of unmatched users in \(\calV\) is on the order of \(n\) or greater. As shown by \hyperref[thm:alpha]{Theorem \ref*{thm:alpha}} and \hyperref[fig:threshold]{Fig. \ref*{fig:threshold}}, this is also true for the almost-exact alignment threshold.

Maximum row estimation corresponds to a relaxation of the maximum likelihood estimation by removing the constraint that every vertex in \(\calV\) can have at most one match in \(\calU\). That is, when looking for the true mapping matrix \(\Mmat\in\{0,1\}^{\calU\times\calV}\), maximum likelihood estimation still only accepts a single non-zero entry on each row, but ignores the number of entries on each column.

For \(\alpha<1\) bounded away from 1, the average non-zero entries in each column is \(\frac{n}{n+n^\alpha} \geq 1-o(1)\). So the column constraint of maximum likelihood estimation is tight for almost every column. Therefore, for such \(\alpha\), the column constraint is relevant and its removal results in the introduction of a significant increase in the expected number of errors. This creates a gap between the boundaries of maximum likelihood estimation and maximum row estimation.

On the other hand, for \(\alpha > 1\) bounded away from 1, the average non-zero entries in each column is \(\frac{n}{n+n^\alpha} \leq o(1)\). So the column constraint is loose for almost every column. Then, for such \(\alpha\), relaxing the column constraint results in no significant loss in performance and the boundaries for the two algorithms coincide.

\subsubsection{Transition from the elliptic to the quadratic boundary for maximum likelihood estimation}

By \hyperref[thm:alphabeta]{Theorem \ref*{thm:alphabeta}}, given \(\alpha>1/2\), there is a phase transition in the achievability boundary for maximum likelihood estimation as the bound on the expected number of errors goes beyond \(n^\alpha\). 

Recall that \(n=|M|=|\calU|\), and \(n^\alpha = |\calV|-n\). Consider a simplified model to generate an estimated mapping \(m\): In step 1, make an independent decision for every true pair whether to include it in \(m\) or not. Each pair is failed to be included with probability \(\frac{n^{1-\beta}}{n}\). In step 2, randomly assign each of the users in \(\calU\) that haven't been mapped in step 1, to a random user in \(\calV\) that also hasn't been mapped in step 1. In expectation, there are \(n^{1-\beta}\) and \(n^{1-\beta}+n^\alpha\) users that haven't been mapped in step 1 in \(\calU\) and \(\calV\) respectively.

Based on this process, every falsely-paired user in \(\calV\) will be mapped to about \(\frac{n^{1-\beta}}{n^{1-\beta}+n^\alpha} = \pth{1+n^{\alpha+\beta-1}}^{-1}\) users in expectation. If \(\alpha+\beta < 1\) bounded away from 1, then this value is \(1-o(1)\). So the restriction on having each user from \(\calV\) mapped to at most 1 user is tight for almost all users in \(\calV\). This implies that, at that point, this constraint is still able to contribute to the elimination of misalignments: An increase  in signal strength leads to users previously not paired in step 1 to be correctly paired, which in turn would trigger the mentioned constraint, leading to chain reactions that might fix even more misaligned pairs.

For \(\alpha+\beta >1\) bounded away from 1, the average number of users mapped to the falsely-paired users in \(\calV\) is \(o(1)\). In this regime, the constraint is loose for almost every falsely-paired user in \(\calV\), and the restriction does not help in eliminating misalignments.

The fact that this restriction stops being relevant around the point \(\beta = 1-\alpha\) demonstrates itself on the boundary as an immediate decrease in the absolute value of the slope. At this point, increasing mutual information does not decrease the error exponent as quickly, as the algorithm can no longer leverage the constraint on the number of pairs of users from \(\calV\).

Note that, for \(\alpha <1\) bounded away from 1, the gap between maximum likelihood estimation and maximum row estimation still persists whether \(\alpha+\beta < 1\) or \(\alpha+\beta >1\). This is due to the fact the correct pairing of the \(n-n^{1-\beta}\) users in step 1 does rely on the column constraint.

\subsubsection{Transition between quadratic boundary to linear boundary of maximum likelihood estimation and maximum row estimation}
\label{subsec:bad-pairs}
When a true pair has a sufficiently low score in \(\Gmat\) or \(\Wmat\), the expected number of errors involving that pair is larger than 1.
Call this a "bad true pair."
This phase transition between the parabolic boundary and the linear one corresponds to a change in the importance of the bad true pairs.
In the parabolic region, most errors involve a bad true pair.
In the linear region, most errors involve true pairs that are not bad.
See the discussion in Appendix E of the cycle-path decomposition of a pair of matchings for the precise interpretation of "most errors" in these statements.

\subsubsection{Halving of slope of linear boundary of maximum likelihood estimation}
\label{subsec:halving}

As shown by \hyperref[thm:alphabeta]{Theorem \ref*{thm:alphabeta}} and \hyperref[thm:alphabeta]{Theorem \ref*{thm:alphabeta}}, the slope of the linear boundary for maximum likelihood estimation is halved when going from \(|\calV|=n\) for \(|\calV|>n\). This is illustrated by the last two curves in \hyperref[fig:expectationMLE]{Fig. \ref*{fig:expectationMLE}}.

For both cases, the linear region boundary is only relevant for \(\beta\) large, i.e. when the error exponent is very small. In this regime, the misalignment of any pair is very rare. For \(|\calV|=n\), a misalignment requires at least two pairs of users. Specifically, we need some pairs \((u_0,v_0)\) and \((u_1,v_1)\) such that, \((u_0,v_1)\) and \((u_1,v_0)\) jointly make a better pairing. This requires the occurrence of two relatively low score pairs as well as the corresponding misalignment to have two relatively high score false pairs. As such, this error event requires many unlikely things to coincide, making its likelihood very small as we increase signal strength beyond the point \(\zeta=2\log n\).

On the other hand, when \(|\calV|\geq n+1\), there is at least one user, say \(v'\in\calV\), that already has no pair. Therefore a misalignment may consist of a single misaligned pair, e.g. \((u_0,v')\) instead of \((u_0,v_0)\). This requires the occurrence of a single relatively low score pair as well as a relatively high score misalignment of the first pair with an with unpaired user, of which there are \(|\calV|-n\). The inclusion of even a single extra therefore makes the unlikely misalignment event somewhat less exceptional.

\subsubsection{Vertical segment of boundary of maximum likelihood estimation}
This phase transition involves a shift in the structure of the typical error.
As explained in \hyperref[subsec:halving]{Subsection \ref*{subsec:halving}}, in the linear region, the dominant type of error is either a 2-cycle error (if \(|\mathcal{V}| = n\)) or a \((1,1)\)-path error (\(|\mathcal{V}| > n\)).
In the elliptic boundary region, much longer cycles and paths are dominant.
In the balanced case all errors come from cycles.
When \(\zeta = 2 \log n\), the expected contributions of each cycle length to the number of errors are equal.
Each of these contributions has a linear dependence on \(\zeta\) with a slope proportional to the cycle length.
Thus the contribution of cycles of length $\omega(1)$ produces the vertical boundary.
The top of the vertical boundary occurs due to an effect similar to that described in \hyperref[subsec:bad-pairs]{Subsection \ref*{subsec:bad-pairs}}.
In the elliptical boundary region, most long-cycle errors involve a bad true pair, so the number of bad true pairs controls the overall number of errors.

\subsubsection{Gap between maximum row estimation and threshold testing}
The difference between the maximum row estimators and the threshold testing estimator is constraint (b) in the linear program, which ensures that the estimated matrix has exactly one 1 in each row.
This constraint can be included because of our assumption that \(|M| = |\mathcal{U}|\), i.e. that every user is the first database has a match in the second database.
The gap between the performance of the maximum row and threshold testing estimators means that this constraint corrects most errors in the threshold testing estimator.
If we has a sufficiently large gap between \(|M|\) and \(|\mathcal{U}\), we would see the performance of these two estimators converge, just and the ML and maximum row estimators converge in performance when \(|\mathcal{V}|\) is sufficiently larger than \(|\mathcal{U}|\).

\section{Open questions and future work}
\TODO{
low dimensional features: we think low dim case is strictly easier. This is suggested by our results that don't require the high dim assumption.
Detailed structure of the zero-order phase transition
More complete understanding of conditions for zero order transition
}
We believe that Gaussian database alignment becomes information-theoretically easier as the feature dimensionality decreases.
In other words, as the same amount of mutual information is concentrated in a smaller number of features, it is more valuable for alignment.
Several of our achievability results do not require Condition 1 and thus provide evidence for this.
Our new converse works only for the planted matching problem, i.e. in the infinite dimension limit of the database alignment problem.
Showing this monotonic dependence on feature dimensionality is an open problem.

Other interesting questions involve the finer structure of the phase transition at the exact recovery threshold. If $\zeta = 2 \log n + c \log \log n$, how many errors does the ML estimator make?
Finally, what are the precise conditions that cause a planted matching or database alignment problem to have a discontinuity in the number of errors at this threshold?

\bibliographystyle{abbrv}

\newpage

\appendix

\section{Algorithms}
\label{sec:algo}

First we introduce the information density matrix \(\Gmat\) in \hyperref[subsec:algo1]{Subsection \ref*{subsec:algo1}}, which justifies the implementations of the algorithms that we present. Then in \hyperref[subsec:algo2]{Subsection \ref*{subsec:algo2}} we formulate these algorithms as linear programs with a clear hierarchy in their constraints. Finally \hyperref[subsec:algo3]{Subsection \ref*{subsec:algo3}} presents an analysis of the computational complexities of each algorithm.

\subsection{Information density matrix for database alignment}
\label{subsec:algo1}

Let \(\Avec,\Bvec\) denote correlated Gaussian databases as described in \hyperref[subsec:modelDatabase]{Subsection \ref*{subsec:modelDatabase}}. Let \(f_{XY}\), \(f_X\) and \(f_Y\) denote the joint and marginal distributions for correlated features in \(\Avec\) and \(\Bvec\). Given any partial mapping \(m\), let \(\calU_m\subseteq\calU\) and \(\calV_m\subseteq\calV\) denote the sets of users that have a mapping under \(m\) and \(\calW_m\subseteq\calU_m\times\calV_m\) denote the set of pairs mapped by \(m\). Then the log-likelihood of observing \(\Avec\) and \(\Bvec\) under the assumption that \(M=m\) is given by
\begin{align}
    \label{eq:algo1}
    \sum_{(u,v)\in \calW_m}\log f_{XY}(\Avec(u),\Bvec(v)) + \sum_{u\in\calU\setminus\calU_m}\log f_X(\Avec(u)) + \sum_{v\in\calV\setminus\calV_m}\log f_Y(\Bvec(v)).
\end{align}

Let \(\Gmat\in\bbR^{\calU\times\calV}\) denote the information density matrix such that \(G_{u,v} = \log\frac{f_{XY}(\Avec(u),\Bvec(v))}{f_X(\Avec(u))f_Y(\Bvec(v))}\). In other words, \(G_{u,v}\) is the log-likelihood ratio between hypotheses \(u\mapped{M}v\) vs. \(u\notmapped{M}v\).

Let \(\mmat\in\{0,1\}^{\calU\times\calV}\) denote a matrix encoding of the mapping \(m\) such that \(m_{u,v} = 1\iff u\mapped{m}v\). Then, the inner product \(\ip{\Gmat,\mmat}\) equals \(\sum_{(u,v)\in \calW_m}\log f_{XY}(\Avec(u),\Bvec(v)) - \sum_{u\in\calU_m}\log f_X(\Avec(u)) - \sum_{v\in\calV_m}\log f_Y(\Bvec(v))\). It then follows that
\begin{align}
    \label{eq:algo2}
    \ip{\Gmat,\mmat} + \sum_{u\in\calU} \log f_X(\Avec(u)) + \sum_{v\in\calV}\log f_Y(\Bvec(v))
\end{align}
exactly equals the expression given in (\ref{eq:algo1}). The terms following \(\ip{\Gmat,\mmat}\) in (\ref{eq:algo2}) do not depend on \(m\). So, the choice of \(m\) that maximizes \(\ip{\Gmat,\mmat}\) is the same as the maximizer for log-likelihood, as given in (\ref{eq:algo2}). Then, Then \(\Gmat\) contains all information relevant to identifying the underlying mapping \(M\).

\subsection{Log-likelihood for planted matching}
\label{subsec:algo11}

The log-likelihood of observation \(\Wmat\) under the planted matching modeled described in \hyperref[subsec:modelPlanted]{Subsection \ref*{subsec:modelPlanted}}. The pdf of \(\Wmat\) given \(M\) is given by the expression
\begin{align*}
    \frac{1}{(2\pi)^{|\calU\times\calV|/2}}\exp\pth{-\frac{1}{2}\ip{\Wmat - \mu\Mmat,\Wmat-\mu\Mmat}}.
\end{align*}
Then, the log-likelihood is a constant factor away from \(-\frac{1}{2}\ip{\Wmat - \mu\Mmat,\Wmat-\mu\Mmat} = \ip{\mu\Wmat,\mmat}-\frac{1}{2}||\Wmat||_F^2 -\frac{\mu^2}{2}||M||_F^2\) where \(||.||_F\) denotes the Frobenius norm. \(||\Wmat||_F^2\) does not depend on \(\Mmat\) and \(||M||_F^2\) is equal to the size of \(M\). Then, maximizing over mappings with fixed size, \(\mmat\) that maximizes \(\ip{\mu\Wmat,\mmat}\) also maximizes the likelihood of \(\Wmat\) given \(M=m\).

Alternatively, we can optimize over \(\ip{\Wmat_G,\mmat}\) where \(\Wmat_G = \mu\Wmat -\mu^2/2\) since \(\ip{\Wmat_G,\mmat}\) is a constant factor of \(|\calU\times\calV|\mu^2/2\) away from \(\ip{\mu\Wmat,\mmat}\).

\subsection{Algorithms}
\label{subsec:algo2}

\textbf{\underline{Maximum likelihood estimation}}

For database alignment, \(\ip{\Gmat,\mmat}\) is a constant factor away from the log-likelihood of mapping \(m\), as shown in \hyperref[subsec:algo1]{Subsection \ref*{subsec:algo1}}. For planted mtaching, given \(\Wmat_G \define \mu\Wmat -\mu^2/2\), \(\ip{\Wmat_G,\mmat}\) is a constant factor away from the log-likelihood of mapping \(m\), as shown in \hyperref[subsec:algo11]{Subsection \ref*{subsec:algo11}}. So, optimizing \(\ip{\Gmat,\mmat}\) or \(\ip{\Wmat_G,\mmat}\) over mapping matrices \(\mmat\) gives us the maximum likelihood estimate for the two problems.

This is an instance of the linear assignment problem, and therefore can be solved by the Hungarian algorithm in polynomial time (\cite{ramshaw2012minimum}).

Alternatively, this can be expressed as a linear problem as given in \hyperref[subsec:model-algo]{Subsection \ref*{subsec:model-algo}}:    \begin{align*}
    \textrm{maximize} \ip{\Gmat-\tau,\mmat} \qquad & \textrm{(a)}\,\, \sum_{u\in\calU} m_{u,v} \leq 1, \,\,\forall v\in\calV\\
    & \textrm{(b)}\,\, \sum_{v\in\calV} m_{u,v} = 1, \,\,\forall u\in\calU\\
    & \textrm{(c)}\,\, m_{u,v} \in [0,1], \,\,\,\,\, \forall (u,v)\in\calU\times\calV
\end{align*}
The value of \(\tau\in\bbR\) is irrelevant under constraint (b) in finding the maximizer \(\mmat\): Given constraint (b), \(\mmat\) has fixed sum of entries, and therefore \(\ip{\Gmat-\tau,\mmat} = \ip{\Gmat,\mmat}-\tau\sum m_{u,v} = \ip{\Gmat,\mmat} - \tau|\calU|\), so the objective function is shifted by a constant term that depends on \(\tau\) but not on \(\mmat\).

\bcomment{
Let \(\mvec\in\bbR^{(\calU\times\calV)}\) denote the vector representation of the matrix \(\mmat\in\bbR^{\calU\times\calV}\). Then, the constraint matrix of the linear program above can be written as
\begin{align*}
    \Amat_\calV\mvec &\leq \onevec\\
    \Amat_\calU\mvec &= \onevec\\
    \Imat\mvec &\leq \onevec\\
    \Imat\mvec &\geq \zerovec
\end{align*}
where \(\Amat_\calV\in\{0,1\}^{\calV\times(\calU\times\calV)}\) and  \(\Amat_\calU\in\{0,1\}^{\calU\times(\calU\times\calV)}\) are matrices with exactly one non-zero entry in each row. Then \(\crocVec{\Amat_\calV}{\Amat_\calU}\) is the incidence matrix of a bipartite graph with \(\calU\) and \(\calV\) the two bipartite sets and \((\calU\cup\calV)\) the edge set. (Specifically, vertex \(w\) is incident to edge \((u,v)\) if and only if \(w=u\) or \(w=v\).) It is well known that the incidence matrix of a bipartite graph is totally unimodular. Furthermore, \(\crocVec{\Amat}{\Imat}\) is totally unimodular if \(\Amat\) is totally unimodular. It then follows that the constraint matrix for this linear program is totally unimodular, and therefore it is guaranteed to have an integral solution. (The existence of an integral solution is trivial since the linear program is equivalent to the linear assignment problem.)
}

\textbf{\underline{Maximum row estimation}}

The objective function \(\ip{\Gmat-\tau,\mmat}\) can be broken down into its row-wise sums \(\sum_u (\Gmat-\tau)_{u,*}^\top \mmat_{u,*}\) where \((.)_{u,*}\) denotes the row of the matrix corresponding to user \(u\in\calU\). Then, removing (a), which is the only constraint that takes into account multiple rows at once, breaks down the optimization problem into the sum of row-wise optimization problems, where each row of \(\mmat\) can be optimized independently. That is, alignment is performed independently over each row. This gives us \textbf{maximum row estimation}.

Given any \(u\), the algorithm looks at the log likelihood scores of mappings \((u,v)\) for each \(v\in\calV\) and picks \(v\) that has the highest likelihood. Users in \(\calV\) may be mapped to multiple users if they happen to be the best match for multiple users in \(\calU\). The mapping of each user \(u\in\calU\) under this relaxation would be the maximum likelihood estimate for \(u\) if we were blind to the existence of other users in \(\calU\).

\textbf{\underline{Threshold testing}}

The objective function \(\ip{\Gmat-\tau,\mmat}\) can be broken down into entry-wise sums \(\sum_{u,v} (G_{u,v}-\tau)m_{u,v}\). Removing conditions (a) and (b) breaks down any dependence between entries of \(\mmat\) allows us to optimize all entries in the matrix independently from each other. Then we are left with an algorithm that independently considers each pair of users and makes a decision on whether they are true pairs or not. Specifically, \((u,v)\) is estimated to be a true match if and only if \(G_{u,v}-\tau\) is positive. Since \(G_{u,v}\) is defined to be the log-likelihood ratio between hypotheses \(u\mapped{M}v\) vs. \(u\notmapped{M}v\), the decision rule \(G_{u,v}-\tau\geq 0\) is equivalent to the likelihood-ratio test at some significance level determined by the log threshold \(\tau\). We refer to this relaxation as \textbf{threshold testing}.

\subsection{Computational complexity}
\label{subsec:algo3}

Let \(d\) denote the maximum of the number of features per user in the two database, \(n_M = |\calU|\) the size of the true mapping, and \(n_\calU,n_\calV\) denote the sizes of the two user sets \(\calU,\calV\).

\underline{\textbf{Summary}}

Let \(d \leq \calO(n_M)\). The computational complexity of the algorithms are given as follows:
\begin{compactitem}
    \item Maximum likelihood estimation: \(\calO\pth{n_M\cdot n_\calU\cdot n_\calV}\) for entire set of users \(\calU\times\calV\)
    \item Maximum row estimation: \(\calO\pth{d\cdot n_\calU\cdot n_\calV}\) for entire set of users \(\calU\times\calV\),\\
    \phantom{Maximum row estimation: }\(\calO\pth{d^2\cdot n_\calV}\) for a given user \(u\in\calU\) against entire set \(\calV\).
    \item Threshold testing: \(\calO\pth{d\cdot n_\calU\cdot n_\calV}\) for entire set of users \(\calU\times\calV\),\\
    \phantom{Threshold testing: }\(\calO\pth{d^2}\) for a given pair of users \((u,v)\in\calU\times\calV\).
\end{compactitem}

\underline{\textbf{Computing \(\Gmat\) through the canonical form}}

Identifying the affine transformations described in \hyperref[sec:canonical]{Appendix \ref*{sec:canonical}} to transform features into canonical form takes a sequence of two Cholesky decompositions (\(\Sigmamat_a = \Lmat_a^\top\Lmat_a\), \(\Sigmamat_b=\Lmat_b^\top\Lmat_b\)), two matrix multiplications with inverted triangular matrices (\(\Lmat_a^{-1}\Sigmamat_{ab}(\Lmat_b^\top)^{-1}\)), one singular value decomposition (\(\Lmat_a^{-1}\Sigmamat_{ab}(\Lmat_b^\top)^{-1}=\Umat\Pmat\Vmat^\top\)) and two more matrix multiplications with inverted triangular matrices (\(\Umat^\top\Lmat_a^{-1}\) and \(\Vmat^\top\Lmat_b^{-1}\)). This can be done in \(\calO(d^3)\)-time.

Performing the affine transformation to transform features into canonical form as described in consists of one vector addition and one matrix-vector multiplication (\((\Umat^\top\Lmat_a^{-1})(\xvec-\muvec_a)\) or \((\Vmat^\top\Lmat_b^{-1})(\yvec-\muvec_b)\)). Then, transforming a single feature vector takes \(\calO(d^2)\)-time and transforming the entire database takes \(\calO(d^2\cdot n_\calU)\) and \(\calO(d^2\cdot n_\calV)\) for \(\calU\) and \(\calV\) respectively.

Given databases in canonical form, a single entry in \(\Gmat\) can be computed in \(\calO(d)\). This follows from the fact that, in canonical form, there is a one-to-one correspondence between feature entries from the two databases and therefore \(\log\frac{f_{XY}(\Avec(u),\Bvec(v))}{f_X(\Avec(u))f_Y(\Bvec(v))} = \sum_i \log\frac{f_{X_iY_i}(A_i(u),B_i(v))}{f_{X_i}(A_i(u))f_{Y_i}(B_i(v))}\), where the summation is over indices \(i\in\calD\). Then it takes \(\calO(d\cdot n_\calU \cdot n_\calV)\) to compute the entire matrix \(\Gmat\) based on features in canonical form.

Therefore, computing values of \(\Gmat\) from the databases has complexity
\begin{compactitem}
    \item \(\calO\bpth{d^3+d^2\cdot(n_\calU+n_\calV)+d\cdot n_\calU\cdot n_\calV }\) for the entire matrix \(\Gmat\),
    \item \(\calO\bpth{d^3+d^2\cdot n_\calV}\) for a row of \(\Gmat\) and
    \item \(\calO\bpth{d^3}\) for a single entry of \(\Gmat\).
\end{compactitem}

\underline{\textbf{Computing \(\Gmat\) without going through the canonical form}}

Given \(\Gmat\) in raw form (i.e. not necessarily canonical form), finding the likelihood requires calculating \(\det(\Sigmamat)\), \(\det(\Sigmamat_a)\) and \(\det(\Sigmamat_b)\), which takes \(\calO(d^3)\)-time, as well as \([\xvec^\top,\yvec^\top] \Sigmamat^{-1}[\xvec^\top,\yvec^\top]^\top\), \(\xvec^\top\Sigmamat_a^{-1}\xvec\) and \(\yvec^\top\Sigmamat_b^{-1}\yvec\) for each feature pair, which takes \(\calO(d^2)\) time per feature pair. Then it takes \(\calO(d^3+d^2\cdot n_\calU \cdot n_\calV)\) to compute the entire matrix \(\Gmat\) based on features in raw form. This is less efficient than doing the calculation through the canonical form which takes \(\calO\bpth{d^3+d^2\cdot(n_\calU+n_\calV)}\) to obtain features in canonical form and \(\calO(d\cdot n_\calU \cdot n_\calV)\) to get \(\Gmat\) based on features in canonical form.

\underline{\textbf{Maximum likelihood estimation}}

The (unbalanced) linear assignment problem which can be solved by the Hungarian algorithm in \(\calO(n_M\cdot n_\calU \cdot n_\calV)\) (\cite{ramshaw2012minimum}). Then, the total complexity of maximum likelihood estimation for database alignment, including the computation of \(\Gmat\), is \(\calO\bpth{d^3+d^2\cdot(n_\calU+n_\calV)+d\cdot n_\calU\cdot n_\calV + n_M\cdot n_\calU \cdot n_\calV}\). The complexity for planted matching is \(\calO(n_M\cdot n_\calU \cdot n_\calV)\).

\underline{\textbf{Maximum row estimation}}

For database alignment, given the corresponding row of \(\Gmat\), identifying the match of a user in \(\calU\) takes \(\calO(n_\calV)\)-time. Then the total complexity to align a single user \(u\in\calU\), including the complexity of calculating row \((\Gmat)_{u,*}\), is \(\calO(d^3+d^2 n_\calV)\). Consequently, aligning the entire set \(\calU\) takes \(\calO\bpth{d^3+d^2\cdot(n_\calU+n_\calV)+d\cdot n_\calU\cdot n_\calV}\)-time.

For planted matching, identifying the match of a user in \(\calU\) takes \(\calO(n_\calV)\)-time, while aligning the entire set takes \(\calO(n_\calU\cdot n_\calV)\)-time.

\underline{\textbf{Threshold testing algorithm}}

For database alignment, given the corresponding entry in \(\Gmat\), performing threshold testing over a pair of users takes \(\calO(1)\)-time. Then the total complexity to test a single user pair \((u,v)\), including the complexity of calculating \(G_{u,v}\), is \(\calO(d^3)\). Then performing the test over all pairs would take \(\calO\bpth{d^3+d^2\cdot(n_\calU+n_\calV)+d\cdot n_\calU\cdot n_\calV }\)-time.

For planted matching, performing threshold testing over a pair of users takes \(\calO(1)\)-time, while performing the test over all pairs would take \(\calO(n_\calU\cdot n_\calV)\)-time.

\section{Canonical form of the correlation statistics}
\label{sec:canonical}

For simplicity of computation as well as analysis, we would like the correlated feature indices in \(\Avec(u)\in\bbR^{\calD_a}\) and \(\Bvec(v)\in\bbR^{\calD_b}\) to have a one-to-one correspondance. Specifically, we would like the index sets \(\calD_a\) and \(\calD_b\) to be identical, and for features across databases to be correlated only if they have the same index. So, given true pair \(u\mapped{M}v\), the features \(A_i(u)\) and \(B_j(v)\) are correlated if and only if \(i=j\).

Given the correlation statistics \(\muvec\) and \(\Sigmamat\), it is possible to perform affine transformations on features to guarantee this type of correspondance between correlated feature vectors. We say a pair of databases with statistics of this desired form is in canonical form.

\subsection{Existence and construction of canonical transformation}

The generality of the canonical form is stated in the following lemma while the construction of the transformation that gives features in canonical form is described in the proof of the lemma.
\begin{lemma}[Existence of the canonical form]
    \label{lemma:canonicalVariance}
    Let \(\Xvec\) taking values in \(\bbR^{\calD_a}\) and \(\Yvec\) taking values in \(\bbR^{\calD_b}\) be a pair of correlated Gaussian vectors. If the mean and joint variance is known, one can define a pair of affine transformations \(\tvec_a:\bbR^{\calD_a}\to\bbR^{\calD}\) and \(\tvec_b:\bbR^{\calD_b}\to\bbR^{\calD}\) for some set \(\calD\) such that the mutual information between \(\tvec_a(\Xvec)\) and \(\tvec_b(\Yvec)\) equals the mutual information between \(\Xvec\) and \(\Yvec\) and \(\tvec_a(\Xvec)\in\bbR^{\calD}\) and \(\tvec_b(\Yvec)\in\bbR^{\calD}\) are a pair of correlated databases with mean \(\zerovec\) and joint variance \(\crocMat{\Imat}{\diag(\rhovec)}{\diag(\rhovec)}{\Imat}\) for some \(\rhovec\in(-1,1)^\calD\).
\end{lemma}
\begin{proof}
    Let \(\muvec = \crocVec{\muvec_a}{\muvec_b}\) and \(\Sigmamat = \crocMat{\Sigmamat_a}{\Sigmamat_{ab}}{\Sigmamat_{ab}^\top}{\Sigmamat_b}\) denote the mean and variance of \(\crocVec{\Xvec}{\Yvec}\).

    If \(\Sigmamat_a\) is not full rank, then there is some subset of \(\calD_a\) that can be discarded without loss of information. That is, we can throw away some indices of \(\Xvec\) to get a shorter vector which allows us to reconstruct the original vector \(\Xvec\). This follows from the fact that a multivariate Gaussian vector with covariance \(\Sigmamat\) can be written as a linear combination of \(\textrm{rank}(\Sigmamat)\) i.i.d. Gaussian normal random variables. Then, without loss of generality, assume \(\Sigmamat_a\) and \(\Sigmamat_b\) are full rank.
    
    \(\Sigmamat_a\) and \(\Sigmamat_b\) are covariance matrices, therefore they are positive semi-definite. It then follows that these matrices have Cholesky decompositions: \(\Sigmamat_a = \Lmat_a\Lmat_a^\top\) and \(\Sigmamat_b = \Lmat_b\Lmat_b^\top\) where \(\Lmat_a\) and \(\Lmat_b\) are lower triangular matrices with non-negative diagonal entries. By the assumption that \(\Sigmamat_a\) and \(\Sigmamat_b\)\ are full rank, it follows that the Cholesky decomposition gives triangular matrix with strictly positive entries. Then \(\Lmat_a\) and \(\Lmat_b\) are invertible.

    Let \(d_a \define |\calD_a|\) and \(d_b \define |\calD_b|\). Consider the singular value decomposition of \(\Lmat_a^{-1}\Sigmamat_{ab}\pth{\Lmat_b^\top}^{-1}\): \(\Umat\in\bbR^{\calD_a\times\{1,2,\cdots,d_a\}}\) and \(\Vmat\in\bbR^{\calD_b\times\{1,2,\cdots,d_b\}}\) orthonormal matrices and \(\Pmat\in\bbR^{d_a\times d_b}\) a diagonal matrix such that \(\Umat\Pmat\Vmat^\top = \Lmat_a^{-1}\Sigmamat_{ab}\pth{\Lmat_b^\top}^{-1}\).
    
    Let \(\fvec_a:\bbR^{\calD_a}\to\bbR^{d_a}\) and \(\fvec_b:\bbR^{\calD_b}\to\bbR^{d_b}\) such that
    \begin{align*}
        \fvec_a(\xvec) &= \Umat^\top \Lmat_a^{-1}\pth{\xvec - \muvec_a}\\
        \fvec_b(\yvec) &= \Vmat^\top \Lmat_b^{-1}\pth{\yvec - \muvec_b}.
    \end{align*}
    Note that both these transformations are invertible.
    
    We can verify that
    \begin{align*}
        \crocVec{\Umat^\top\Lmat_a^{-1}}{\Vmat^\top\Lmat_b^{-1}}\crocMat{\Sigmamat_a}{\Sigmamat_{ab}}{\Sigmamat_{ab}^\top}{\Sigmamat_b}\crocVec{\Umat^\top\Lmat_a^{-1}}{\Vmat^\top\Lmat_b^{-1}}^\top = \crocMat{\Imat}{\Pmat}{\Pmat}{\Imat}.
    \end{align*}
    Then
    \begin{align*}
        (\Xvec,\Yvec)\sim\calN(\muvec,\Sigmamat) \iff \pth{\fvec_a(\Xvec),\fvec_b(\Yvec)}\sim \calN\pth{\zerovec,\crocMat{\Imat}{\Pmat}{\Pmat^\top}{\Imat}}.
    \end{align*}
     
    By the invertibility of the transformations, there has been no loss of mutual information. If \(\Pmat\) has no empty rows or columns, then \(d_a=d_b\) must hold and we are done.
    
    If the \(i\)-th row of \(\Pmat\) is all-zero, then the \(i\)-th entry of \(\fvec_a(\Xvec)\) is completely independent from \(\Yvec\) or \(\fvec_b(\Yvec)\). It then follows that we can drop this entry without any loss of mutual information. The same argument applies for columns of \(\Pmat\) in relation to entries of \(\fvec_b(\Yvec)\).
    
    Let \(d\) denote the number of non-zero entries in the diagonal matrix \(\Pmat\), \(\calD\) be some arbitrary set of size \(d\). Let \(\Emat_a\in\{0,1\}^{\calD\times\{1,2,\cdots,d_a\}}\) such that rows of \(\Emat_a\) are the \(d\) standard basis vectors corresponding to the \(d\) non-empty rows of \(\Pmat\). Left multiplying a vector by \(\Emat_a\) gives us a shorter vector by `throwing away' all entries corresponding to empty rows of \(\Pmat\). Let \(\Emat_b\in\{0,1\}^{\calD\times\{1,2,\cdots,d_b\}}\) be a matrix of the same kind that `throws away' entries corresponding to empty columns of \(\Pmat\). Then \(\Emat_a \Pmat \Emat_b^\top\) is a diagonal matrix with no zeros on the diagonal. We use \(\rhovec\in\bbR^{\calD}\) to denote the vector formed by the diagonal entries of \(\Emat_a \Pmat \Emat_b^\top\) (i.e. the non-zero diagonal entries of \(\Pmat\).)
    
    Let \(\tvec_a:\bbR^{\calD_a}\to\bbR^{\calD}\) and \(\tvec_b:\bbR^{\calD_b}\to\bbR^{\calD}\) such that
    \begin{align*}
        \tvec_a'(\xvec) &= \Emat_a\Umat^\top \Lmat_a^{-1}\pth{\xvec - \muvec_a}\\
        \tvec_b'(\yvec) &= \Emat_b\Vmat^\top \Lmat_b^{-1}\pth{\yvec - \muvec_b}.
    \end{align*}
    
    It can be verified that
    \begin{align*}
        (\Xvec,\Yvec)\sim\calN(\muvec,\Sigmamat) \iff \pth{\tvec_a(\Xvec),\tvec_b(\Yvec)}\sim \calN\pth{\zerovec,\crocMat{\Imat}{\diag(\rhovec)}{\diag(\rhovec)}{\Imat}}.
    \end{align*}
\end{proof}

\subsection{Low per-feature correlation}

\hyperref[lemma:svd2rho]{Lemma \ref*{lemma:svd2rho}} shows the significance of \(\left|\left|\Sigmamat_a^{-1/2}\Sigmamat_{ab}\Sigmamat_b^{-1/2}\right|\right|_2\), which is used to characterize \hyperref[cond:highDimensional]{Condition \ref*{cond:highDimensional}}.

\begin{lemma}
    \label{lemma:svd2rho}
    Let \(\Sigmamat = \crocMat{\Sigmamat_a}{\Sigmamat_{ab}}{\Sigmamat_{ab}^\top}{\Sigmamat_b}\) be the covariance matrix between pairs of correlated feature vectors and let and the \(\rhovec\in(-1,1)^{\calD}\) correlation vector that characterizes the correlation in canonical form. Then
    \begin{align*}
        \left|\left|\Sigmamat_a^{-1/2}\Sigmamat_{ab}\Sigmamat_b^{-1/2}\right|\right|_2 = \max_{i\in\calD} |\rho_i|,
    \end{align*}
    where \(||\cdot||_2\) denotes the \(\ell_2\) operator norm, i.e. largest singular value.
\end{lemma}
\begin{proof}
    Let \(\Pmat,\Lmat_a,\Lmat_b,\Umat,\Vmat\) be as defined in the proof of \hyperref[lemma:canonicalVariance]{Lemma \ref*{lemma:canonicalVariance}}. Specifically, let \(\Lmat_a\) and \(\Lmat_b\) be triangular matrices such that \(\Lmat_a\Lmat_a^\top=\Sigmamat_a\) and \(\Lmat_b\Lmat_b^\top=\Sigmamat_b\), and \(\Pmat\) diagonal and \(\Umat,\Vmat\) orthonormal matrices such that \(\Umat\Pmat\Vmat^\top = \Lmat_a^{-1}\Sigmamat_{ab}(\Lmat_b^\top)^{-1}\).
    
    Define \(\Amat_0 \define \Sigmamat_a^{-1/2}\Sigmamat_{ab}\Sigmamat_b^{-1/2}\), \(\Amat_1 \define \Sigmamat_a^{-1/2}\Sigmamat_{ab}(\Lmat_b^\top)^{-1}\Vmat\) and \(\Amat_2 \define \Umat^\top\Lmat_a^{-1}\Sigmamat_{ab}(\Lmat_b^\top)^{-1}\Vmat\). 

    First we show that \(\Amat_2\) has the same singular values as \(\Amat_0\): The singular values of some matrix \(\Amat\) can be found by finding the eigenvalues of \(\Amat^\top\Amat\), or those of \(\Amat\Amat^\top\).
    \begin{itemize}
        \item Since \(\Lmat_b\Lmat_b^\top = \Sigmamat_b\) and \(\Vmat\Vmat^\top = \Imat\), it follows that \(\Amat_0\Amat_0^\top = \Amat_1\Amat_1^\top\).\\
        Then, \(\Amat_0\) and \(\Amat_1\) must have the same singular values.
        \item Since \(\Lmat_a\Lmat_a^\top = \Sigmamat_a\) and \(\Umat\Umat^\top = \Imat\), it follows that \(\Amat_1^\top\Amat_1 = \Amat_2^\top\Amat_2\).\\
        Then, \(\Amat_1\) and \(\Amat_2\) must have the same singular values.
    \end{itemize}
    Then \(\Amat_0\) and \(\Amat_2\) have the same singular values.

    In the proof of \hyperref[lemma:canonicalVariance]{Lemma \ref*{lemma:canonicalVariance}}, we are given that \(\Umat\Pmat\Vmat^\top = \Lmat_a^{-1}\Sigmamat_{ab}(\Lmat_b^\top)^{-1}\). Then, by the orthonormality of \(\Umat\) and \(\Vmat\), we have \(\Pmat = \Umat^\top\Lmat_a^{-1}\Sigmamat_{ab}(\Lmat_b^\top)^{-1}\Vmat = \Amat_2\).

    Since \(\Pmat = \diag\pth{\rhovec}\) is a diagonal matrix, its singular values are simply its diagonal entries in absolute value. Then the largest singular value is \(\max |\rho_i|\)
\end{proof}

\section{Achievability proofs}

The proofs for threshold testing and maximum row estimation state inequalities using the variable \(\nu \define \frac{\log|\calV|}{\log n}\). These inequalities directly translate to the statements in the main results using the variable \(\alpha\) by the fact that \(\nu=1\) if 1 if \(|\calV|=n\) and \(\nu=\max\acc{1, \frac{\log \pth{|\calV|-n}}{\log n}}\) if \(|\calV|>n\).

\subsection{Threshold testing}

\underline{\textbf{Quadratic boundary}}
\begin{proof}
    \label{proof:theorem:TT}
    Let \(\tau\) the threshold such that \(|\tau|\leq \zeta \). By \hyperref[lemma:typicality]{Lemma \ref*{lemma:typicality}} (\hyperref[lemma:typicalityPlanted]{Lemma \ref*{lemma:typicalityPlanted}}), the probability of a true pair failing the test is at most \(\exp\pth{-\frac{(\zeta -\tau)^2}{4\zeta }}\) and by \hyperref[lemma:FP]{Lemma \ref*{lemma:FP}} (\hyperref[lemma:FPPlanted]{Lemma \ref*{lemma:FPPlanted}}) the probability of a false pair passing the test is at most \(\exp\pth{-\frac{(\zeta +\tau)^2}{4\zeta }}\). The number of true pairs and false pairs are \(|M|\) and \(|\calU|\cdot|\calV|-|M|\) respectively. We bound the latter by \(|\calU|\cdot|\calV|\). So the expected number of false negatives is bounded by \(|M|\exp\pth{-\frac{(\zeta -\tau)^2}{4\zeta }}\) and expected number of false positives is bounded by \(|\calU|\cdot|\calV|\exp\pth{-\frac{(\zeta +\tau)^2}{4\zeta }}\).\\
    
    The log of the ratio of these two bounds is
    \begin{align*}
        \log\frac{|M|\exp\pth{-\frac{(\zeta -\tau)^2}{4\zeta }}}{|\calU|\cdot|\calV|\exp\pth{-\frac{(\zeta +\tau)^2}{4\zeta }}} &= \tau - \log\frac{|\calU|\cdot|\calV|}{|M|}.
    \end{align*}
    The choice of \(\tau = \log\frac{|\calU|\cdot|\calV|}{|M|}\) makes the log of the ratio zero. Hence the bounds for the expected numbers of false positives and negatives are equal. Then the bound on the number of errors is twice the bound on the number of false negatives.

    Let \(n\define |M|\), \(|\calU|=n\) and \(|\calV|=n^\nu\) for some \(\nu\geq 1\). Then, \(\tau = \log\frac{|\calU|\cdot|\calV|}{|M|} = \nu\log n\). Let \(x\define \frac{\zeta }{\log n}\). Then the number of false negatives (which is half the total error bound) is given by
    \begin{align*}
        |M|\exp\pth{-\frac{(\zeta -\tau)^2}{4\zeta }} &= n^{1-\frac{(x-\nu)^2}{4x}}.
    \end{align*}
    This expression is bounded by \(n^{1-\beta}\) if
    \begin{align}
        \label{eq:curve_LRT}
        x \geq \pth{\sqrt{\nu+\beta}+\sqrt{\beta}}^2.
    \end{align}

    This gives us the following inequalities that form part of the main results:
    \begin{compactitem}
        \item \hyperref[thm:alpha]{Theorem \ref*{thm:alpha}}\\
        Almost-exact alignment is achieved if (\ref{eq:curve_LRT}) is satisfied for some \(\beta\) such that \(n^{-\beta} \leq o(1)\), which is equivalent to \(\beta \geq \omega(1/\log n)\). Such \(\beta\) exists if \(\zeta \geq \nu\log n + \omega\pth{\sqrt{\log n}}\).\\
        Exact alignment is achieved if (\ref{eq:curve_LRT}) is satisfied for some \(\beta\) such that \(n^{1-\beta} \leq o(1)\), which is equivalent to \(\beta-1 \geq \omega(1/\log n)\). Such \(\beta\) exists if \(\zeta \geq \pth{1+\sqrt{1+\nu}}^2\log n + \omega\pth{1}\).
        \item \hyperref[thm:alphabeta]{Theorem \ref*{thm:alphabeta}}\\
        The number of errors is bounded by \(2n^{1-\beta}\) if \(\zeta  \geq \pth{\sqrt{\nu+\beta}+\sqrt{\beta}}^2\log n\).
        \item \hyperref[thm:beta]{Theorem \ref*{thm:beta}}\\
        This is a special case of \hyperref[thm:alphabeta]{Theorem \ref*{thm:alphabeta}} with \(\nu=1\).
    \end{compactitem}
\end{proof}

\subsection{Maximum row estimation}

Consider users \(u\in\calU\) and \(v,v'\in\calV\) such that \(u\mapped{M}v\). We want to bound the probability of the error event where \(u\) is falsely mapped to \(v'\). Under maximum row estimation, this corresponds to the event \(G_{u,v}\leq G_{u,v'}\). Without loss of generality, assume the set \(\calU\) consists of the single user \(u\).\\

\textbf{\underline{Linear boundary - relevant for large \(\beta\) and small \(\alpha\)}}

\begin{proof}
    By \hyperref[lemma:misalignment]{Lemma \ref*{lemma:misalignment}} (\hyperref[lemma:misalignmentPlanted]{Lemma \ref*{lemma:misalignmentPlanted}}), the probability of \(\{G_{u,v}\leq G_{u,v'}\}\) is upper bounded by \(\exp\pth{-\frac{\zeta }{2}}\).
    
    There are no more than\(|\calV|\) vertices \(v'\in\calV\setminus\{v\}\) that to which \(u\) can be falsely mapped. By the union bound, the probability that any of these events happens is upper bounded by \(|\calV|\exp\pth{-\frac{\zeta }{2}}\). Then, the expected number of misalignments over all \(|\calU|\) of rows is upper bounded by \(|\calU|\cdot|\calV|\exp\pth{-\frac{\zeta }{2}}\).

    Let \(n\define |M|\), \(|\calU|=n\), \(|\calV|=n^\nu\) for some \(\nu\geq 1\), and \(x\define \frac{\zeta }{\log n}\). Then, the bound on the expected number of misalignments is given by \(n^{1+\nu-\frac{x}{2}}\). This expression is bounded by \(n^{1-\beta}\) if
    \begin{align}
        \label{eq:curve_MRE0}
        x \geq 2(\nu+\beta).
    \end{align}

    This gives us the following inequalities that form part of the main results:
    \begin{compactitem}
        \item \hyperref[thm:alpha]{Theorem \ref*{thm:alpha}}\\
        Exact alignment is achieved if (\ref{eq:curve_MRE0}) is satisfied for some \(\beta\) such that \(n^{1-\beta} \leq o(1)\), which is equivalent to \(\beta-1 \geq \omega(1/\log n)\). Such \(\beta\) exists if \(\zeta \geq 2(1+\nu)\log n + \omega(1)\).
        \item \hyperref[thm:alphabeta]{Theorem \ref*{thm:alphabeta}}\\
        The number of errors is bounded by \(2n^{1-\beta}\) if \(\zeta  \geq 2(\nu+\beta)\log n\).
        \item \hyperref[thm:beta]{Theorem \ref*{thm:beta}}\\
        This is a special case of \hyperref[thm:alphabeta]{Theorem \ref*{thm:alphabeta}} with \(\nu=1\).
    \end{compactitem}
\end{proof}

\underline{\textbf{Quadratic boundary - relevant for small \(\beta\) and large \(\alpha\)}}

\begin{proof}
    Let \(\tau\) a threshold such that \(0\leq \tau \leq \zeta \). For the purpose of analysis, let us consider the alignment of the row corresponding to \(u\) a failure if either \(\{G_{u,v} < \tau\}\) or \(\{G_{u,v}\geq\tau \textrm{ and } G_{u,v}\leq G_{u,v'}\}\) for some \(v'\in\calV\setminus\{v\}\). In other words, we say the the algorithm has failed on the given row if either the true pair has score atypically low, in which case many false pairs will beat the true pair, or if a false pair beats the true pair despite the true pair having sufficiently high score. These two events fully cover the actual error event \(\{G_{u,v}\leq G_{u,v'}\}\).\\
    By \hyperref[lemma:typicality]{Lemma \ref*{lemma:typicality}} (\hyperref[lemma:typicalityPlanted]{Lemma \ref*{lemma:typicalityPlanted}}), the probability of \(\{G_{u,v} \leq \tau\}\) is bounded by \(\exp\pth{-\frac{(\zeta -\tau)^2}{4\zeta }}\). For database alignment, by \hyperref[lemma:condMisalignment]{Lemma \ref*{lemma:condMisalignment}}, the probability of \(\{G_{u,v}\geq\tau \textrm{ and } G_{u,v}\leq G_{u,v'}\}\) is bounded by \(\exp\pth{-\frac{\zeta ^2+\tau^2}{2\zeta}+6\rho_{\max}^2\tau}\). (For planted matching, by \hyperref[lemma:condMisalignmentPlanted]{Lemma \ref*{lemma:condMisalignmentPlanted}}, our bound is \(\exp\pth{-\frac{\zeta ^2+\tau^2}{2\zeta}}\) and there is additional term.) There are no more than \(|\calV|\) vertices \(v'\in\calV\setminus\{v\}\) that to which \(u\) can be falsely mapped. Then, by the union bound, the probability of \(\{G_{u,v}\geq\tau \textrm{ and } G_{u,v}\leq G_{u,v'}\}\) for some \(v'\) is bounded by \(|\calV|\exp\pth{-\frac{\zeta ^2+\tau^2}{2\zeta}+6\rho_{\max}^2\tau}\).\\
    The log of the ratio of these two bounds is
    \begin{align*}
        \log\frac{\exp\pth{-\frac{(\zeta -\tau)^2}{4\zeta }}}{|\calV|\exp\pth{-\frac{\zeta ^2+\tau^2}{2\zeta}+6\rho_{\max}^2\tau}} = \frac{(\zeta +\tau)^2}{4\zeta } - \log|\calV| - 6\rho_{\max}^2\tau
    \end{align*}
    for database alignment. (For planted matching, we drop the \(-6\rho_{\max}^2\tau\) term.)\\
    The choice of \(\tau^* = 2\sqrt{\zeta \log|\calV|}-\zeta \) makes the log of the ratio \(-6\rho_{\max}^2\tau^*\) for database alignment. Then the bound on the failure probability is no more than \(1+\exp\pth{6\rho_{\max}^2\tau^*}\) times the atypicality bound. (For planted matching, the log of the ratio is zero, so the bounds on the two types of error are equal. Therefore, the bound on the total error probability is twice that of the atypicality bound.)

    Let \(n\define |M|\), \(|\calU|=n\), \(|\calV|=n^\nu\) for some \(\nu\geq 1\), and \(x\define \frac{\zeta }{\log n}\). Then \(\tau^* = \pth{2\sqrt{\nu\cdot x}-x}\log n = 
    \pth{\nu -(\sqrt{x}-\sqrt{\nu})^2}\log n\). There are \(n\) rows. Then, the bound on the expected number of atypicality errors is given by
    \begin{align*}
        n\exp\pth{-\frac{(\zeta -\tau)^2}{4\zeta }} &= n^{1-\pth{\sqrt{\nu}-\sqrt{x}}^2}.
    \end{align*}
    
    This expression is bounded by \(n^{1-\beta}\) if
    \begin{align}
        \label{eq:curve_MRE1}
        x \geq \pth{\sqrt{\nu}+\sqrt{\beta}}^2.
    \end{align}
    For such \(x\), we get \(\frac{\tau^*}{\log n} = 2\sqrt{\nu x}-x \leq \nu-\beta\). So the total number of errors is \(n^{1-\beta}\cdot \pth{1+n^{6\rho_{\max}^2(\nu-\beta)}}\). By \hyperref[lemma:svd2rho]{Lemma \ref*{lemma:svd2rho}}, under \hyperref[cond:highDimensional]{Condition \ref*{cond:highDimensional}}, \(\rho_{\max}^2 \leq o(1)\) so the error bound becomes \(n^{1-\beta}\pth{1+n^{o(1)}} = n^{1-\beta+o(1)}\) for any finite value of \(\beta\) and \(\nu = \frac{\log |\calV|}{\log n}\).

    This gives us the following inequalities that form part of the main results:
    \begin{compactitem}
        \item \hyperref[thm:alpha]{Theorem \ref*{thm:alpha}}\\
        Almost-exact alignment is achieved if (\ref{eq:curve_MRE1}) is satisfied for some \(\beta\) such that \(n^{-\beta} \leq o(1)\), which is equivalent to \(\beta \geq \omega(1/\log n)\). Such \(\beta\) exists if \(\zeta \geq \nu\log n + \omega\pth{\sqrt{\log n}}\).\\
        Exact alignment is achieved if (\ref{eq:curve_MRE1}) is satisfied for some \(\beta\) such that \(n^{1-\beta} \leq o(1)\), which is equivalent to \(\beta-1 \geq \omega(1/\log n)\). Such \(\beta\) exists if \(\zeta \geq \pth{1+\sqrt{\nu}}^2\log n + \omega\pth{1}\).
        \item \hyperref[thm:alphabeta]{Theorem \ref*{thm:alphabeta}}\\
        The number of errors is bounded by \(2n^{1-\beta}\) if \(\beta \leq \nu\) and \(\zeta  \geq \pth{\sqrt{\nu}+\sqrt{\beta}}^2\log n\).
        \item \hyperref[thm:beta]{Theorem \ref*{thm:beta}}\\
        This is a special case of \hyperref[thm:alphabeta]{Theorem \ref*{thm:alphabeta}} with \(\nu=1\).
    \end{compactitem}
\end{proof}

\subsection{Maximum likelihood estimation}

\underline{\textbf{Linear boundary - relevant for large \(\beta\) and small \(\alpha\)}}

\begin{proof}
    Consider some elementary misalignment of size \(\delta\). (See \hyperref[subsec:elementary]{Subsection \ref*{subsec:elementary}} for elementary misalignments.) By \hyperref[lemma:misalignment]{Lemma \ref*{lemma:misalignment}} (\hyperref[lemma:misalignmentPlanted]{Lemma \ref*{lemma:misalignmentPlanted}}), the probability of the given misalignment is at most \(\exp\pth{-\delta\frac{\zeta }{2}}\).

    Let \(|\calU|=n\) and \(|\calV|=n+s\), where \(n\define|M|\) is the size of the matching. By \hyperref[lemma:countElementary]{Lemma \ref*{lemma:countElementary}}, there are at most \(\frac{n^\delta}{\delta}\) different type-I misalignments and at most \(sn^\delta\) different type-II misalignments of size \(\delta\). Furthermore, the number of type-I misalignments is 0 if \(\delta=1\). Define \(\varepsilon = \exp\pth{\log n - \frac{\zeta }{2}}\). Then, the expected number of type-I and type-II misalignments of size \(\delta\) are bounded by \(\varepsilon^\delta/\delta\) and \(s\varepsilon^\delta\) respectively.
    
    The contribution of a misalignment of size \(\delta\) is equal to \(\delta\). The total total contribution of all elementary misalignments gives us the expected number of errors. This expectation is bounded by
    \begin{align*}
        \sum_{\delta=2}^n \varepsilon^\delta + \sum_{\delta=1}^n s\delta\varepsilon^\delta
        &\leq \frac{\varepsilon^2}{1-\varepsilon} + \frac{s\varepsilon}{(1-\varepsilon)^2}
    \end{align*}
    Let us write \(x\define \frac{\zeta }{\log n}\) and \(\alpha \define \frac{\log s}{\log n}\) if \(s\geq 1\). Then the expression above can be written as \(\frac{n^{2-x}}{1-n^{1-x/2}} + \frac{n^{\alpha+1-x/2}}{\pth{1-n^{1-x/2}}^2}\).
    
    If \(x\geq 2+\frac{2\log \pth{\frac{\sqrt{5}-1}{2}}}{\log n}\) and \(s=0\), then this expression is bounded by 1. If \(x\geq 2+\frac{2\log\pth{\frac{3+\sqrt{5}}{2}}}{\log n}\) and \(\alpha> 0\), then the expression is bounded by \(n^\alpha(1+o(1))\).

    Furthermore, given some \(\beta>1+\Omega(1/\log n)\), the expected number of errors is
    \begin{align}
        \label{eq:curve_MLE0}
        \textrm{bounded by } n^{1-\beta} \textrm{ if } s = 0 \textrm{ and } x&\geq 1+\beta\\
        \label{eq:curve_MLE1}
        \textrm{bounded by } n^{1-\beta}(1+o(1)) \textrm{ if } s\geq 1 \textrm{ and } x&\geq 2(\alpha+\beta)
    \end{align}
    \begin{compactitem}
        \item \hyperref[thm:alpha]{Theorem \ref*{thm:alpha}}\\
        Exact alignment is achieved if (\ref{eq:curve_MLE0}) or (\ref{eq:curve_MLE1}) is satisfied for some \(\beta\) such that \(n^{1-\beta} \leq o(1)\), which is equivalent to \(\beta-1 \geq \omega(1/\log n)\). Such \(\beta\) exists if \(\zeta \geq (1+\alpha)\log n + \omega(1)\).
        \item \hyperref[thm:beta]{Theorem \ref*{thm:beta}}\\
        By (\ref{eq:curve_MLE0}), the number of errors is bounded by \(n^{1-\beta}\) if \(\zeta  \geq (1+\beta)\log n\).
        \item \hyperref[thm:alphabeta]{Theorem \ref*{thm:alphabeta}}\\
        By (\ref{eq:curve_MLE1}), the number of errors is bounded by \(n^{1-\beta}(1+o(1))\) if \(\zeta  \geq 2(\alpha+\beta)\log n\).
    \end{compactitem}
\end{proof}

\underline{\textbf{Elliptic boundary - relevant for smallest \(\beta\) and small \(\alpha\)}}

\begin{proof}

    Let \(\tau\) a threshold such that \(0\leq \tau \leq \zeta \). Consider a specific misalignment of size \(\delta\). Such a misalignment occurs if and only if one of the following is true:
    \begin{compactitem}
        \item Atypicality event: the average information density scores of the \(\delta\) true pairs is below \(\tau\), or
        \item Misalignment-despite-Typicality event: the set of \(\delta\) true pairs have average score greater than \(\tau\) but nevertheless the set of \(\delta\) false pairs have greater score than the corresponding set of true pairs. 
    \end{compactitem}

    By \hyperref[lemma:typicality]{Lemma \ref*{lemma:typicality}} (\hyperref[lemma:typicalityPlanted]{Lemma \ref*{lemma:typicalityPlanted}}), the probability of the true pairs having average score below the threshold is bounded by \(\exp\pth{-\delta\frac{(\zeta -\tau)^2}{4\zeta }}\). For database alignment, by \hyperref[lemma:condMisalignment]{Lemma \ref*{lemma:condMisalignment}}, the probability of that the false pairs have score greater than the true pair despite the true pairs having high score is bounded by \(\exp\pth{-\delta\cdot\frac{\zeta ^2+\tau^2}{2\zeta}+6\rho^2_{\max}\delta\tau}\). (For planted matching, by \hyperref[lemma:condMisalignmentPlanted]{Lemma \ref*{lemma:condMisalignmentPlanted}}, our bound is \(\exp\pth{-\delta\cdot\frac{\zeta ^2+\tau^2}{2\zeta}}\) and there is no extra term.)

    Let \(|\calU|=n\) and \(|\calV|=n+s\), where \(n\define|M|\) is the size of the matching. Let \(\alpha\define \frac{\log s}{\log n}\) and let \(\beta\in(0,1/2)\) some number less than \(1-\alpha\). Define \(\delta^*\define n^{1-\beta}\). Since \(\beta\leq 1-\alpha\), we have \(\delta^* = n^{1-\beta}\geq n^\alpha = s\). Consider some \(\delta \geq \delta^*\).

    By \hyperref[lemma:countMisalignment]{Lemma \ref*{lemma:countMisalignment}}, there are no more than \(\exp\bpth{\delta\pth{1+\log n + \log2}}\) different misalignment-despite-typicality events of size \(\delta\). Then, in log-expectation, the number of such events is no more than \(\delta\pth{\pth{1+\log n + \log2}-\frac{\zeta ^2+\tau^2}{2\zeta }+6\rho^2_{\max}\tau}\). If
    \begin{align}
        \label{eq:tau_MLE1}
        \tau = \sqrt{2\zeta \pth{\log n + \log \eta +1+\log2}-\zeta ^2},
    \end{align}
    the expected misalignment-despite-typicality events of size \(\delta\) is bounded by \(\exp\pth{-\delta (\eta-6\rho^2_{\max}\tau)}\). Define \(\varepsilon \define \exp\pth{-\eta+6\rho^2_{\max}\tau}\). \(\varepsilon \in (0,1)\) if \(\eta > 6\rho^2_{\max} I_{XY}\) (which is greater than \( 6\rho^2_{\max}\tau\)). Then the bound on the number of such events of size at least \(\delta^*\) is bounded by \(\sum_{\delta\geq\delta^*} \varepsilon^{\delta} \leq \frac{\varepsilon^{\delta^*}}{1-\varepsilon}\).

    By \hyperref[lemma:svd2rho]{Lemma \ref*{lemma:svd2rho}}, under \hyperref[cond:highDimensional]{Condition \ref*{cond:highDimensional}}, \(\rho_{\max}^2 \leq o(1)\). Then, there exists some choice for \(\eta\) that is \(o(\log n)\) and satisfies \(\eta > 6\rho^2_{\max} I_{XY}\).
    
    If \(\beta\leq 1/2\), then \(\delta^* > n^{1-\beta} \geq \sqrt{n}\). If, furthermore, \(\eta - 6\rho^2_{\max} I_{XY} \geq \Omega(1)\), then \(\frac{\varepsilon^{\delta^*}}{1-\varepsilon} = \frac{\exp\pth{-\Omega(\sqrt{n})}}{1-\exp\pth{-\Omega(1)}}\leq e^{-\Omega(\sqrt{n})}\). A misalignment can result in no more than \(n\) errors. Then, the expected number of errors caused by misalignment-despite-typicality errors of size at least \(\delta^*\) is \(ne^{-\Omega(\sqrt{n})} \leq o(1)\).

    Next we confirm the number of atypicality errors is small: There are no more than \(\binom{n}{\delta}\) different ways to get an atypicality event. This is bounded by \(\exp\pth{\delta + \delta\log\frac{n}{\delta}}\), which can further be bounded by \(\exp\pth{\delta + \delta\log\frac{n}{\delta^*}}\). Then, in expectation, there are no more than \(\exp\pth{\delta + \delta\log\frac{n}{\delta^*}-\delta\frac{(\zeta -\tau)^2}{4\zeta }}\) atypicality events of size \(\delta\).

    The log of the ratio of the bound on expected number of atypicality errors versus the expected number of misalignment-despite-typicality errors is equal to \(\delta\pth{\frac{(\zeta +\tau)^2}{4\zeta }-\log \delta^*-2-6\rho_{\max}^2\tau}\). If
    \begin{align}
        \label{eq:tau_MLE0}
        \tau \leq 2\sqrt{\zeta \pth{\log \delta^* + \log2}}-\zeta ,
    \end{align}
    then the log-ratio is at most \(-6\rho_{\max}^2\tau<0\) and the bound on the expected number of atypicality events of size \(\delta\) is bounded by that of misalignment-despite-typicality events of size \(\delta\).
    
    We identify the smallest \(\delta^*\) such that our choice of \(\tau\) in (\ref{eq:tau_MLE1}) satisfies the inequality in (\ref{eq:tau_MLE0}): Let us write \(x\define \frac{\zeta }{\log n}\). By (\ref{eq:tau_MLE1}), for the appropriate choice of \(\eta\) on the order of \(o(\log n)\) we have , \(\frac{\tau}{\log n} = \sqrt{2x\pth{1+o(1)}-x}\). For \(x = 1+2\sqrt{\beta(1-\beta)}\), such choice of \(\tau\) gives us \(\frac{\tau}{\log n} \leq |1-2\beta|+o(1)\). There exists some \(\varepsilon_1,\varepsilon_2\leq o(1)\) such that for any \(\beta \in (\varepsilon_1,1/2-\varepsilon_2)\), the choices of \(\tau \leq |1-2\beta|+o(1)\) and \(x = 1+2\sqrt{\beta(1-\beta)}\) satisfy (\ref{eq:tau_MLE0}).
    
    By (\ref{eq:tau_MLE0}), we require \(\frac{\tau}{\log n}\leq 2\sqrt{x\pth{1-\beta+\varepsilon+\frac{\log 2}{\log n}}}-x\). Picking \(x=1+2\sqrt{\beta(1-\beta)}\), there exists some \(\varepsilon \leq o(1)\) that satisfies \(\frac{\tau}{\log n}\leq |1-2\beta|+o(1)\leq 2\sqrt{x\pth{1-\beta+\varepsilon+\frac{\log 2}{\log n}}}-x\).

    We have shown that errors from misalignments of size greater than \(\delta^*\) is \(o(1)\): Those due to misalignment-despite-typicality type errors is \(o(1)\) and those due to atypicality type errors is bounded by that of misalignment-despite-typicality type errors, so also \(o(1)\). Finally, the expected number of errors due to misalignments smaller than \(\delta^*\) is at most \(\delta^*\). (This follows from the fact that only one of the misalignments can occur.)

    This gives us the following inequality that form part of the main results:
    \begin{compactitem}
        \item \hyperref[thm:alpha]{Theorem \ref*{thm:alpha}}\\
        Almost-exact alignment is achieved if \(x\geq 1+2\sqrt{\beta(1-\beta)}\) is satisfied for some \(\beta\) such that \(n^{-\beta} \leq o(1)\), which is equivalent to \(\beta \geq \omega(1/\log n)\). Such \(\beta\) exists if \(\zeta \geq \nu\log n + \omega\pth{\sqrt{\log n}}\).
        \item \hyperref[thm:beta]{Theorem \ref*{thm:beta}} and \hyperref[thm:alphabeta]{Theorem \ref*{thm:alphabeta}}\\
        The number of errors is bounded by \(n^{1-\beta+o(1)}\) if \(\beta\in(0,1/2)\), \(\beta\leq 1-\alpha\) and \(\zeta  \geq \pth{1+2\sqrt{\beta(1-\beta)}}\log n\).
    \end{compactitem}
\end{proof}

\underline{\textbf{Quadratic boundary - relevant for small \(\beta\) and large \(\alpha\)}}

\begin{proof}
    By the previous proof, we know that, for \(\tilde{\beta} = 1-\frac{\log s}{\log n}\), the expected number of errors due to misalignments of size at least \(s\) is \(o(1)\) if \(\zeta  \geq \pth{1+\sqrt{\tilde{\beta}(1-\tilde{\beta})}}\log n\). Here we show that, given a stronger bound on \(\zeta \), we can also bound the number of errors due to misalignments of size between \(\delta^*\) and \(s\) for some \(\delta^* = n^{1-\beta+\eta} \leq s\) where \(\beta\in[0,1/2]\) another constant strictly less than \(\tilde{\beta}\) and \(\eta\) some non-negative function of \(n\) which is to be determined later.

    Let \(\tau\) a threshold such that \(0\leq \tau \leq \zeta \). Consider a specific misalignment of size \(\delta\). Once again, we cover the misalignment event using two auxiliary events:
    \begin{compactitem}
        \item Atypicality event: the average information density scores of the \(\delta\) true pairs is below \(\tau\). By \hyperref[lemma:typicality]{Lemma \ref*{lemma:typicality}} (\hyperref[lemma:typicalityPlanted]{Lemma \ref*{lemma:typicalityPlanted}}), the probability of this event for size \(\delta\) is bounded by \(\exp\pth{-\delta\frac{(\zeta -\tau)^2}{4\zeta }}\).
        \item Misalignment-despite-Typicality event: the set of \(\delta\) true pairs have average score greater than \(\tau\) but nevertheless the set of \(\delta\) false pairs have greater score than the corresponding set of true pairs. For database alignment, by \hyperref[lemma:condMisalignment]{Lemma \ref*{lemma:condMisalignment}}, the probability of that the false pairs have score greater than the true pair despite the true pairs having high score is bounded by \(\exp\pth{-\delta\cdot\frac{\zeta ^2+\tau^2}{2\zeta}+6\rho^2_{\max}\delta\tau}\). (For planted matching, by \hyperref[lemma:condMisalignmentPlanted]{Lemma \ref*{lemma:condMisalignmentPlanted}}, our bound is \(\exp\pth{-\delta\cdot\frac{\zeta ^2+\tau^2}{2\zeta}}\) and there is no extra term.)
    \end{compactitem}

    Let \(|\calU|=n\) and \(|\calV|=n+s\), where \(n\define|M|\) is the size of the matching. The number of atypicality events of size \(\delta\) is bounded by \(\exp\pth{\delta + \delta \log \frac{n}{\delta}}\). By \hyperref[lemma:countMisalignment]{Lemma \ref*{lemma:countMisalignment}}, the number of misalignment-despite-typicality events of size \(\delta\) is bounded by \(\exp\bpth{\delta\pth{1+\log \frac{ns}{\delta} + \log2}}\).

    The log of the ratio of the bounds on the expected number of atypicality events and the expected number of misalignment-despite-typicality events is equal to \(\delta\pth{\frac{(\zeta +\tau)^2}{4\zeta }-\log s - \log 2 - 6\rho^2_{\max}\tau}\). If
    \begin{align}
        \label{eq:tau_MLE3}
        \tau = 2\sqrt{\zeta \pth{\log s + \log2}}-\zeta ,
    \end{align}
    the log-ratio is \(-6\rho^2_{\max}\tau<0\) and the bound on the expected number of atypicality events of size \(\delta\) is bounded by that of misalignment-despite-typicality events of size \(\delta\).

    Define \(x=\frac{\zeta }{\log n}\) and \(\alpha = \frac{\log s}{\log n}\). Then, given the value of \(\tau\) in (\ref{eq:tau_MLE3}), the expected number of misalignments-despite-typicality errors of size \(\delta\geq \delta^*\) bounded by \(\exp\pth{\delta(1+\log 2) + \delta\log\frac{ns}{\delta}-\delta\cdot\frac{\zeta ^2+\tau^2}{2\zeta }+6\rho^2_{\max}\tau}\) which is further bounded by \(\exp\pth{\delta(1+\log 2) + \delta\log\frac{ns}{\delta^*}-\delta\cdot\frac{\zeta ^2+\tau^2}{2\zeta }+6\rho^2_{\max}\tau} = n^{\delta\pth{\beta-\eta - (\sqrt{x}-\sqrt{\alpha})^2}}e^{\delta(1+\log 2 + 6\rho^2_{\max}\tau)}\). Pick \(\eta\) to be \(\frac{1+\log 2 + 6\rho^2_{\max}\tau}{\log n}\). Then the expression in the previous bound simplifies at \(n^{\delta\pth{\beta-(\sqrt{x}-\sqrt{\alpha})^2}}\). Define \(\varepsilon = n^{\delta\pth{\beta-(\sqrt{x}-\sqrt{\alpha})^2}}\). \(\varepsilon < 1\) if \(x > \pth{\sqrt{\alpha}+\sqrt{\beta}}^2\) and \(\alpha \geq \beta\). The contribution of each event to the size of the misalignment is at most \(\delta\). (This is because at most one misalignment can occur at a time. So the contribution of a misalignment-despite-typicality event is either 0 or \(\delta\).) Then, the total number of errors due to misalignment-despite-typicality events is bounded as
    \begin{align*}
        \sum_{\delta=0}^\infty \delta\varepsilon^\delta &= \frac{\varepsilon}{(1-\varepsilon)^2}.
    \end{align*}
    The bound above is \(o(1)\) if \(\varepsilon \leq o(1)\). In that case, the number of errors due to misalignments of size \(\delta\in\croc{n^{1-\beta+\eta},n^{1-\tilde{\beta}}}\) as well as those of size \(\delta\in\croc{n^{1-\tilde{\beta}},n}\) is \(o(1)\). The expected number of errors due to misalignments smaller than \(n^{1-\beta+\eta}\) is at most \(n^{1-\beta+\eta}\). (This follows from the fact that only one of the misalignments can occur.)

    By \hyperref[lemma:svd2rho]{Lemma \ref*{lemma:svd2rho}}, under \hyperref[cond:highDimensional]{Condition \ref*{cond:highDimensional}}, \(\rho_{\max}^2 \leq o(1)\). Then, \(\eta = \frac{1+\log 2 + 6\rho^2_{\max}\tau}{\log n} \leq o(1)\).
    
    \(\varepsilon \leq o(1)\) is equivalent to \(x > (\sqrt{\alpha}+\sqrt{\beta+\omega(1/\log n)})^2\), which we can rewrite as
    \begin{align}
        \label{eq:curve_MLE4}
        x > (\sqrt{\alpha}+\sqrt{\beta})^2 + \pth{1+\sqrt{\frac{\alpha}{\beta}}}\omega(1/\log n)
    \end{align}

    This gives us the following inequality that form part of the main results:
    \begin{compactitem}
        \item \hyperref[thm:alpha]{Theorem \ref*{thm:alpha}}\\
        Almost-exact alignment is achieved if (\ref{eq:curve_MLE4}) is satisfied for some \(\beta\) such that \(n^{-\beta} \leq o(1)\), which is equivalent to \(\beta \geq \omega(1/\log n)\). Such \(\beta\) exists if \(\zeta \geq \nu\log n + \omega\pth{\sqrt{\log n}}\).\\
        Exact alignment is achieved if (\ref{eq:curve_MLE4}) is satisfied for some \(\beta\) such that \(n^{1-\beta} \leq o(1)\), which is equivalent to \(\beta-1 \geq \omega(1/\log n)\). Such \(\beta\) exists if \(\alpha\geq 1\) and \(\zeta \geq (1+\sqrt{\alpha})^2\log n + \omega(1)\).
        \item \hyperref[thm:alphabeta]{Theorem \ref*{thm:alphabeta}}\\
        The number of errors is bounded by \(n^{1-\beta+o(1)}\) if \(1-\alpha \leq \beta \leq \alpha\) and \(\zeta  \geq \pth{\sqrt{\alpha}+\sqrt{\beta}}^2\log n\).
        \item \hyperref[thm:beta]{Theorem \ref*{thm:beta}}\\
        This is a special case of \hyperref[thm:alphabeta]{Theorem \ref*{thm:alphabeta}} with \(\alpha=1\).
    \end{compactitem}
\end{proof}

\section{Converse proofs for planted matching}
Our achievability statements take the form of upper bounds on \(E[d(\hat{M},M]\) in terms of \(\zeta\).
If \(E[d(\hat{M},M] \leq n^{1-\beta}\), then by Markov's inequality, these can be converted into upper bounds on \(\Pr[d(\hat{M},M) \geq \frac{n^{1-\beta}}{\varepsilon}] \leq \varepsilon\).
Note that \(\frac{n^{1-\beta}}{\varepsilon} = \exp((1-\beta)\log n + \log \frac {1}{\varepsilon})\), so \(\varepsilon\) appears only in a lower order term.

For $0 \leq \beta < 1$, we prove converse statements of the form \(\Pr[d(\hat{M},M) \leq n^{1-\beta}] \leq o(1)\).
These imply bounds \(E[d(\hat{M},M)] \geq(1-o(1))n^{1-\beta}\).

For all \(\beta > 1\), \(d(\hat{M},M) \leq n^{1-\beta} \) is equivalent to \(\hat{M}=M\).
Thus \(E[d(\hat{M},M)] \geq 1 - \Pr[\hat{M}=M]\).
If \(\Pr[\hat{M}=M] \leq \exp(-n^{1-\beta}) \) for \(\beta >1\), then again we have \(E[d(\hat{M},M)] \geq (1-o(1))n^{1-\beta}\).

\paragraph{Technical Lemmas:} Throughout this section, let $|\calV| = n+s$ and $s = n^\alpha$.
Recall that in the planted matching model, $\zeta = \mu^2/2$.

\begin{lemma}
\label{lemma:distance-converse}
Let \(0 \leq \beta \leq 1\) and \(\alpha \leq \Omega(1)\).
If for any estimator \(\hat{M}\) we have the bound 
\begin{equation}
\log \log \Pr[\hat{M} = M]^{-1} \geq (1-\beta)\log n + \log \log n +\omega(1),
\label{eqn:map-converse}
\end{equation}
then for any estimator \(\hat{M}'\) we have \(\Pr[d(\hat{M}',M) \leq n^{1-\beta}] \leq o(1)\). 
\end{lemma}
\begin{proof}
  
  Given an estimator \(\hat{M}'\), define a second estimator $\hat{M}$ as follows: select a matching uniformly from those within distance $d$ of $\hat{M}'$ such that $\hat{M}$ and the observation are conditionally independent given $\hat{M}'$.
  Then \[
    \Pr[\hat{M} = M] \geq \frac{\Pr[d(\hat{M}',M) \leq n^{1-\beta}]}{|\{m : d(m,\hat{M}') \leq n^{1-\beta}\}|}.
  \]
Combining this with \eqref{eqn:map-converse} and Lemma~\ref{lemma:countMisalignment}, we have
\begin{align*}
\log \Pr[d(\hat{M}',M) \leq n^{1-\beta}]^{-1} 
&\geq \log \Pr[\hat{M}' = M]^{-1} - n^{1-\beta}(2 + \max(1, \alpha+\beta) \log n)\\
&\geq \omega(n^{1-\beta} \log n) - O(n^{1-\beta} \log n)\\
&\geq \omega(n^{1-\beta} \log n) \geq \omega(\log n).
\end{align*}
\end{proof}

\begin{lemma}
For any estimator \(\hat{M}\)
\[
\Pr[\hat{M} = M] \leq \inf_{0 \leq \gamma \leq 1} \sum_k \binom{n}{k} \frac{s!}{(k+s)!} (1-Q((1-\gamma)\mu))^{n-k}(1-Q(\gamma\mu))^{k(k+s)} e^{k (2\gamma-1)\mu^2/2}
\]
\label{lemma:covering}
\end{lemma}
\begin{proof}
  For any estimator $\hat{M}$,
  \begin{equation}
    \Pr[\hat{M} = M] \leq \int_{\mathbb{R}^{n \times n}} \max_{m} \frac{s!}{(n+s)!}f(g-\mu m) dg \label{map-accuracy}
  \end{equation}
  where $f(g) = (2 \pi)^{\frac{n^2}{2}} \exp(-\frac{1}{2}\sum_{i,j} g_{i,j}^2)$, the density function of a $n^2$-dimensional standard Gaussian vector, and the maximization is over $n \times (n+s)$ matching matrices.

  Write $|m|$ for $\sum_{i,j} m_{i,j}$, which is the number of ones in $m$ and the number of pairs in the matching.
  Now we upper bound \eqref{map-accuracy} by
  \begin{equation}
    \int_{\mathbb{R}^{n \times n}} \max_{m} \frac{e^{(n-|m|)\theta}s!}{(n+s)!}f(g-\mu m) dg \label{map-accuracy-two}
  \end{equation}
  where the maximization is now over all partial and full matching matrices.

  Let $m'$ be a matching such that $m' \geq m$ and $|m'| = |m|+1$.
  Let $(i,j)$ be the entry where $m'_{i,j} = 1$ and $m_{i,j} = 0$.
  Then $f(x - \mu m') \geq f(x - \mu m')$ when
  \begin{align*}
    e^{(n-|m|-1)\theta}\exp(-(x_{i,j} - \mu)^2/2) &\geq e^{(n-|m|)\theta} \exp(-x_{i,j}^2/2)\\
    -2\theta - (x_{i,j} - \mu)^2 &\geq -x_{i,j}^2\\
    -2 \theta + 2 \mu x_{i,j} - \mu^2 &\geq 0\\
    x_{i,j} &\geq \frac{\theta}{\mu} + \frac{\mu}{2} = \mu \pth{\frac{\theta}{\mu^2}+\frac{1}{2}} 
  \end{align*}
  Let $\gamma = \frac{\theta}{\mu^2}+\frac{1}{2}$: the location of the boundary between the regions covered by $m$ and $m'$ as a fraction of the distance between the means.
  We will select $\gamma$ and use $\theta = \frac{2\gamma-1}{2}\mu^2$.
  
  A partial matching $m$ of size $|m| = n-k$ has $n-k$ smaller neighboring matchings and $k(k+s)$ larger neighbors.
  If $m_{i,j} = 1$, the gaussian centered at the neighboring matching with $m'_{i,j}=0$ has higher density in the region $x_{i,j} \leq \gamma\mu$
  The measure of the density centered at $m$ in that tail is $Q((1-\gamma)\mu)$.
  If $m_{i,j} = 0$ and there is a neighboring matching $m'$ with $m'_{i,j} = 1$, the gaussian centered at $m'$ has higher density in the region $x_{i,j} \geq \gamma\mu$
  The measure in that tail is $Q(\gamma\mu)$.
  These inequalities each involve one entry of the matrix, which are independent random variables.
  The measure not part of some tail is $(1-Q((1-\gamma)\mu))^{n-k}(1-Q(\gamma\mu))^{k(k+s)}\frac{e^{k \theta}s!}{(n+s)!}$.

  There are $\binom{n}{k}\binom{n+s}{k+s}(n-k)! = \binom{n}{k}\frac{(n+s)!}{(k+s)!}$ partial matchings of size $n-k$.
  Thus \eqref{map-accuracy-two} is at most
  \begin{align*}
  &\frac{s!}{(n+s)!}\sum_k \binom{n}{k} \frac{(n+s)!}{(k+s)!} (1-Q((1-\gamma)\mu))^{n-k}(1-Q(\gamma\mu))^{k(k+s)}e^{k \theta}\\
    &= \sum_k \binom{n}{k} \frac{s!}{(k+s)!} (1-Q((1-\gamma)\mu))^{n-k}(1-Q(\gamma\mu))^{k(k+s)} e^{k \theta}
  \end{align*}
\end{proof}
\begin{lemma}
\label{lemma:ellipse-coverse}
  Let $\log n \leq \mu^2/2 \leq 2 \log n - \omega\pth{\log \log n}$.
  For any estimator $\hat{M}$,
  \[
    \frac{\log \log \Pr[\hat{M} = M]^{-1}}{\log n} \geq \frac{1}{2} + \sqrt{\frac{\mu^2}{4\log n}\pth{1 - O\pth{\frac{\log \log n}{\log n}} - \frac{\mu^2}{4\log n}}} - O\pth{\frac{\log \log n}{\log n}}.
  \]
\end{lemma}
\begin{proof}
  Let $x = Q(\gamma\mu)$, $y = Q((1-\gamma)\mu)$, and $\theta = \frac{2\gamma-1}{2}\mu^2$.
  From Lemma~\ref{lemma:covering}, for any $0 \leq \gamma \leq 1$, 
  \begin{align*}
  \Pr[\hat{M} = M]
  &\leq \sum_k \binom{n}{k} \frac{s!}{(k+s)!} (1-y)^{n-k}(1-x)^{k(k+s)}e^{k \theta}\\
  &\leq \sum_k \binom{n}{k} (1-y)^{n-k}e^{-k^2 x}\frac{e^{k \theta}}{k!}\\
  &\leq \sum_k \binom{n}{k} (1-y)^{n-k}e^{-(2 \ell k - \ell^2)x}e^{k \theta}\\
  & = e^{\ell^2 x}(1-y+e^{\theta-2\ell x})^n,
\end{align*}
where we used $(k+s)! \geq s!$ and $k(k+s) \geq k^2 \geq 2\ell k - \ell^2$, which holds for any $\ell$.
Suppose that can find values of $\theta$ and $\ell$ such that $e^{\theta-2\ell x} \leq y/3$ and $\ell^2 x \leq ny/3$.
Then
\begin{equation*}
  e^{\ell^2}(1-y+e^{\theta-2\ell x})^n \leq e^{ny/3}(1-2y/3)^n \leq e^{-ny/3}
\end{equation*}
and $\log \log \Pr[\hat{M}=M]^{-1} \geq \log(ny/3)$.

Let $q(z) = \frac{\exp(-z^2/2)}{Q(z)}$, so $1 \leq q(z) \leq O(z)$.
Then $x = Q(\gamma\mu) = \exp(-\gamma^2\mu^2/2)/q(\gamma\mu)$ and $y = Q((1-\gamma)\mu) = \exp(-(1-\gamma)^2\mu^2/2)/q((1-\gamma)\mu)$.

Now we find the value of $\ell$ that satisfies $e^{\theta-2\ell x} = y/3$:
\begin{align*}
  2 \ell x
  &= \theta - \log(y/3)\\
  &= \frac{(2\gamma-1)\mu^2}{2} + \log(3 q((1-\gamma) \mu)) + \frac{(1-\gamma)^2\mu^2}{2}\\
  &= \frac{\gamma^2 \mu^2}{2} + \log(3 q((1-\gamma) \mu)).
\end{align*}

Now we find a condition on $\gamma$ that ensures $(\ell x)^2 \leq nxy/3$. 
Expanding the definitions of $x$ and $y$, we have
\begin{align*}
  \frac{nxy}{3}
  &= \frac{1}{3 q(\gamma \mu) q((1-\gamma)\mu)}\exp\pth{\log n - \frac{\gamma^2\mu^2}{2}-\frac{(1-\gamma)^2\mu^2}{2}}
\end{align*}
so we get the condition
\[
  \log n
  \geq (\gamma^2 + (1-\gamma)^2)\frac{\mu^2}{2} + 2 \log\pth{\frac{\gamma^2 \mu^2}{4} + \frac{\log(3 q((1-\gamma) \mu))}{2}} + \log (3 q(\gamma \mu) q((1-\gamma)\mu)).
\]
Using $0 \leq \gamma \leq 1$ and picking the worst possible $\gamma$ in the lower order terms, we observe that a stronger condition is
\begin{align*}
  \log n
  &\geq (\gamma^2 + (1-\gamma)^2)\frac{\mu^2}{2} + 2 \log\pth{\frac{\mu^2}{4} + \frac{\log(3 q(\mu))}{2}} + \log (3 q(\mu)^2)\\
  &= (\gamma^2 + (1-\gamma)^2)\frac{\mu^2}{2} + \epsilon \log n.  
\end{align*}
Using $\mu^2 \leq 4 \log n$, we see that $0 \leq \epsilon \leq O\pth{\frac{\log \log n}{\log n}}$.
Because our final upper bound is $e^{-ny/3}$, we want to maximize $y$ and thus maximize $\gamma$.
Let 
$a  = \frac{\mu^2}{4 (1-\epsilon) \log n}$.
The larger root of $\frac{1}{2a} = \gamma^2 + (1-\gamma)^2$ is $\gamma = \frac{1+ \sqrt{\frac{1}{a}-1}}{2}$, so we need $a \leq 1$  to get a real root and $a \geq \frac{1}{2}$ to get $\gamma \leq 1$.
Then $(1-\gamma)^2 = \frac{\frac{1}{2a} - \sqrt{\frac{1}{a}-1}}{2}$, and $2a(1-\gamma)^2 = \frac{1}{2} - \sqrt{a-a^2}$, and
When $\mu^2 \leq 4 \log n - \omega(\log \log n)$, $a \geq 1/2$ for sufficiently large $n$.

Let

\begin{align*}
\frac{\log(ny/3)}{\log n}
  &= \frac{\log n + \log Q(\gamma\mu)- \log 3}{\log n}\\
  &= 1 - \frac{(1-\gamma)^2\mu^2}{2 \log n}- \frac{\log q((1-\gamma)\mu) - \log 3}{\log n}\\
  &= 1 - 2(1-\epsilon)a(1-\gamma)^2 - O\pth{\frac{\log \log n}{\log n}}\\
  &= 1 - (1-\epsilon)\pth{\frac{1}{2} - \sqrt{a(1-a)}}- O\pth{\frac{\log \log n}{\log n}}\\
  &= \frac{1}{2} + \frac{\epsilon}{2} + \sqrt{(1-\epsilon)a(1-\epsilon - (1-\epsilon)a)}- O\pth{\frac{\log \log n}{\log n}}\\
  &\geq \frac{1}{2} + \sqrt{\frac{\mu^2}{\log n}\pth{1 - O\pth{\frac{\log \log n}{\log n}} - \frac{\mu^2}{\log n}}} - O\pth{\frac{\log \log n}{\log n}}.
\end{align*}


\end{proof}

\begin{lemma}
\label{lemma:parabola-converse}
  Let $\alpha \log n + \omega(\log \log n) \leq \mu^2/2$.
  For any estimator $\hat{M}$,
  \[
    \frac{\log \log \Pr[\hat{M} = M]^{-1}}{\log n} \geq 1 - \pth{\sqrt{\frac{\mu^2}{2 \log n}} - \sqrt{\alpha-O\pth{\frac{\log \log n}{\log n}}}}^2 - O\pth{\frac{\log \log n}{\log n}}.
  \]
\end{lemma}
\begin{proof}
  Let $x = Q(\gamma\mu)$, $y = Q((1-\gamma)\mu)$, and $\theta = \frac{2\gamma-1}{2}\mu^2$.
  From Lemma~\ref{lemma:covering}, for any $0 \leq \gamma \leq 1$, 
  \begin{align*}
  \Pr[\hat{M} = M]
  &\leq \sum_k \binom{n}{k} \frac{s!}{(k+s)!} (1-y)^{n-k}(1-x)^{k(k+s)}e^{k \theta}\\
  &\sum_k \binom{n}{k} \frac{(n+s)!}{(k+s)!} (1-Q((1-\gamma)\mu))^{n-k}(1-Q(\gamma\mu))^{k(k+s)}\frac{e^{k \theta}s!}{(n+s)!}\\
  &= \sum_k \binom{n}{k} (1-y)^{n-k}(1-x)^{k(k+s)}\frac{e^{k \theta}s!}{(k+s)!}\\
  &\leq \sum_k \binom{n}{k} (1-y)^{n-k}e^{-k(k+s) x}\frac{e^{k \theta}s!}{(k+s)!}\\
  &\leq \sum_k \binom{n}{k} (1-y)^{n-k}e^{-ksx}e^{k \theta}\\
  & = (1-y+e^{\theta-s x})^n,
\end{align*}
where we used $(k+s)! \geq s!$ and $k^2 \geq 0$.
Suppose that can find $\theta$ such that $e^{\theta-s x} \leq y/2$.
Then
\begin{equation*}
  (1-y+e^{\theta-s x}) \leq (1-y/2)^n \leq e^{-ny/2}.
\end{equation*}
and $\log \log \Pr[\hat{M}=M]^{-1} \geq \log(ny/2)$.

Let $q(z) = \frac{\exp(-z^2/2)}{Q(z)}$, so $1 \leq q(z) \leq O(z)$.
Then $x = Q(\gamma\mu) = \exp(-\gamma^2\mu^2/2)/q(\gamma\mu)$ and $y = Q((1-\gamma)\mu) = \exp(-(1-\gamma)^2\mu^2/2)/q((1-\gamma)\mu)$.

Now we find the value of $\theta$ that satisfies $e^{\theta-s x} \leq y/2$:
\begin{align*}
  s x
  &\geq \theta - \log(y/2)\\
  &= \frac{(2\gamma-1)\mu^2}{2} + \log(2 q((1-\gamma) \mu)) + \frac{(1-\gamma)^2\mu^2}{2}\\
  &= \frac{\gamma^2 \mu^2}{2} + \log(3 q((1-\gamma) \mu))\\
  \alpha \log n - \gamma^2 \mu^2/2 - \log(q(\gamma \mu))
  &\geq \log\pth{\frac{\gamma^2 \mu^2}{2} + \log(3 q((1-\gamma) \mu))}
\end{align*}
A stronger condition is
\begin{align*}
    \alpha \log n - \gamma^2 \mu^2/2
  &\geq \log(q(\mu))+\log\pth{\frac{\mu^2}{2} + \log(3 q(\mu))}\\ 
  &= \epsilon \log n.
\end{align*}
We have $0 \leq \epsilon \leq O\pth{\frac{\log \log n}{\log n}}$.
Pick $\gamma = \sqrt{\frac{2 (\alpha - \epsilon) \log n}{\mu^2}}$.
To ensure $\gamma \leq 1$, we need $\mu^2/2 \geq \alpha \log n + \omega(\log \log n)$.
Then
\begin{align*}
  \frac{\log(ny/2)}{\log n}
  &= \frac{\log n + \log Q(\gamma\mu)- \log 2}{\log n}\\
  &= 1 - \frac{(1-\gamma)^2\mu^2}{2 \log n}- \frac{\log q((1-\gamma)\mu) - \log 2}{\log n}\\
  &= 1 - \pth{\sqrt{\frac{\mu^2}{2 \log n}} - \sqrt{\frac{\gamma^2\mu^2}{2 \log n}}}^2 - \frac{\log q((1-\gamma)\mu) - \log 2}{\log n}\\
    &= 1 - \pth{\sqrt{\frac{\mu^2}{2 \log n}} - \sqrt{\alpha-\epsilon}}^2 - \frac{\log q((1-\gamma)\mu) - \log 2}{\log n}\\
  &= 1 - \pth{\sqrt{\frac{\mu^2}{2 \log n}} - \sqrt{\alpha-O\pth{\frac{\log \log n}{\log n}}}}^2 - O\pth{\frac{\log \log n}{\log n}}.
\end{align*}
\end{proof}

\begin{theorem}
    Let \(n = |M| = |\calU|\) and \(\alpha = \frac{\log\pth{|\calV|-n}}{\log n}\) if \(|\calV|>n\). Each of the following conditions guarantee that any estimator makes at least $\Omega(n^{1-\beta})$ errors with probability $1-o(1)$:
    \bcomment{
    \begin{tabular}{|r l c|}
        \hline
        Necessary cond. & Range of \(\beta\) & Boundary\\
        \hline
        \(\zeta \leq \pth{1+2\sqrt{\beta(1-\beta)}}\log n - \omega(\log \log n)\) & \(0<\beta\leq\min\{1-\alpha,1/2\}\) & elliptic\\
        \(\zeta \leq \pth{\sqrt{\alpha}+\sqrt{\beta}}^2 \log n - \omega(\log\log n)\) & \(1-\alpha<\beta\) & parabolic\\
        \(\zeta \leq 2 \log n - \omega(\log \log n) \) & \(1/2<\beta\) & vertical\\
        \hline
    \end{tabular}
    }
    \begin{tabular}{l c}
      \(\zeta \leq \pth{1+2\sqrt{\beta(1-\beta)}}\log n - \omega(\log \log n)\) & \(0 < \beta \leq \frac{1}{2}\)\\
      \(\zeta \leq \pth{\sqrt{\alpha}+\sqrt{\beta}}^2 \log n - \omega(\log\log n)\) & \(0 < \beta \leq 1\)\\
      \(\zeta \leq 2 \log n - \omega(\log \log n) \) & \(\frac{1}{2} \leq \beta \leq 1\).
    \end{tabular}
\end{theorem}
\begin{proof}
Assume \(\zeta \leq \pth{1+2\sqrt{\beta(1-\beta)}}\log n - \omega(\log \log n)\).
Because $\beta \leq \frac{1}{2}$, \(\zeta \leq 2 \log n - \omega(\log \log n)\), and from Lemma~\ref{lemma:ellipse-coverse} \[
\frac{\log \log \Pr[\hat{M} \neq M]^{-1}}{\log n} \geq \frac{1}{2} + \sqrt{\frac{\zeta}{2\log n}\pth{1 - O\pth{\frac{\log \log n}{\log n}} - \frac{\zeta}{2\log n}}} - O\pth{\frac{\log \log n}{\log n}}.
\]
If \(\beta \geq \omega\pth{\frac{\log \log n}{\log n}}\), then
\begin{align*}
\log \log \Pr[\hat{M} \neq M]^{-1} 
&\geq \frac{\log n}{2} + \frac{1}{2}\sqrt{\zeta\pth{2 \log n - O\pth{\log \log n} - \zeta}} - O\pth{\log \log n}\\
&\geq (1-\beta)\log n + \omega(\log \log n)
\end{align*}
and the first part of the theorem follows from Lemma~\ref{lemma:distance-converse}.

If \(\zeta \leq 2 \log n - \omega(\log \log n)\), 
\begin{align*}
\log \log \Pr[\hat{M} \neq M]^{-1} 
&\geq \frac{\log n}{2} + \frac{1}{2}\sqrt{\zeta\pth{2 \log n - O\pth{\log \log n} - \zeta}} - O\pth{\log \log n}\\
&\geq \log n + \omega(\log \log n)
\end{align*}
and the third part of the theorem again follows from Lemma~\ref{lemma:distance-converse}.

Similarly, the second part of the theorem follows from Lemma~\ref{lemma:parabola-converse} and Lemma~\ref{lemma:distance-converse}.
\end{proof}

\section{Combinatorial analysis}

\subsection{Elementary misalignments between mappings}
\label{subsec:elementary}

Let \(\mmat,\mmat'\in\{0,1\}^{\calU\times\calV}\) be the matrix encodings of mappings \(m,m'\) and and let \(\Gmat\in\bbR^{\calU\times\calV}\) denote the score matrix for databases \(\Avec,\Bvec\). As shown in \hyperref[sec:algo]{Section \ref*{sec:algo}}, comparing the likelihoods of \(\Avec,\Bvec\) being generated by \(m\) versus \(m'\) is equivalent to comparing the values of \(\ip{\Gmat,\mmat}\) and \(\ip{\Gmat,\mmat'}\). 

Assume \(\mmat-\mmat'\) can be written in block diagonal form \(\crocMat{\Delta\mmat_1}{\zeromat}{\zeromat}{\Delta\mmat_2}\). Let \(\mmat_1' \define \mmat - \crocMat{\Delta\mmat_1}{\zeromat}{\zeromat}{\zeromat}\) and \(\mmat_2' \define \mmat - \crocMat{\zeromat}{\zeromat}{\zeromat}{\Delta\mmat_2}\). \(m_1',m_2'\) are two valid mappings that in some sense partition the disagreement between \(m\) and \(m'\) into two.
\begin{align*}
    \ip{\Gmat,\mmat} < \ip{\Gmat,\mmat'} \iff \ip{\Gmat,\mmat-\mmat'} &< 0\\
    \iff \ip{\Gmat,\crocMat{\Delta\mmat_1}{\zeromat}{\zeromat}{\zeromat}} + \ip{\Gmat,\crocMat{\zeromat}{\zeromat}{\zeromat}{\Delta\mmat_2}} <0\\
    \iff \ip{\Gmat,\mmat-\mmat_2'} + \ip{\Gmat,\mmat-\mmat_1'} < 0\\
    \implies \ip{\Gmat,\mmat-\mmat_2'} \leq 0 \textrm{ or } \ip{\Gmat,\mmat-\mmat_1'} \leq 0
\end{align*}
Then \(m'\) has higher score than \(m\) only if at least one of \(m_1',m_2'\) also has higher score than \(m\). Furthermore, \(m'\) is the minimizer for the inner product \(\ip{\Gmat,\mmat-\mmat'}\) only if both \(\ip{\Gmat,\mmat-\mmat_1'}\) and \(\ip{\Gmat,\mmat-\mmat_2'}\) are negative. So \(m'\) is the optimal mapping only if each of its `submappings' (i.e. mappings whose mismatch with \(m\) are entirely contained in \(m'\)) have higher score than \(m\).

It is then of interest to define elementary misalignments between mappings.
\begin{definition}
    \label{def:elementaryBlocks}
    Let \(m_1,m_2\) be a pair of mappings between \(\calU\) and \(\calV\) that are bijective between from their domain to their co-domain. Let \(\mmat_1,\mmat_2\in\{0,1\}^{\calU\times\calV}\) be binary matrices that encode these mappings. We say the mismatch between the two mappings is elementary if and only if \(\mmat_1-\mmat_2\) does not have a block-diagonal representation with multiple zero-sum, non-zero blocks.
\end{definition}

\begin{figure}
    \centering
    \begin{minipage}{0.3\textwidth}
        \begin{tikzpicture}
            \draw (-.2,0) node {\(u_1\)};
            \draw (-.2,-1) node {\(u_2\)};
            \draw (-.2,-2) node {\(u_3\)};
            \draw (-.2,-3) node {\(u'\)};
            
            \draw (2.2,0) node {\(v_1\)};
            \draw (2.2,-1) node {\(v_2\)};
            \draw (2.2,-2) node {\(v_3\)};
            \draw (2.2,-3) node {\(v'\)};
            
            \draw[dotted, rounded corners = 5] (-.5,0.3) -- (-.5,-3.3) -- (.1,-3.3) -- (.1,0.3) -- cycle;
            \draw[dotted, rounded corners = 5] (1.9,0.3) -- (1.9,-3.3) -- (2.5,-3.3) -- (2.5,0.3) -- cycle;
            
            \draw[thick,blue,-] (0,0) -- (2,0);
            \draw[thick,blue,-] (0,-1) -- (2,-1);
            \draw[thick,blue,-] (0,-2) -- (2,-2);
            
            \draw[thick,red,-] (0,0) -- (2,-1);
            \draw[thick,red,-] (0,-1) -- (2,-2);
            \draw[thick,red,-] (0,-2) -- (2,0);
            
            \draw (1,-3.2) node {I};
        \end{tikzpicture}
    \end{minipage}
    \begin{minipage}{0.3\textwidth}
        \begin{tikzpicture}
            \draw (-.2,0) node {\(u_1\)};
            \draw (-.2,-1) node {\(u_2\)};
            \draw (-.2,-2) node {\(u_3\)};
            \draw (-.2,-3) node {\(u'\)};
            
            \draw (2.2,0) node {\(v_1\)};
            \draw (2.2,-1) node {\(v_2\)};
            \draw (2.2,-2) node {\(v_3\)};
            \draw (2.2,-3) node {\(v'\)};
            
            \draw[dotted, rounded corners = 5] (-.5,0.3) -- (-.5,-3.3) -- (.1,-3.3) -- (.1,0.3) -- cycle;
            \draw[dotted, rounded corners = 5] (1.9,0.3) -- (1.9,-3.3) -- (2.5,-3.3) -- (2.5,0.3) -- cycle;
            
            \draw[thick,blue,-] (0,0) -- (2,0);
            \draw[thick,blue,-] (0,-1) -- (2,-1);
            \draw[thick,blue,-] (0,-2) -- (2,-2);
            
            \draw[thick,red,-] (0,0) -- (2,-1);
            \draw[thick,red,-] (0,-1) -- (2,-2);
            \draw[thick,red,-] (0,-2) -- (2,-3);
            
            \draw (1,-3.2) node {II};
        \end{tikzpicture}
    \end{minipage}
    \begin{minipage}{0.3\textwidth}
        \begin{tikzpicture}
            \draw (-.2,0) node {\(u_1\)};
            \draw (-.2,-1) node {\(u_2\)};
            \draw (-.2,-2) node {\(u_3\)};
            \draw (-.2,-3) node {\(u'\)};
            
            \draw (2.2,0) node {\(v_1\)};
            \draw (2.2,-1) node {\(v_2\)};
            \draw (2.2,-2) node {\(v_3\)};
            \draw (2.2,-3) node {\(v'\)};
            
            \draw[dotted, rounded corners = 5] (-.5,0.3) -- (-.5,-3.3) -- (.1,-3.3) -- (.1,0.3) -- cycle;
            \draw[dotted, rounded corners = 5] (1.9,0.3) -- (1.9,-3.3) -- (2.5,-3.3) -- (2.5,0.3) -- cycle;
            
            \draw[thick,blue,-] (0,0) -- (2,0);
            \draw[thick,blue,-] (0,-1) -- (2,-1);
            \draw[thick,blue,-] (0,-2) -- (2,-2);
            
            \draw[thick,red,-] (0,0) -- (2,-1);
            \draw[thick,red,-] (0,-3) -- (2,-2);
            \draw[thick,red,-] (0,-2) -- (2,-3);
            
            \draw (1,-3.2) node {III};
        \end{tikzpicture}
    \end{minipage}
    \caption{Examples the 3 types of elementary mismatches (of size 3) as bigraphs: Cycle (I), even path (II), pair of odd paths (III)}
    \label{fig:decomposition1}
\end{figure}
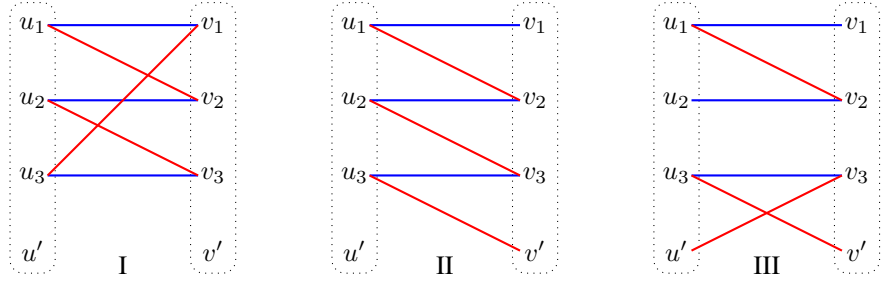

\begin{figure}
    \centering
    \begin{minipage}{0.3\textwidth}
        \begin{tikzpicture}
            \matrix[matrix of math nodes,left delimiter={[},right delimiter={]}]{ 
            \blue{\textbf{+}} & \red{\textbf{--}} & 0 & 0\\
            0 & \blue{\textbf{+}} & \red{\textbf{--}} & 0 \\
            \red{\textbf{--}} & 0 & \blue{\textbf{+}} & 0 \\
            0 & 0 & 0 & 0 \\
            };
            \draw (0,-1.3) node {I};
        \end{tikzpicture}
    \end{minipage}
    \begin{minipage}{0.3\textwidth}
        \begin{tikzpicture}
            \matrix[matrix of math nodes,left delimiter={[},right delimiter={]}]{ 
            \blue{\textbf{+}} & \red{\textbf{--}} & 0 & 0\\
            0 & \blue{\textbf{+}} & \red{\textbf{--}} & 0 \\
            0 & 0 & \blue{\textbf{+}} & \red{\textbf{--}} \\
            0 & 0 & 0 & 0 \\
            };
            \draw (0,-1.3) node {II};
        \end{tikzpicture}
    \end{minipage}
    \begin{minipage}{0.3\textwidth}
        \begin{tikzpicture}
            \matrix[matrix of math nodes,left delimiter={[},right delimiter={]}]{ 
            \blue{\textbf{+}} & \red{\textbf{--}} & 0 & 0 \\
            0 & \blue{\textbf{+}} & 0 & 0 \\
            0 & 0 & \blue{\textbf{+}} & \red{\textbf{--}} \\
            0 & 0 & \red{\textbf{--}} & 0 \\
            };
            \draw (0,-1.3) node {III};
        \end{tikzpicture}
    \end{minipage}
    \caption{Examples the 3 types of elementary misalignmend (of size 3) as matrices: Cycle (I), even path (II), pair of odd paths (III)}
    \label{fig:decomposition2}
\end{figure}

There are three types of elementary misalignments, as shown in \hyperref[fig:decomposition1]{Fig. \ref*{fig:decomposition1}} and \hyperref[fig:decomposition2]{Fig. \ref*{fig:decomposition2}}. These are
\begin{compactitem}
    \item I - Cycles: The two mappings use the same set of users from \(\calU\) and \(\calV\) but pair them up differently. This type of mismatch consists of a single cycle.
    \item II - Even paths: The two mappings use the same set of users from one of the sets (say \(\calU\)) but differ in the users they map from the other side (\(\calV\)) by 1 user. This type of mismatch consists of one path and is cycle-free.
    \item III - Pair of odd paths: The two mappings differ in the users they map on both sides, by 1 user per side. This type of mismatch consists of two paths and is cycle-free.
\end{compactitem}

The bigraph representation in \hyperref[fig:decomposition1]{Fig. \ref*{fig:decomposition1}} can be used to explain why these three are the only types of elementary misalignments. Since \(m_1\) and \(m_2\) map each user at most once, each vertex can have at most one edge from each mapping and a degree of at most 2. Then the bigraph has alternating edges and maximum degree 2. Graphs of maximum degree 2 decompose into cycles and paths. Each component in the bigraph corresponds to a block in the adjacency matrix.

Since edges are alternating between the two mappings, each cycle has even length and an equal number of edges coming from both graphs. Then each cycle corresponds to a block in \(\mmat_1-\mmat_2\) with sum of entries equal to 0. Therefore each cycle is an elementary misalignment. The same holds for even paths.

Odd paths contain one more edge from one mapping than from the other. Therefore these correspond to blocks in \(\mmat_1-\mmat_2\) whose sum equal \(1\) or \(-1\). Since \(\mmat_1-\mmat_2\) has sum of entries equal to zero, it follows that there must be an equal number of blocks whose entries sum up to \(1\) and blocks whose entries sum up to \(-1\). Pairing these up gives us elementary blocks.

\begin{lemma}
    \label{lemma:countElementary}
    Let \(\calU\) and \(\calV\) sets of users of size \(n\) and \(n+s\) respectively. Let \(m\) be the true mapping of size \(n\). Consider the elementary misalignments induced by all mappings \(m'\) of size \(n\).\\
    The number of distinct elementary type-I misalignments of size \(\delta\) is upper bounded by \(\frac{n^\delta}{\delta}\) if \(\delta\in\{2,3,\cdots\}\) and 0 if \(\delta=1\). The number of distinct elementary type-I misalignments of size \(\delta\) is upper bounded by \(sn^\delta\). There are be no elementary misalignments of type III.
\end{lemma}
\begin{proof}
    We count the ways to pick some \(m'\) that induces an elementary misalignment with \(m\) of size \(\delta\).
    
    There are \(\binom{n}{\delta}\) ways to pick the \(\delta\) pairs from \(m\) to be misaligned by \(m'\). Let us denote the sets of these users as \(\calU'\subseteq\calU\) and \(\calV'\subseteq\calV\). \(|\calU'|=|\calV'|=\delta\). Assume these sets are fixed.
    \begin{compactitem}
        \item If \(\delta=1\), there is no way to obtain a type I mismatch, since the only way to pair the single user in \(\calU'\) to the single user in \(\calV'\) is the same as the original mapping in \(m\).\\
        For \(\delta\in\{2,3,\cdots\}\), there are \((\delta-1)!\) ways to pair \(\calU'\) and \(\calV'\) to obtain a type I mismatch. (Forming the `cycle' in \hyperref[fig:decomposition1]{Fig. \ref*{fig:decomposition1}} is simply a matter of  arranging the blue edges around the cycle, which results in a unique way to pick the red edges.)\\
        Then, in total, there are \(\binom{n}{\delta}(\delta-1)! = \frac{1}{\delta}\binom{n}{\delta}\delta!\) ways to pick a type I mistmatch.
        \item There are \(s\) ways to pick a user \(w'\) from \(\calV\) that is not mapped by \(m\).\\
        Given this user, there are \(\delta!\) ways to pair \(\calU',\calV'\) and \(w'\), leaving one user from either \(\calU'\) unpaired. We can generate this pairing as follows: Take any of the \((\delta-1)!\) type I matchings. Break the cycle at any of the \(\delta\) red edges and connect that edge to \(w'\) from the appropriate side. This gives us an even path.\\
        Then, in total, there are \(s\binom{n}{\delta}\delta!\) ways to pick a type II mistmatch.
        \item We only consider mappings \(m'\) of size \(n\), which is the same as the true mapping \(m\). So, given the representation in \hyperref[fig:decomposition1]{Fig. \ref*{fig:decomposition1}}, there must be an equal number of red and blue edges.\\
        Odd paths have an extra edge of either color. To construct an odd path with more edges belonging to \(m'\), there need to vertices in both \(\calU\) and \(\calV\) not covered by \(m\). (These correspond to vertices \(u'\) and \(v'\) in \hyperref[fig:decomposition1]{Fig. \ref*{fig:decomposition1}}.) Since \(\calU\) has size equal to that of \(m\), all vertices in \(\calU\) are covered, and there can be no odd path with more edges from \(m'\).\\
        Since the total number of edges from each mapping needs to be equal, it then follows that there can also be no odd path with more edges from \(m\).
    \end{compactitem}
    
    Using the fact that \(\binom{n}{\delta}\delta! \leq n^\delta\), we simplify the expression to get the result.
\end{proof}

\begin{lemma}
    \label{lemma:countMisalignment}
    Let \(\calU\) and \(\calV\) sets of users of size \(n\) and \(n+s\) respectively. Let \(m\) be the true mapping of size \(n\). Let \(c\in(0,\infty)\) some arbitrary constant. The number of different mappings \(m'\) that result in a misalignment of size \(\delta\) is upper bounded by:
    \begin{compactitem}
        \item \(\exp\bpth{\delta\pth{1+\log n + \log(1+1/c)}}\) if \(\delta\geq cs\), and
        \item \(\exp\bpth{\delta\pth{1+\log \frac{ns}{\delta} + \log(1+c)}}\) if \(\delta\leq cs\).
    \end{compactitem}
\end{lemma}
\begin{proof}
    We count the number of different ways to construct \(m'\) that results in a misalignment of size \(\delta\).

    There are \(\binom{n}{\delta}\) different ways to pick the set of vertices to be misaligned by \(m'\).

    Given the \(\delta\) pairs of vertices to be misaligned, there are no more than \(\delta+s\) ways to misalign each vertex. So, there are no more than \((\delta+s)^\delta\) ways to misalign the set of \(\delta\) pairs.

    \(\binom{n}{\delta}\) is strictly less than \(\pth{\frac{en}{\delta}}^\delta\). If \(\delta\geq cs\), then \((\delta+s)^\delta\) is at most \((1+1/c)^\delta \delta^\delta\). If \(\delta\leq cs\), then \((\delta+s)^\delta\) is at most \((1+c)^\delta s^\delta\). The products of these terms give us the claimed results.
\end{proof}

\section{Concentration inequalities}

We use \(m,m',m_1\) etc to denote mappings between \(\calU\) and \(\calV\), and \(\mmat,\mmat',\mmat_1\in\{0,1\}^{\calU\times\calV}\) to denote binary matrix representations of these mappings, where \((\mmat)_{u,v}=1\) if and only if \(m\) maps \(u\) to \(v\).\\

\subsection{Concentration inequalities for database alignment}

\(\Gmat\) refers to the information density matrix under the database alignment setting as described in \hyperref[subsec:modelDatabase]{Subsection \ref*{subsec:modelDatabase}}. Specifically, \(\Gmat\) be the matrix such that \(G_{u,v}\) is the log-likelihood ratio of hypotheses \(u\mapped{M}v\) vs. \(u\notmapped{M}v\) for any \((u,v)\in\calU\times\calV\).

\begin{lemma}[Atypicality]
    \label{lemma:typicality}
    Given some \(|\tau|\leq I_{XY}\) and \(m\) a partial mapping fully contained in the true mapping \(M\), \(\Pr\croc{\ip{\Gmat-\tau,\mmat}\leq 0|m\subseteq M} \leq \exp\pth{-|m|\cdot\frac{(I_{XY}-\tau)^2}{4I_{XY}}}\).
\end{lemma}

\begin{proof}
    The atypicality event is completely independent from users that are not contained in \(m\). Then, without loss of generality, we can assume \(m=M\) instead of \(m\subseteq M\).\\
    By \hyperref[cor:chernoff]{Corollary \ref*{cor:chernoff}},
    \begin{align*}
        \Pr\croc{\tau|m| \geq \ip{\Gmat,\mmat}\big|M=m} \leq \exp\pth{\theta\tau|m|}R\pth{(1-\theta)\mmat}.
    \end{align*}
    By \hyperref[cor:rMapping]{Corollary \ref*{cor:rMapping}}, this last expression is upper bounded by \(\exp\bpth{\theta\tau|m|-\theta(1-\theta)I_{XY}|m|}\).
    
    Let \(\theta = \frac{I_{XY}-\tau}{2I_{XY}}\). Then \(\tau\theta = \frac{I_{XY}\tau-\tau^2}{2I_{XY}}\) and \(\theta(1-\theta)I_{XY} = \frac{I_{XY}^2-\tau^2}{4I_{XY}}\), which gives us
    \begin{align*}
        &\Pr\croc{\tau|m| \geq \ip{\Gmat,\mmat}\big|M=m} \\
        &\leq \exp\bpth{\tau\theta|m|-\theta(1-\theta)I_{XY}|m|}\\
        &= \exp\bpth{-|m|\cdot\frac{I_{XY}^2-2I_{XY}\tau+\tau^2}{4I_{XY}}}
    \end{align*}
    which matches the claimed result.
\end{proof}

\begin{lemma}[False positive probability]
    \label{lemma:FP}
    Given some \(|\tau|\leq I_{XY}\), \(\Pr\croc{G_{u,v}\geq\tau|u\notmapped{M}v} \leq \exp\pth{-\frac{(I_{XY}+\tau)^2}{4I_{XY}}}\).
\end{lemma}

\begin{proof}
    \begin{align*}
        \Pr\croc{G_{u,v}\geq \tau |u\notmapped{M}v} &\leq e^{-\tau\theta}\bbE\croc{\exp\pth{\theta G_{u,v}}|u\notmapped{M}v}
    \end{align*}
    
    \(G_{u,v}\) only depends on \(\Avec(u)\) and \(\Bvec(v)\). So, without loss of generality, we can assume \(\calU = \{u\}\) and \(\calV=\{v\}\). Since \(u\notmapped{M}v\), it follows that \(M\) maps nothing and the databases are independent. Let \(m\) denote the empty mapping and \(m'\) denote the mapping consisting of \((u,v)\). Then \(\theta G_{u,v} = \ip{\Gmat,\theta\mmat'} = \ip{\Gmat,\theta\mmat' + (1-\theta)\mmat}\). It then follows that
    \begin{align*}
        \Pr\croc{G_{u,v}\geq \tau |u\notmapped{M}v} &\leq e^{-\tau\theta}\bbE\croc{\exp\pth{\theta G_{u,v}}|u\notmapped{M}v}\\
        &= e^{-\tau\theta}\bbE\croc{\exp\pth{\ip{\Gmat,\theta\mmat' -\mmat}}|M=m}\\
        &= e^{-\tau\theta}R(\theta\mmat')
    \end{align*}
    where the last line follows from \hyperref[lemma:genFunc2exp]{Lemma \ref*{lemma:genFunc2exp}}. \(R(\theta\mmat') = R([\theta])\). By \hyperref[cor:rMapping]{Corollary \ref*{cor:rMapping}}, \(R([\theta]) \leq \exp\pth{-\theta(1-\theta)I_{XY}}\).
    
    \begin{align*}
        &\Pr\croc{G_{u,v}\geq \tau |u\notmapped{M}v}\\
        &\leq \exp\pth{-\tau\theta - -\theta(1-\theta)I_{XY}}\\
        &= \exp\pth{-\frac{I_{XY}^2+2I_{XY}\tau+\tau^2}{4I_{XY}}}
    \end{align*}
    which matches the claimed result.
\end{proof}

\begin{lemma}[Misalignment]
    \label{lemma:misalignment}
    Let \(m\) and \(m'\) denote two mappings of same size and \(\delta\) denote the number of pairs mapped by \(m\) but not by \(m'\). Then \(\Pr[\ip{\Gmat,\mmat}\leq\ip{\Gmat,\mmat'}|M=m] \leq \exp\pth{-\frac{\delta}{2}I_{XY}}\).
\end{lemma}

\begin{proof}
    By \hyperref[cor:chernoff]{Corollary \ref*{cor:chernoff}}, \(\Pr[\ip{\Gmat,\mmat}\leq\ip{\Gmat,\mmat'}|M=m]\) is upper bounded by \(R(\Thetamat)\) where \(\Thetamat = \frac{1}{2}\mmat+\frac{1}{2}\mmat'\). \(\Thetamat\) can be represented in block-diagonal form. For each \(n-\delta\) pair of users \((u,v)\) that is mapped both by \(\mmat\) and \(\mmat'\), we get a \(1\times1\) block \([1]\). The remaining \(\delta\) pairs whose mapping is not the same between \(\mmat\) and \(\mmat'\), we get blocks \(\Thetamat_j\) that correspond to cycles or even paths as described in \hyperref[def:elementaryBlocks]{Definition \ref*{def:elementaryBlocks}}.\\
    Given this block-diagonal form, by \hyperref[lemma:blockGenFunc]{Lemma \ref*{lemma:blockGenFunc}}, \(R(\Thetamat)\) is equal to the product \(\croc{R([1])}^{n-\delta}\cdot\prod_j R(\Thetamat_j)\). By \hyperref[lemma:rOne2One]{Lemma \ref*{lemma:rOne2One}}, \(R([1])=1\). By Lemmas \ref{lemma:rCycle} and \ref{lemma:rEvenPath}, plugging in \(\nu=1\), we have
    \(R(\Thetamat_j)\leq \exp\pth{-\frac{\delta_j I_{XY}}{2}}\) where \(\delta_j\) is the total length (i.e. number of blue edges or number of red edges) of the corresponding cycle or even paths. The total number of user pairs that whose mapping differs between \(\mmat\) and \(\mmat'\) is \(\delta\). Then \(\prod_j R(\Thetamat_j) \leq \exp\pth{\sum_j-\frac{\delta_j I_{XY}}{2}} = \exp\pth{-\frac{\delta I_{XY}}{2}}\).\\
    It then follows that \(\Pr[\ip{\Gmat,\mmat}\leq\ip{\Gmat,\mmat'}|M=m] \leq R(\Thetamat) \leq \exp\pth{-\frac{\delta I_{XY}}{2}}\).
\end{proof}

\begin{lemma}[Misalignment-despite-typicality]
    \label{lemma:condMisalignment}
    Let \(m\) and \(m'\) denote two mappings of same size \(\delta\) such that no pair is mapped under both mappings. Given \(0\leq\tau\leq I_{XY}\),
    \begin{align*}
        \Pr\croc{\ip{\Gmat,\mmat}\geq \tau |m| \,\textrm{ and }\, \ip{\Gmat,\mmat'}\geq\ip{\Gmat,\mmat}\big|m\subseteq M} \leq \exp\pth{-\delta\cdot\frac{I_{XY}^2+\tau^2}{2I_{XY}}+\delta\cdot6\rho_{\max}^2\tau}
    \end{align*}
    where \(\rho_{\max} = \max |\rho_i|\), the largest correlation coefficient under the canonical form.
\end{lemma}

By \hyperref[lemma:svd2rho]{Lemma \ref*{lemma:svd2rho}}, under \hyperref[cond:highDimensional]{Condition \ref*{cond:highDimensional}}, \(\rho_{\max}^2 \leq o(1)\), so the bound can be simplified as \(\exp\pth{-\delta\cdot\frac{I_{XY}^2+\tau^2}{2I_{XY}}(1-o(1))}\).

\begin{proof}
    By \hyperref[cor:chernoff]{Corollary \ref*{cor:chernoff}}, \(\Pr\croc{\ip{\Gmat,\mmat}\geq \tau |m| \,\textrm{ and }\, \ip{\Gmat,\mmat'}\geq\ip{\Gmat,\mmat}\big|m\subseteq M}\) is upper bounded by \(e^{-\tau|m|(\nu-1)}R\pth{\Thetamat}\) where \(\Thetamat = \nu(1-\theta)\mmat + \nu\theta\mmat'\). Consider the decomposition of \(\Thetamat\) into blocks: we get cycles and even paths as described in \hyperref[def:elementaryBlocks]{Definition \ref*{def:elementaryBlocks}}. (There are no one-by-one blocks since \(m\) and \(m'\) have no intersection.) This decomposition gives us a block-diagonal representation of \(\Thetamat\). By \hyperref[lemma:blockGenFunc]{Lemma \ref*{lemma:blockGenFunc}}, \(R(\Thetamat) = \prod R(\Thetamat_j)\), where \(\Thetamat_j\) denotes the block corresponding to an elementary misalignment (i.e. cycle or even path).
    
    Let \(m_j\) and \(m_j'\) denote the partial misalignments that correspond to the intersection of block \(\Thetamat_j\) with \(m\) and \(m'\). Let \(\delta_j=|m_j|=|m_j'|\) denote their size. (Cycle and even path type misalignment consist of mappings of equal size.) Under \hyperref[cond:highDimensional]{Condition \ref*{cond:highDimensional}}, Lemmas \ref{lemma:rCycle} and \ref{lemma:rEvenPath} give us
    \begin{align*}
        R\pth{\nu(1-\theta)\mmat_j + \nu\theta\mmat_j'} &\leq \exp\pth{-\delta_j\cdot\frac{I_{XY}}{2}\nu(2-\nu) + 6\delta_j\rho_{\max}^2 I_{XY}}\\
        \implies R\pth{\nu(1-\theta)\mmat + \nu\theta\mmat'} &\leq \exp\pth{-\delta\cdot\frac{I_{XY}}{2}\nu(2-\nu)+6\delta \rho_{\max}^2 I_{XY}}
    \end{align*}
    Pick \(\nu = 1 + \frac{\tau}{I_{XY}}\). Then
    \begin{align*}
        R\pth{\nu(1-\theta)\mmat + \nu\theta\mmat} &\leq \exp\pth{-\delta\cdot\frac{I_{XY}^2-\tau^2}{2I_{XY}}+\delta\cdot 6\rho_{\max}^2\tau}
    \end{align*}
    and \(e^{-\tau\delta(\nu-1)} = \exp\pth{-\delta\cdot\frac{\tau^2}{I_{XY}}}\). Then \(e^{-\tau|m|(\nu-1)}R\pth{\Thetamat} \leq \exp\pth{-\delta\cdot\frac{I_{XY}^2+\tau^2}{2I_{XY}}+\delta\cdot 6\rho_{\max}^2\tau}\) and we have the claimed result.
\end{proof}

\subsection{Concentration inequalities for planted matching}

\(\Wmat\) refers to the edge weight matrix of the bipartite graph under the planted matching setting described in \hyperref[subsec:modelPlanted]{Subsection \ref*{subsec:modelPlanted}}. We state the concentration inequalities in terms of \(\Wmat\), as well as another matrix \(\Wmat_G\define\mu\Wmat - \frac{\mu^2}{2}\), which is \(\Wmat\) scaled and shifted to match the statistics of \(\Gmat\), as given in \hyperref[app:stats]{Appendix \ref*{app:stats}}.

\begin{lemma}[Atypicality]
    \label{lemma:typicalityPlanted}
    Let \(m\) a partial mapping fully contained in the true mapping \(M\). Given some \(\tau_W \leq \mu\) and \(\tau_G = \mu\tau - \frac{\mu^2}{2}\),
    \begin{align*}
        \Pr\croc{\tau_W|m|\geq \ip{\Wmat,\mmat}|m\subseteq M} &\leq \exp\pth{-\frac{(\mu-\tau_W)^2}{2}}
    \end{align*}
    which is equivalent to
    \begin{align*}
        \Pr\croc{\tau_G|m|\geq \ip{\Wmat_G,\mmat}|m\subseteq M} &\leq \exp\pth{-\frac{(\zeta-\tau_G)^2}{4\zeta}}
    \end{align*}
    where \(\zeta = \mu^2/2\).
\end{lemma}
\begin{proof}
    Given \(u\mapped{M}v\), \(W_{u,v}\) is normal with mean \(\mu\) and unit variance. Then, its moment generating function is given by \(\bbE\croc{e^{\theta W_{u,v}}}=e^{\theta \mu + \frac{\theta^2}{2}}\). By Markov's inequality, for any \(\theta < 0\),
    \begin{align*}
        \Pr\croc{W_{u,v} \leq \tau_W|u\mapped{M}v} &\leq e^{-\theta\tau_W}\bbE\croc{e^{\theta W_{u,v}}}\\
        &= \exp\pth{-\theta \tau_W +\theta\mu + \frac{\theta^2}{2}}
    \end{align*}
    Pick \(\theta = -(\mu - \tau_W)\). Then \(\Pr\croc{W_{u,v} \leq \tau_W|u\mapped{M}v} \leq \exp\pth{-\frac{(\mu-\tau_W)^2}{2}}\). Since all entries in \(\Wmat\) are independent, \(\Pr\croc{\ip{\Wmat,\mmat} \leq \tau_W|m|\,\big|m\subseteq M}\) is the product of all of these terms.

    \(\frac{(\mu-\tau_W)^2}{2} = \frac{\pth{\mu^2 - \mu\tau_W}^2}{2\mu^2} = \frac{\pth{\mu^2/2 - (\mu\tau_W-\mu^2/2)}^2}{2\mu^2} = \frac{(\zeta - \tau_G)^2}{4\zeta}\). Then \(\Pr\croc{(W_G)_{u,v} \leq \tau_G|u\mapped{M}v} = \Pr\croc{W_{u,v} \leq \tau_W|u\mapped{M}v} \leq \exp\pth{-\frac{(\zeta - \tau_G)^2}{4\zeta}}\). Once again, taking the product of all of this term over all pairs in \(m\) gives us the claimed result.
\end{proof}

\begin{lemma}[False positive probability]
    \label{lemma:FPPlanted}
    Given some \(\tau_W \geq 0\) and \(\tau_G = \mu\tau - \frac{\mu^2}{2}\),
    \begin{align*}
        \Pr\croc{W_{u,v} \geq \tau_W|u\notmapped{M}v} &\leq \exp\pth{-\frac{\tau_W^2}{2}}
    \end{align*}
    which is equivalent to
    \begin{align*}
        \Pr\croc{(W_G)_{u,v} \leq \tau_G|u\notmapped{M}v} &\leq \exp\pth{-\frac{(\zeta + \tau_G)^2}{4\zeta}}
    \end{align*}
    where \(\zeta = \mu^2/2\).
\end{lemma}
\begin{proof}
    Given \(u\notmapped{M}v\), \(W_{u,v}\) is normal with zero \(\mu\) and unit variance. Then, its moment generating function is given by \(\bbE\croc{e^{\theta W_{u,v}}}=e^{\frac{\theta^2}{2}}\). By Markov's inequality, for any \(\theta > 0\),
    \begin{align*}
        \Pr\croc{W_{u,v} \leq \tau_W|u\notmapped{M}v} &\leq e^{-\theta\tau_W}\bbE\croc{e^{\theta W_{u,v}}}\\
        &= \exp\pth{-\theta \tau_W + \frac{\theta^2}{2}}
    \end{align*}
    Pick \(\theta = \tau_W\). Then \(\Pr\croc{W_{u,v} \geq \tau_W|u\notmapped{M}v} \leq \exp\pth{-\frac{\tau_W^2}{2}}\).

    \(\frac{\tau_W^2}{2} = \frac{(\mu\tau_W)^2}{2\mu^2} = \frac{\pth{\mu^2/2 + (\mu\tau_W-\mu^2/2)}^2}{2\mu^2} = \frac{(\zeta + \tau_G)^2}{4\zeta}\). Then \(\Pr\croc{(W_G)_{u,v} \leq \tau_G|u\notmapped{M}v} = \Pr\croc{W_{u,v} \leq \tau_W|u\notmapped{M}v} \leq \exp\pth{-\frac{(\zeta + \tau_G)^2}{4\zeta}}\).
\end{proof}

\begin{lemma}[Misalignment]
    \label{lemma:misalignmentPlanted}
    Let \(m\) and \(m'\) denote two mappings of same size and \(\delta\) denote the number of pairs mapped by \(m\) but not by \(m'\). Then \( \leq \exp\pth{-\delta\frac{\mu^2}{4}}\).
    \begin{align*}
        \Pr\croc{\ip{\Wmat,\mmat}\leq\ip{\Wmat,\mmat'}|M=m} &\leq \exp\pth{-\delta\frac{\mu^2}{4}}
    \end{align*}
    which is equivalent to 
    \begin{align*}
        \Pr\croc{\ip{\Wmat_G,\mmat}\leq\ip{\Wmat_G,\mmat'}|M=m} &\leq \exp\pth{-\delta\frac{\zeta}{2}}
    \end{align*}
    where \(\zeta = \mu^2/2\).
\end{lemma}
\begin{proof}
    \(\ip{\Wmat,\mmat}-\ip{\Wmat,\mmat'}\) is the linear combination of \(2\delta\) independent Gaussian random variables and therefore is Gaussian. (\(|m|-\delta\) of the terms in \(\ip{\Wmat,\mmat}\) get canceled out by the \(|m|-\delta\) common terms in \(\ip{\Wmat,\mmat'}\).) The difference has mean \(\delta \mu\) and variance \(2\delta\). Then, the moment generating function is given by \(\bbE\croc{\exp\pth{\theta \ip{\Wmat,\mmat}-\theta\ip{\Wmat,\mmat'}}} = \exp\pth{\theta \delta \mu + \theta^2 \delta}\). Then, by Markov's inequality, for any \(\theta <0\),
    \begin{align*}
        \Pr\croc{\ip{\Wmat,\mmat}\leq\ip{\Wmat,\mmat'}|M=m} &= \Pr\croc{\ip{\Wmat,\mmat}-\ip{\Wmat,\mmat'}\leq0|M=m}\\
        &\leq \bbE\croc{\exp\pth{\theta \ip{\Wmat,\mmat}-\theta\ip{\Wmat,\mmat'}}}\\
        &= \exp\pth{\theta \delta \mu + \theta^2 \delta}.
    \end{align*}
    Picking \(\theta = -\mu/2\) gives us the claimed result.
\end{proof}

\begin{lemma}[Misalignmen-despite-typicality]
    \label{lemma:condMisalignmentPlanted}
    Let \(m\) and \(m'\) denote two mappings of same size \(\delta\) such that no pair is mapped under both mappings. Given some \(\tau_W\geq \mu/2\) and \(\tau_G = \mu\tau - \frac{\mu^2}{2}\),
    \begin{align*}
        \Pr\croc{\ip{\Wmat,\mmat}\geq \tau_W |m| \,\textrm{ and }\, \ip{\Wmat,\mmat'}\geq\ip{\Wmat,\mmat}\big|m\subseteq M} \leq \exp\pth{-\delta\cdot\pth{\tau_W-\mu/2}^2 - \delta\mu^2/4}
    \end{align*}
    which is equivalent to
    \begin{align*}
        \Pr\croc{\ip{\Wmat_G,\mmat}\geq \tau_G |m| \,\textrm{ and }\, \ip{\Wmat_G,\mmat'}\geq\ip{\Wmat_G,\mmat}\big|m\subseteq M} \leq \exp\pth{-\delta\cdot\frac{\tau_G^2+\zeta^2}{2\zeta}}
    \end{align*}
    where \(\zeta = \mu^2/2\).
\end{lemma}

\begin{proof}
    If \(y\geq x\) and \( x\geq t\), then \(\theta_1(y-x)+\theta_2(x-t) \geq 0\) for any choice of \(\theta_1,\theta_2> 0 \). Replacing \(y\) by \(\ip{\Wmat,\mmat'}\), \(x\) by \(\ip{\Wmat,\mmat}\) and \(t\) by \(\tau_W |m|\), we get the implication between the events of interest: \(\ip{\Wmat,\mmat'}\geq\ip{\Wmat,\mmat}\) and \(\ip{\Wmat,\mmat}\geq \tau_W |m|\) implies \(\theta_1 \pth{\ip{\Wmat,\mmat'}-\ip{\Wmat,\mmat}}+\theta_2\ip{\Wmat,\mmat} \geq \theta_2\tau_W|m|\).

    \(\ip{\Wmat,\mmat}\geq \tau_W |m|\) implies \(\theta_1 \pth{\ip{\Wmat,\mmat'}-\ip{\Wmat,\mmat}}+\theta_2\ip{\Wmat,\mmat}\) is the linear combination of \(2\delta\) independent Gaussian random variables and is therefore Gaussian. It has mean \(\delta\mu(\theta_2-\theta_1)\) and variance \(\delta\theta_1^2 + \delta(\theta_2-\theta_1)^2 = \delta\pth{2\theta_1^2-2\theta_1\theta_2+\theta_2^2}\).
    
    By Markov's inequality
    \begin{align*}
        &\Pr\croc{\theta_1 \pth{\ip{\Wmat,\mmat'}-\ip{\Wmat,\mmat}}+\theta_2\ip{\Wmat,\mmat} \geq \theta_2\tau_W|m|\,\big|m\subseteq M}\\ &\leq e^{-\theta_2\tau_W|m|}\bbE\croc{e^{\theta_1 \pth{\ip{\Wmat,\mmat'}-\ip{\Wmat,\mmat}}+\theta_2\ip{\Wmat,\mmat}}|m\subseteq M}\\
        &= \exp\pth{-\theta_2\tau_W \delta +\delta\mu(\theta_2-\theta_1) + \frac{\delta}{2}\pth{2\theta_1^2-2\theta_1\theta_2+\theta_2^2}}
    \end{align*}
    Pick \(\theta_1 = \tau_W\) and \(\theta_2= 2(\tau_W-\mu/2)\). Then, the expression in the last line simplifies to
    \begin{align*}
        \exp\pth{-(\tau_W-\mu/2)^2 - \mu^2/4},
    \end{align*}
    which matches the first part of the claim.

    \(-(\tau_W-\mu/2)^2 - \mu^2/4 = -\frac{(\mu\tau-\mu^2/2)+(\mu^2/2)^2}{\mu^2} = \frac{\tau_G^2 + \zeta^2}{2\zeta}\) gives us the second part of the claim.
\end{proof}

\subsection{Geometric intuition behind concentration inequalities}

For database alignment, by \hyperref[lemma:highDimStats]{Lemma \ref*{lemma:highDimStats}}, entries corresponding to true pairs in \(\Gmat\) have mean \(I_{XY}\) and false pairs have mean \(-I_{XY}\). All entries have variance \(2I_{XY}(1\pm o(1))\).

For planted matching, given \(\Wmat_G = \mu\Wmat - \mu^2/2\) a scaled and shifted version of the original edge weight matrix \(\Wmat\), entries corresponding to true pairs in \(\Wmat_G\) have mean \(\mu^2/2\) and false pairs have mean \(-mu^2/2\). All entries have variance \(\mu^2\).

For the rest of the section, we use \(\zeta\) refers to \(I_{XY}\) in the context of database alignment and \(\mu^2/2\) in the context of planted matching. Then true pairs have mean \(\zeta\), false pair have mean \(-\zeta\), and all pairs have variance \(2\zeta\) in both \(\Gmat\) and \(\Wmat_G\).

We want to bound the measure of the probability spaces that correspond to each type of error event. Consider the probability space \(\bbR^{\calU\times\calV}\). A two-dimensional projection of this space is given in \hyperref[fig:CE0]{Fig. \ref*{fig:CE0}}. Note the mean point \((\zeta ,-\zeta )\).

\begin{figure}[ht]
\centering
\resizebox{120pt}{!}{
    \begin{tikzpicture}[
     text height = 1.5ex,
     text depth =.1ex,
     b/.style={very thick}
     ]
     
    \draw (3,0.3) node {\(G_{u,v}\)};
    \draw (0.5,1.8) node {\(G_{u,v'}\)};
    
    \draw[b] (-1,0) -- (3.5,0);
    \draw[b] (0,-2.5) -- (0,2);

    \draw[b] (3.4,-0.1) -- (3.5,0);
    \draw[b] (3.4,0.1) -- (3.5,0);
    \draw[b] (-0.1,1.9) -- (0,2);
    \draw[b] (0.1,1.9) -- (0,2);

    \draw (2,-2) node {\(\cdot\)};
    \draw (2,-2) node {\(\circ\)};
    
    \draw (1.7,-0.2) node {\(\zeta\)};
    \draw (-0.5,-1.7) node {\(-\zeta\)};
    \draw[dotted] (2,0) -- (2,-2);
    \draw[dotted] (0,-2) -- (2,-2);
    \draw[b] (2,-0.1) -- (2,0.1);
    \draw[b] (-0.1,-2) -- (0.1,-2);
    
    \end{tikzpicture}
}
\caption{2-dimensional projection of probability space. \(G_{u,v}\) corresponds to a true pair and has mean \(\zeta\), while \(G_{u,v'}\) corresponds to a false pair and has mean \(-\zeta\).}
\label{fig:CE0}
\end{figure}
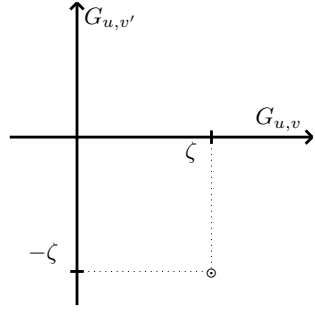

As shown in \hyperref[sec:algo]{Section \ref*{sec:algo}}, the objective function for all three algorithms is to maximize a linear combination of a shifted version of \(\Gmat\) or \(\Wmat_G\). All concentration equalities we use are bounds on the measures of half-spaces in the probability space \(\bbR^{\calU\times\calV}\). Approximating the entries of \(\Gmat\) as independent normal random variables with appropriate statistics (i.e. \((\mu,\sigma^2)=(\zeta ,2\zeta )\) for true pairs and \((\mu,\sigma^2)=(-\zeta ,2\zeta )\) for false pairs), we are able to get quick approximations for the bounds on half-spaces using the Chernoff bound. Specifically, the probability of a half-space is bounded by \(\exp\pth{-\frac{\ell^2}{2\sigma^2}}\) where \(\ell\) denotes the separation between the half-space and the mean point and \(\sigma^2=2\zeta \) is the variance of the terms. These bounds hold exactly in the planted matching case since entries of \(\Wmat_G\) are indeed independent normal random variables.

\subsubsection{True pair failing threshold testing}

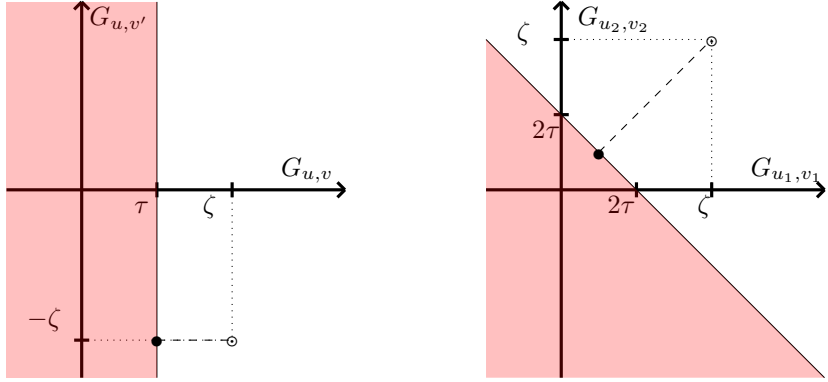
\begin{figure}[h]
\centering
\begin{minipage}{.45\textwidth}
\begin{tikzpicture}[
 text height = 1.5ex,
 text depth =.1ex,
 b/.style={very thick}
 ]
 
\draw (3,0.3) node {\(G_{u,v}\)};
\draw (0.5,2.3) node {\(G_{u,v'}\)};
    
\draw[b] (-1,0) -- (3.5,0);
\draw[b] (0,-2.5) -- (0,2.5);

\draw[b] (3.4,-0.1) -- (3.5,0);
\draw[b] (3.4,0.1) -- (3.5,0);
\draw[b] (-0.1,2.4) -- (0,2.5);
\draw[b] (0.1,2.4) -- (0,2.5);

\draw (2,-2) node {\(\cdot\)};
\draw (2,-2) node {\(\circ\)};
    
\draw (1.7,-0.2) node {\(\zeta\)};
\draw (-0.5,-1.7) node {\(-\zeta\)};
\draw[dotted] (2,0) -- (2,-2);
\draw[dotted] (0,-2) -- (2,-2);
\draw[b] (2,-0.1) -- (2,0.1);
\draw[b] (-0.1,-2) -- (0.1,-2);

\draw (1,-2.5) -- (1,2.5);
\draw (0.8,-0.2) node {\(\tau\)};
\draw[b] (1,-0.1) -- (1,0.1);
\fill[red,opacity=0.25] (1,-2.5) -- (1,2.5) -- (-1,2.5) -- (-1,-2.5) -- cycle;

\draw[dashed] (1,-2) -- (2,-2);
\draw (1,-2) node {\(\bullet\)};

\end{tikzpicture}
\end{minipage}
\begin{minipage}{.45\textwidth}
\begin{tikzpicture}[
 text height = 1.5ex,
 text depth =.1ex,
 b/.style={very thick}
 ]
 
\draw (3,0.3) node {\(G_{u_1,v_1}\)};
\draw (0.7,2.3) node {\(G_{u_2,v_2}\)};
    
\draw[b] (-1,0) -- (3.5,0);
\draw[b] (0,-2.5) -- (0,2.5);

\draw[b] (3.4,-0.1) -- (3.5,0);
\draw[b] (3.4,0.1) -- (3.5,0);
\draw[b] (-0.1,2.4) -- (0,2.5);
\draw[b] (0.1,2.4) -- (0,2.5);

\draw (2,2) node {\(\cdot\)};
\draw (2,2) node {\(\circ\)};
    
\draw (1.9,-0.2) node {\(\zeta\)};
\draw (-0.5,2.1) node {\(\zeta\)};
\draw[dotted] (2,0) -- (2,2);
\draw[dotted] (0,2) -- (2,2);
\draw[b] (2,-0.1) -- (2,0.1);
\draw[b] (-0.1,2) -- (0.1,2);

\draw (-1,2) -- (3.5,-2.5);
\draw (-0.2,0.8) node {\(2\tau\)};
\draw (0.8,-0.2) node {\(2\tau\)};
\draw[b] (-0.1,1) -- (0.1,1);
\draw[b] (1,-0.1) -- (1,0.1);
\fill[red,opacity=0.25] (-1,2) -- (3.5,-2.5) -- (-1,-2.5) -- cycle;

\draw[dashed] (0.5,0.5) -- (2,2);
\draw (0.5,0.5) node {\(\bullet\)};

\end{tikzpicture}
\end{minipage}
\caption{2-dimensional projection of the regions of the probability space corresponding to the event \(\{G_{u,v}\leq \tau\}\) (left-hand side) and the event \(\{G_{u_1,v_1}+G_{u_2,v_2}\leq 2\tau\}\) (right-hand side).}
\label{fig:CE1}
\end{figure}

The left-hand side of \hyperref[fig:CE1]{Fig. \ref*{fig:CE1}} illustrates the half-space corresponding to \(\{G_{u,v}\leq \tau\}\). The separation between the half-space and the mean point is equal to \(\ell=\zeta -\tau\). Then, the Chernoff bound gives us \(\exp\pth{-\frac{\ell^2}{2\sigma^2}} = \exp\pth{-\frac{\pth{\zeta -\tau}^2}{4\zeta }}\), which exactly matches the statement in \hyperref[lemma:typicality]{Lemma \ref*{lemma:typicality}}.

Similarly, the right-hand side of \hyperref[fig:CE1]{Fig. \ref*{fig:CE1}} illustrates the half-space corresponding to the case with 2 true pairs: \(\{G_{u_1,v_1}+G_{u_2,v_2}\leq 2\tau\}\). The separation between the half-space and the mean point is equal to \(\ell=\pth{\zeta -\tau}\sqrt{2}\). Then, the Chernoff bound gives us \(\exp\pth{-\frac{\ell^2}{2\sigma^2}} = \exp\pth{-2\cdot\frac{\pth{\zeta -\tau}^2}{4\zeta }}\), which exactly matches the statement in \hyperref[lemma:typicality]{Lemma \ref*{lemma:typicality}}.

This argument can be generalized to an arbitrary number of true pairs.

\subsubsection{False pair passing threshold testing}

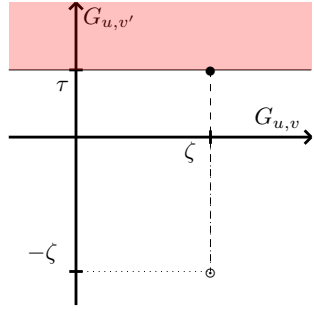
\begin{figure}[ht]
\centering
\resizebox{120pt}{!}{
\begin{tikzpicture}[
 text height = 1.5ex,
 text depth =.1ex,
 b/.style={very thick}
 ]
 
\draw (3,0.3) node {\(G_{u,v}\)};
\draw (0.5,1.8) node {\(G_{u,v'}\)};
    
\draw[b] (-1,0) -- (3.5,0);
\draw[b] (0,-2.5) -- (0,2);

\draw[b] (3.4,-0.1) -- (3.5,0);
\draw[b] (3.4,0.1) -- (3.5,0);
\draw[b] (-0.1,1.9) -- (0,2);
\draw[b] (0.1,1.9) -- (0,2);

\draw (2,-2) node {\(\cdot\)};
\draw (2,-2) node {\(\circ\)};
    
\draw (1.7,-0.2) node {\(\zeta\)};
\draw (-0.5,-1.7) node {\(-\zeta\)};
\draw[dotted] (2,0) -- (2,-2);
\draw[dotted] (0,-2) -- (2,-2);
\draw[b] (2,-0.1) -- (2,0.1);
\draw[b] (-0.1,-2) -- (0.1,-2);

\draw (-1,1) -- (3.5,1);
\draw (-0.2,0.8) node {\(\tau\)};
\draw[b] (-0.1,1) -- (0.1,1);
\fill[red,opacity=0.25] (-1,1) -- (3.5,1) -- (3.5,2) -- (-1,2) -- cycle;

\draw[dashed] (2,1) -- (2,-2);
\draw (2,1) node {\(\bullet\)};

\end{tikzpicture}
}
\caption{2-dimensional projection of the regions of the probability space corresponding to the event \(\{G_{u,v'}\geq \tau\}\)}
\label{fig:CE2}
\end{figure}

\hyperref[fig:CE2]{Fig. \ref*{fig:CE2}} illustrates the half-space corresponding to \(\{G_{u,v'}\geq \tau\}\). The separation between the half-space and the mean point is equal to \(\ell=\zeta +\tau\). Then, the Chernoff bound gives us \(\exp\pth{-\frac{\ell^2}{2\sigma^2}} = \exp\pth{-\frac{\pth{\zeta +\tau}^2}{4\zeta }}\), which exactly matches the statement in \hyperref[lemma:FP]{Lemma \ref*{lemma:FP}}.

\subsubsection{Misalignment}

\begin{figure}[ht]
\centering
\resizebox{120pt}{!}{
    \begin{tikzpicture}[
     text height = 1.5ex,
     text depth =.1ex,
     b/.style={very thick}
     ]
     
    \draw (3,0.3) node {\(G_{u,v}\)};
    \draw (0.5,1.8) node {\(G_{u,v'}\)};
        
    \draw[b] (-1,0) -- (3.5,0);
    \draw[b] (0,-2.5) -- (0,2);
    
    \draw (2,-2) node {\(\cdot\)};
    \draw (2,-2) node {\(\circ\)};
        
    \draw (1.7,-0.2) node {\(\zeta\)};
    \draw (-0.5,-1.7) node {\(-\zeta\)};
    \draw[dotted] (2,0) -- (2,-2);
    \draw[dotted] (0,-2) -- (2,-2);
    \draw[b] (2,-0.1) -- (2,0.1);
    \draw[b] (-0.1,-2) -- (0.1,-2);
    
    \draw (-1,-1) -- (2,2);
    \fill[red,opacity=0.25] (-1,-1) -- (2,2) -- (-1,2) -- cycle;
    
    \draw[dashed] (0,0) -- (2,-2);
    \draw (0,0) node {\(\bullet\)};
    
    \end{tikzpicture}
}
\caption{2-dimensional projection of the half-space corresponding to the event \(\{G_{u,v}\leq G_{u,v'}\}\).}
\label{fig:CE3}
\end{figure}
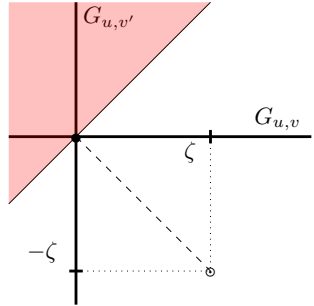

\hyperref[fig:CE3]{Fig. \ref*{fig:CE3}} illustrates the half-space corresponding to \(\{G_{u,v}\leq G_{u,v'}\}\). The separation between the half-space and the mean point is equal to \(\ell=\zeta \sqrt{2}\). Then, the Chernoff bound gives us \(\exp\pth{-\frac{\ell^2}{2\sigma^2}} = \exp\pth{-\frac{\zeta }{4}}\), which exactly matches the statement in \hyperref[lemma:misalignment]{Lemma \ref*{lemma:misalignment}}.

\subsubsection{Misalignment despite typicality}

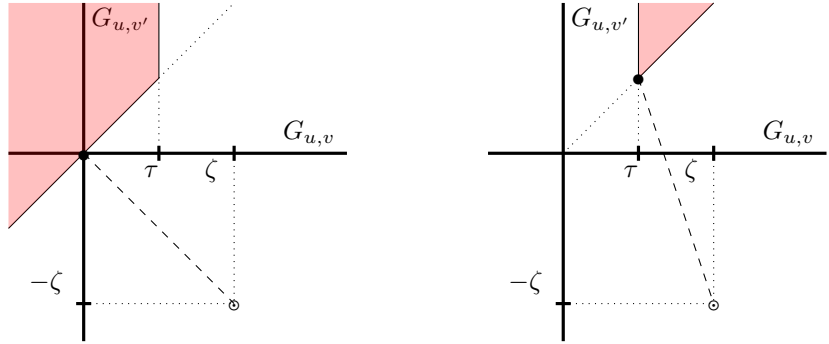
\begin{figure}[h]
\centering
\begin{minipage}{.45\textwidth}
\begin{tikzpicture}[
 text height = 1.5ex,
 text depth =.1ex,
 b/.style={very thick}
 ]
 
\draw (3,0.3) node {\(G_{u,v}\)};
\draw (0.5,1.8) node {\(G_{u,v'}\)};
    
\draw[b] (-1,0) -- (3.5,0);
\draw[b] (0,-2.5) -- (0,2);

\draw (2,-2) node {\(\cdot\)};
\draw (2,-2) node {\(\circ\)};
    
\draw (1.7,-0.2) node {\(\zeta\)};
\draw (-0.5,-1.7) node {\(-\zeta\)};
\draw[dotted] (2,0) -- (2,-2);
\draw[dotted] (0,-2) -- (2,-2);
\draw[b] (2,-0.1) -- (2,0.1);
\draw[b] (-0.1,-2) -- (0.1,-2);
    
\draw (-1,-1) -- (1,1);
\draw[dotted] (1,1) -- (2,2);
\fill[red,opacity=0.25] (-1,-1) -- (1,1) -- (1,2) -- (-1,2) -- cycle;
    
\draw[dashed] (0,0) -- (2,-2);
\draw (0,0) node {\(\bullet\)};

\draw[dotted] (1,0) -- (1,1);
\draw (1,1) -- (1,2);
\draw (0.9,-0.2) node {\(\tau\)};
\draw[b] (1,-0.1) -- (1,0.1);

\end{tikzpicture}
\end{minipage}
\begin{minipage}{.45\textwidth}
\begin{tikzpicture}[
 text height = 1.5ex,
 text depth =.1ex,
 b/.style={very thick}
 ]
 
\draw (3,0.3) node {\(G_{u,v}\)};
\draw (0.5,1.8) node {\(G_{u,v'}\)};
    
\draw[b] (-1,0) -- (3.5,0);
\draw[b] (0,-2.5) -- (0,2);

\draw (2,-2) node {\(\cdot\)};
\draw (2,-2) node {\(\circ\)};
    
\draw (1.7,-0.2) node {\(\zeta\)};
\draw (-0.5,-1.7) node {\(-\zeta\)};
\draw[dotted] (2,0) -- (2,-2);
\draw[dotted] (0,-2) -- (2,-2);
\draw[b] (2,-0.1) -- (2,0.1);
\draw[b] (-0.1,-2) -- (0.1,-2);
    
\draw[dotted] (0,0) -- (1,1);
\draw (1,1) -- (2,2);
\fill[red,opacity=0.25] (1,1) -- (1,2) -- (2,2) -- cycle;
    
\draw[dashed] (1,1) -- (2,-2);
\draw (1,1) node {\(\bullet\)};

\draw[dotted] (1,0) -- (1,1);
\draw (1,1) -- (1,2);
\draw (0.9,-0.2) node {\(\tau\)};
\draw[b] (1,-0.1) -- (1,0.1);

\end{tikzpicture}
\end{minipage}
\caption{Partition of the half-space in \hyperref[fig:CE3]{Fig. \ref*{fig:CE3}} according to whether information density \(G_{u,v}\) of the true pair \((u,v)\) is above or below the threshold \(\tau\).}
\label{fig:CE4}
\end{figure}

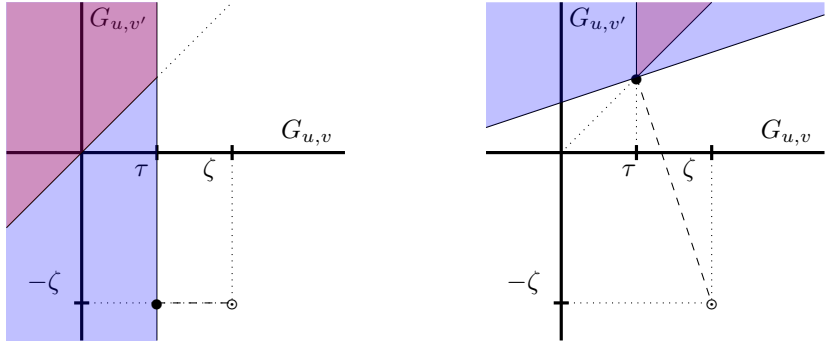
\begin{figure}[h]
\centering
\begin{minipage}{.45\textwidth}
\begin{tikzpicture}[
 text height = 1.5ex,
 text depth =.1ex,
 b/.style={very thick}
 ]
 
\draw (3,0.3) node {\(G_{u,v}\)};
\draw (0.5,1.8) node {\(G_{u,v'}\)};
    
\draw[b] (-1,0) -- (3.5,0);
\draw[b] (0,-2.5) -- (0,2);

\draw (2,-2) node {\(\cdot\)};
\draw (2,-2) node {\(\circ\)};
    
\draw (1.7,-0.2) node {\(\zeta\)};
\draw (-0.5,-1.7) node {\(-\zeta\)};
\draw[dotted] (2,0) -- (2,-2);
\draw[dotted] (0,-2) -- (2,-2);
\draw[b] (2,-0.1) -- (2,0.1);
\draw[b] (-0.1,-2) -- (0.1,-2);

\draw (1,-2.5) -- (1,2);
\draw (0.8,-0.2) node {\(\tau\)};
\draw[b] (1,-0.1) -- (1,0.1);
\fill[blue,opacity=0.25] (1,-2.5) -- (1,2) -- (-1,2) -- (-1,-2.5) -- cycle;
\draw (-1,-1) -- (1,1);
\draw[dotted] (1,1) -- (2,2);
\fill[red,opacity=0.25] (-1,-1) -- (1,1) -- (1,2) -- (-1,2) -- cycle;

\draw[dashed] (1,-2) -- (2,-2);
\draw (1,-2) node {\(\bullet\)};

\end{tikzpicture}
\end{minipage}
\begin{minipage}{.45\textwidth}
\begin{tikzpicture}[
 text height = 1.5ex,
 text depth =.1ex,
 b/.style={very thick}
 ]
 
\draw (3,0.3) node {\(G_{u,v}\)};
\draw (0.5,1.8) node {\(G_{u,v'}\)};
    
\draw[b] (-1,0) -- (3.5,0);
\draw[b] (0,-2.5) -- (0,2);

\draw (2,-2) node {\(\cdot\)};
\draw (2,-2) node {\(\circ\)};
    
\draw (1.7,-0.2) node {\(\zeta\)};
\draw (-0.5,-1.7) node {\(-\zeta\)};
\draw[dotted] (2,0) -- (2,-2);
\draw[dotted] (0,-2) -- (2,-2);
\draw[b] (2,-0.1) -- (2,0.1);
\draw[b] (-0.1,-2) -- (0.1,-2);
    
\draw[dotted] (0,0) -- (1,1);
\draw (1,1) -- (2,2);
\fill[red,opacity=0.25] (1,1) -- (1,2) -- (2,2) -- cycle;
    
\draw[dashed] (1,1) -- (2,-2);
\draw (1,1) node {\(\bullet\)};

\draw[dotted] (1,0) -- (1,1);
\draw (1,1) -- (1,2);
\draw (0.9,-0.2) node {\(\tau\)};
\draw[b] (1,-0.1) -- (1,0.1);

\draw (-1,1/3) -- (3.5,11/6);
\fill[blue,opacity=0.25] (-1,1/3) -- (3.5,11/6) -- (3.5,2) -- (-1,2) -- cycle;

\end{tikzpicture}
\end{minipage}
\caption{Blue-purple half-spaces containing the purple `slices' from \hyperref[fig:CE4]{Fig. \ref*{fig:CE4}}.}
\label{fig:CE5}
\end{figure}

The misalignment half-space shown in \hyperref[fig:CE3]{Fig. \ref*{fig:CE3}} can be broken down into two cases based on whether or not the true pairs have high enough average score. This is illustrated in \hyperref[fig:CE4]{Fig. \ref*{fig:CE4}}. These `slices' of a half-space can then be covered by another set of half-spaces, as illustrated in \hyperref[fig:CE5]{Fig. \ref*{fig:CE5}}. In both figures, the left-hand side corresponds to the atypicality event and the right-hand side corresponds to the misalignment-despote-typicality event.

If can be shown that, the half-space in the right-hand side of \hyperref[fig:CE5]{Fig. \ref*{fig:CE5}} is at distance (\(\sqrt{2\zeta ^2+2\eta^2}\)) to the mean point. Then, the Chernoff bound gives us \(\exp\pth{-\frac{\ell^2}{2\sigma^2}} = \exp\pth{-\frac{\zeta ^2+\tau^2}{2\zeta }}\), which exactly matches the statement in \hyperref[lemma:condMisalignment]{Lemma \ref*{lemma:condMisalignment}}.

These figures also help demonstrate the contribution of this approach in analysis: Both the original misalignment half-space in \hyperref[fig:CE3]{Fig. \ref*{fig:CE3}} as well as the misalignment-despite-typicality half-space in the right-hand side of \hyperref[fig:CE5]{Fig. \ref*{fig:CE5}} change based on the choice of \((u,v')\). The separation between half-spaces is greater with the misalignment-despite-typicality event, which gives us some improvement in the error bound at the cost of having to consider the atypicality error shown in the left-hand side of \hyperref[fig:CE5]{Fig. \ref*{fig:CE5}}. This last half-space however, does not depend on \((u,v')\). The atypicality half-space is fixed once we pick \((u,v)\). Therefore this term does not require taking a union bound.

For an appropriate choice of \(\tau\), the gains made by the improvement from the misalignment half-space to the misalignment-despite-typicality half-space can compensate for the cost of having to consider the atypicality half-space. This is thanks to the fact that, unlike misalignment and misalignment-despite-typicality, we need not need to reconsider atypicality for every choice of \((u,v')\).

\section{Generating function}

\begin{definition}[Generalized generating function]

    The generating function \(R=R^{\calU,\calV}:\bbR^{\calU\times\calV}\to\bbR\) is defined such that
    \begin{align*}
        R(\Thetamat) &\define \int\int \exp\pth{
        \ip{\Gmat,\Thetamat}} f_{\Avec}(\avec)f_{\Bvec}(\bvec)d\avec d\bvec
    \end{align*}
    where \(f_{\Avec},f_{\Bvec}\) denote the marginal probabilities for the two databases and \(\Gmat\in\bbR^{\calA\times\calB}\) denotes the information density matrix for \(\avec,\bvec\) as defined in \hyperref[subsec:algo1]{Section \ref*{subsec:algo1}}.
\end{definition}

\begin{lemma}
    \label{lemma:genFunc2exp}
    \(\bbE\croc{\exp\pth{\ip{\Gmat,\Thetamat}}|M=m_1} = R(\Thetamat+\mmat_1)\)
\end{lemma}
\begin{proof}
    The key equality of the proof is that \(\ip{\Gmat,\Mmat} = \log f_{\Avec,\Bvec|M}(\avec,\bvec) - \log f_{\Avec}(\avec)-\log f_{\Bvec}(\bvec)\), where \(f_{\Avec,\Bvec|M}\) is the joint distribution between databases given \(M\). We show as follows:\\
    Let \(\calU_M\subseteq\calU\) and \(\calV_M\subseteq\calV\) denote the set of users that have a mapping under \(M\) and \(\calW_M\subset \calU_M\times\calV_M\) denote the set of pairs mapped by \(M\). By the model for Gaussian data structures (as given in \hyperref[subsec:modelDatabase]{Subsection \ref*{subsec:modelDatabase}}), all matched feature pairs and unmatched features are mutually independent. It then follows that
    \begin{align*}
        \log f_{\Avec,\Bvec|M}(\avec,\bvec) &= \sum_{(u,v)\in\calW_m}\log f_{XY}(\Avec(u),\Bvec(v))\\
        &\pheq + \sum_{u'\in\calU\setminus\calU_m}\log f_X(\Avec(u')) + \sum_{v'\in\calV\setminus\calV_m}\log f_Y(\Avec(v'))
    \end{align*}
    where \(f_{XY}\) denotes the joint distribution of correlated features while \(f_X,f_Y\) denote the marginals.\\
    As defined in \hyperref[subsec:algo1]{Section \ref*{subsec:algo1}}, \(G_{u,v}=\log\frac{f_{XY}(\Avec(u),\Bvec(v))}{f_X(\Avec(u))f_Y(\Bvec(v))}\) for any \(u\in\calU\) and \(v\in\calV\). Then,
    \begin{align*}
        \ip{\Gmat,\Mmat} &= \sum_{(u,v)\in\calW_m} \pth{\log f_{XY}(\Avec(u),\Bvec(v)) - \log f_X(\Avec(u)) - \log f_Y(\Bvec(v))}\\
        \implies \log f_{\Avec,\Bvec|M}(\avec,\bvec)-\ip{\Gmat,\Mmat} &= \sum_{(u,v)\in\calW_m}\pth{\log f_X(\Avec(u))+\log f_Y(\Bvec(v))}\\
        &\pheq + \sum_{u'\in\calU\setminus\calU_m}\log f_X(\Avec(u')) + \sum_{v'\in\calV\setminus\calV_m}\log f_Y(\Avec(v'))\\
        &= \sum_{u\in\calU}\log f_X(\Avec(u)) + \sum_{v\in\calV}\log f_Y(\Avec(v))\\
        &= \log f_{\Avec}(\avec) + \log f_{\Bvec}(\bvec)
    \end{align*}
    which shows that \(\ip{\Gmat,\Mmat} = \log f_{\Avec,\Bvec|M}(\avec,\bvec) - \log f_{\Avec}(\avec)-\log f_{\Bvec}(\bvec)\).\\
    Then
    \begin{align*}
        &\bbE\croc{\exp\pth{\ip{\Gmat,\Thetamat-\Mmat}}|M}\\
        &= \int\int \exp\pth{
        \ip{\Gmat,\Thetamat-\Mmat}} f_{\Avec,\Bvec|M}(\avec,\bvec)d\avec d\bvec\\
        &= \int\int \exp\pth{
        \ip{\Gmat,\Thetamat}+\log f_{\Avec,\Bvec|M}(\avec,\bvec) - \ip{\Gmat,\Mmat}} d\avec d\bvec\\
        &= \int\int \exp\pth{
        \ip{\Gmat,\Thetamat}+\log f_{\Avec}(\avec) + \log f_{\Bvec}(\bvec) } d\avec d\bvec\\
        &= \int\int \exp\pth{
        \ip{\Gmat,\Thetamat}} f_{\Avec}(\avec)f_{\Bvec}(\bvec)d\avec d\bvec = R(\Thetamat)
    \end{align*}
    This completes the proof.
\end{proof}

\begin{corollary}
    \label{cor:chernoff}
    Let \(\mmat_1,\mmat_2\in \{0,1\}^{\calU\times\calV}\) be the matrix encodings of the mappings \(m_1\) and \(m_2\) respectively. We have the following Chernoff bounds:
    \begin{compactenum}[a)]
        \item Probability of atypicality:\\
        \begin{align*}
            \Pr\croc{\tau|m_1| \geq \ip{\Gmat,\mmat_1}\big|M=m_1} \leq \exp\pth{\theta\tau|m_1|}R\pth{(1-\theta)\mmat_1}
        \end{align*}
        for any \(\theta>0\).
        \item Probability of misalignment:
        \begin{align*}
            \Pr\croc{\ip{\Gmat,\mmat_2} \geq \ip{\Gmat,\mmat_1}|M=m_1} \leq R\pth{(1-\theta)\mmat_1 + \theta\mmat_2}
        \end{align*}
        for any \(\theta>0\).
        \item Probability of misalignment despite typicality:
        \begin{align*}
            \Pr\croc{\ip{\Gmat,\mmat_1}\geq \tau|m_1|\land \ip{\Gmat,\mmat_2}\geq\ip{\Gmat,\mmat_1}\big|M=m_1} \leq e^{-\tau|m_1|(\nu-1)}R\pth{\nu(1-\theta)\mmat_1 + \nu\theta\mmat_2}
        \end{align*}
        for any \(\theta>0\) and \(\nu>1\).
    \end{compactenum}
\end{corollary}
\begin{proof}
    \begin{compactenum}[a)]
        \item Probability of atypicality:
        \begin{align*}
            &\Pr\croc{\tau|m_1| \geq \ip{\Gmat,\mmat_1}\big|M=m_1}\\
            &=\Pr\croc{\ip{\Gmat-\tau,\mmat_1}\leq 0|M=m_1}\\
            &= \Pr\croc{\exp\pth{\theta\cdot\ip{\Gmat-\tau,-\mmat_1}}\geq 1^\theta | M=m_1}\\
            &= \Pr\croc{\exp\pth{\ip{\Gmat-\tau,-\theta\mmat_1}}\geq 1 | M=m_1}\\
            &\leq \bbE\croc{\exp\pth{\ip{\Gmat-\tau,-\theta\mmat_1}}| M=m_1}\\
            &= \exp\pth{\theta\tau|m_1|}\bbE\croc{\exp\pth{\ip{\Gmat,-\theta\mmat_1}}| M=m_1}\\
            &= \exp\pth{\theta\tau|m_1|}R\pth{(1-\theta)\mmat_1}.
        \end{align*}
        \item Probability of misalignment:
        \begin{align*}
            &\Pr\croc{\ip{\Gmat,\mmat_2}\geq\ip{\Gmat,\mmat_1}|M=m_1}\\
            &=\Pr\croc{\ip{\Gmat,\mmat_2-\mmat_1}\geq 0|M=m_1}\\
            &= \Pr\croc{\exp\pth{\ip{\Gmat,\theta(\mmat_2-\mmat_1)}}\geq 1 | M=m_1}\\
            &\leq \bbE\croc{\exp\pth{\ip{\Gmat,\theta(\mmat_2-\mmat_1)}}| M=m_1}\\
            &= R\pth{(1-\theta)\mmat_1 + \theta\mmat_2}.
        \end{align*}
        \item Probability of misalignment despite typicality:\\
        If \(y\geq x\) and \( x\geq t\), then \(\theta_1(y-x)+\theta_2(x-t) \geq 0\) for any choice of \(\theta_1,\theta_2> 0 \). Replace \(\theta_1\) by \(\nu\theta\) and \(\theta_2\) by \(\nu-1\). \(\theta_1,\theta_2> 0 \) holds for any \(\theta>0\) and \(\nu>1\). It then follows that, if \(y\geq x\) and \( x\geq t\), then \(\nu \theta (y-x)-(1-\nu)x \geq (\nu-1)t\). Then
        \begin{align*}
            &\Pr\croc{\ip{\Gmat,\mmat_2}\geq\ip{\Gmat,\mmat_1}\land \ip{\Gmat,\mmat_1}\geq \tau|m_1|\,\,\big|M=m_1}\\
            &\leq \Pr\croc{\ip{\Gmat,\nu\theta(\mmat_2-\mmat_1)-(1-\nu)\mmat_1}\geq \tau|m_1|(\nu-1) |M=m_1}\\
            &= \Pr\croc{e^{\ip{\Gmat,\nu\theta(\mmat_2-\mmat_1)-(1-\nu)\mmat_1}}\geq e^{\tau|m_1|(\nu-1)} |M=m_1}\\
            &\leq e^{-\tau|m_1|(\nu-1)}\bbE\croc{e^{\ip{\Gmat,\nu\theta(\mmat_2-\mmat_1)-(1-\nu)\mmat_1}}|M=m_1}\\
            &= e^{-\tau|m_1|(\nu-1)}R\pth{\nu\theta(\mmat_2-\mmat_1)+\nu\mmat_1}\\
            &= e^{-\tau|m_1|(\nu-1)}R\pth{\nu(1-\theta)\mmat_1 + \nu\theta\mmat_2}
        \end{align*}
    \end{compactenum}
\end{proof}

\subsection{Main lemmas on the generating function}

\begin{lemma}
    \label{lemma:computationGenFunc}
    Define \(\Pmat=\Pmat^{\calU,\calV}:\bbR^{\calU\times\calV}\times(-1,1)\to\bbR^{(\calU\sqcup\calV)\times(\calU\sqcup\calV)}\) such that
    \begin{align*}
        \Pmat(\Thetamat,\rho) \define \pth{1-\rho^2}\Imat+\crocMat{\rho^2\cdot\diag(\Thetamat\onevec)}{-\rho\Thetamat}{-\rho\Thetamat^\top}{\rho^2\cdot\diag(\Thetamat^\top \onevec)}
    \end{align*}
    where \(\onevec\) represents appropriately indexed vectors of all ones.
    
    If \(\Pmat(\Thetamat,\rho_i)\) positive definite for each \(i\in\calD\), then evaluating the generating function \(R\) at \(\Thetamat\) gives the expression
    \begin{align*}
        R(\Thetamat) = \prod_{i\in\calD} \croc{\frac{\pth{1-\rho_i^2}^{|\calU|+|\calV|-\sum_{(u,v)} \Theta_{u,v}}}{\det \Pmat(\Thetamat,\rho_i)}}^{\frac{1}{2}}
    \end{align*}
    where \(\rhovec \in (-1,1)^{\calD}\) is the correlation vector in canonical form.
\end{lemma}
\begin{proof}
    Let \(\Avec,\Bvec\) databases in canonical form with statistics \(\muvec = \zerovec\) and \(\Sigmamat = \crocMat{\Imat}{\diag(\rhovec)}{\diag(\rhovec)}{\Imat}\) and \(\Gmat\in\bbR^{\calU\times\calV}\) their information density matrix.
    \begin{align*}
        &g_{XY}(\xvec,\yvec) = \log\frac{f_{\Xvec,\Yvec}(\xvec,\yvec)}{f_{\Xvec}(\xvec)f_{\Yvec}(\yvec)}\\
        &= \sum_{i\in\calD}\croc{-\frac{1}{2}\log\pth{1-\rho_i^2}-\frac{\rho_i^2\pth{x_i^2+y_i^2}-2\rho_ix_iy_i}{2\pth{1-\rho_i^2}}}
    \end{align*}

Define \(\avec\in\bbR^{\calU}\) and \(\bvec\in\bbR^{\calV}\) such that \(a_u\) and \(b_v\) denote the features associated with users \(u\) and \(v\) respectively. As shown in the proof for \hyperref[lemma:genFunc2exp]{Lemma \ref*{lemma:genFunc2exp}}, \(\log f_{\Avec,\Bvec|M}(\avec,\bvec) - \ip{\Gmat,\Mmat} = \sum\log f_X(a_u)+\sum\log f_Y(b_v)\). Without loss of generality, assume \(\rhovec=[\rho]\). Then,
\begin{align*}
    &\ip{\Gmat,\Thetamat} + \log f_{\Avec,\Bvec|M}(\avec,\bvec) - \ip{\Gmat,\Mmat}\\
    &= \ip{\Gmat,\Thetamat}+\sum_{u\in\calU}\log f_X(a_u)+\sum_{v\in\calV}\log f_Y(b_v)\\
    &= -\frac{1}{2}\log\pth{1-\rho^2}\sum_{(u,v)} \Theta_{u,v}-\frac{|\calU|+|\calV|}{2}\log\pth{2\pi}\\
    &\pheq - \frac{1}{2(1-\rho^2)}\sum_{u\in\calU} a_u^2\pth{1-\rho^2+\rho^2\sum_{v\in\calV}\Theta_{u,v}}\\
    &\pheq - \frac{1}{2(1-\rho^2)}\sum_{v\in\calV} b_v^2\pth{1-\rho^2+\rho^2\sum_{u\in\calU}\Theta_{u,v}}\\
    &\pheq -\frac{1}{2(1-\rho^2)}\sum_{u\in\calV}\sum_{v\in\calV} \rho\Thetamat_{u,v} a_ub_v + \rho\Thetamat_{u,v} b_va_u\\
    &= -\frac{1}{2}\log\pth{1-\rho^2}\sum_{(u,v)} \Theta_{u,v}-\frac{|\calU|+|\calV|}{2}\log\pth{2\pi}\\
    &\pheq -\frac{1}{2(1-\rho^2)}\crocVec{\avec}{\bvec}^\top \Pmat(\Thetamat,\rho)\crocVec{\avec}{\bvec}
\end{align*}
    
    Then we can write
    \begin{align*}
        \exp&\pth{\ip{\Gmat,\Thetamat-\Mmat}} f_{\Avec,\Bvec|M}(\avec,\bvec)\\
        &= \frac{\pth{1-\rho^2}^{-\frac{1}{2}\sum_{(u,v)} \Theta_{u,v}}}{(2\pi)^{\frac{|\calU|+|\calV|}{2}}}\cdot\exp\pth{-\frac{1}{2\pth{1-\rho^2}}\crocVec{\alphavec_i}{\betavec_i}^\top \Pmat(\Thetamat,\rho)\crocVec{\alphavec_i}{\betavec_i}}
    \end{align*}
    Taking the integral of this expression gives us the claimed result. For the case where \(\rhovec\) is multi-dimensional, taking the product of this expression over each \(i\in\calD\) gives us the proper expression.
\end{proof}

\begin{lemma}[Decomposition for block diagonal matrices]
    \label{lemma:blockGenFunc}
    
    If \(\Thetamat\in\bbR^{\calU\times\calV}\) can be written in block diagonal form, i.e. if \(\calU\) and \(\calV\) can be partitioned into \(\calU_1,\calU_2\) and \(\calV_1,\calV_2\) such that \(\Thetamat\) can be written as \(\Thetamat=\crocMat{\Thetamat_1}{0}{0}{\Thetamat_2}\in\bbR^{(\calU_1\sqcup\calU_2)\times(\calV_1\sqcup\calV_2)}\), then
    \begin{align*}
        R^{\calU,\calV}(\Thetamat) = R^{\calU_1,\calV_1}(\Thetamat_1)\cdot R^{\calU_2,\calV_2}(\Thetamat_2)
    \end{align*}
\end{lemma}
\begin{proof}
    This follows from the fact that \(\Pmat^{\calU,\calV}(\Thetamat,\rho)\in\bbR^{\calU\times\calV}\) can be transformed into black matrix form as \(\crocMat{\Pmat^{\calU_1,\calV_1}(\Thetamat_1,\rho)}{0}{0}{\Pmat^{\calU_2,\calV_2}(\Thetamat_2,\rho)}\) by simultaneously permuting rows and columns.
    
    Then \(\det \Pmat^{\calU,\calV}(\Thetamat) = \det \Pmat^{\calU_1,\calV_1}(\Thetamat_1,\rho)\cdot\det \Pmat^{\calU_2,\calV_2}(\Thetamat_2,\rho)\) and we get the claimed result.
\end{proof}

\begin{lemma}
    \label{lemma:rOne2One}
    If \(\theta\in\bbR\) such that \(|\theta|<|1/\rho_i|\) for any \(i\in\calD\), then
    \begin{align*}
        R([1-\theta]) &= \prod_{i\in\calD}\croc{\frac{\pth{1-\rho_i^2}^{\theta}}{1-\rho_i^2\theta^2}}^{\frac{1}{2}}
    \end{align*}
    Furthermore, if \(\theta\in[-1,1]\) then
    \begin{align*}
        R([1-\theta]) &\leq \exp\pth{-\theta(1-\theta)I_{XY}},
    \end{align*}
    where \(I_{XY} = -\frac{1}{2}\sum_{i\in\calD}\log\pth{1-\rho_i^2}\).
\end{lemma}
\begin{proof}
    \begin{align*}
        \Pmat([1-\theta],\rho) = \crocMat{1-\rho^2\theta}{-\rho(1-\theta)}{-\rho(1-\theta)}{1-\rho^2\theta}
    \end{align*}
    The eigenvalues of this matrix are \(\pth{1-\rho^2\theta}-\rho(1-\theta)\) and \(\pth{1-\rho^2\theta}+\rho(1-\theta)\). They are both strictly positive if and only if \(|\theta|<|1/\rho|\). Their product equals \((1-\rho^2)(1-\rho^2\theta^2)\). Plugging in this value in the denominator of \(\frac{\pth{1-\rho_i^2}^{|\calU|+|\calV|-\sum_{(u,v)} \Theta_{u,v}}}{\det \Pmat(\Thetamat,\rho_i)}\) as given in \hyperref[lemma:computationGenFunc]{Lemma \ref*{lemma:computationGenFunc}} gives us the exact expression for \(R([1-\theta])\).
    
    By \hyperref[lemma:push2exp]{Lemma \ref*{lemma:push2exp}}, \((1-\rho^2\theta^2)\geq(1-\rho^2)^{\theta^2}\) if \(\theta\in[-1,1]\), which gives us
    \begin{align*}
        R([1-\theta]) &\leq \prod_{i\in\calD}\pth{1-\rho_i^2}^{\frac{\theta(1-\theta)}{2}}
    \end{align*}
    This gives us the upper bound for \(R([1-\theta])\).
\end{proof}

\begin{corollary}
    \label{cor:rMapping}
    Let \(\mmat\in\{0,1\}^{\calU\times\calV}\) some binary matrix with row sums and columns sums at most 1. Let \(n\) denote the sum of its entries. Then
    \begin{align*}
        R((1-\theta)\mmat) &= \prod_{i\in\calD}\croc{\frac{\pth{1-\rho_i^2}^{\theta}}{1-\rho_i^2\theta^2}}^{\frac{n}{2}}
    \end{align*}
    Furthermore, if \(\theta\in[-1,1]\) then
    \begin{align*}
        R([1-\theta]\mmat) &\leq \exp\pth{-n\theta(1-\theta)I_{XY}},
    \end{align*}
    where \(I_{XY} = -\frac{1}{2}\sum_{i\in\calD}\log\pth{1-\rho_i^2}\).
\end{corollary}

\begin{proof}
    \((1-\theta)\mmat\) has at most 1 non-zero entry in each row and in each column. Then it can be arranged to have block diagonal form where each non-zero block on the diagonal has size 1 and is equal to \([1-\theta]\). By \hyperref[lemma:blockGenFunc]{Lemma \ref*{lemma:blockGenFunc}}, \(R((1-\theta)\mmat)\) decomposes into the product of \(\croc{R([1-\theta])}^n\) and \(R(\zeromat)\) where \(\zeromat\) is the zero block of the block diagonal decomposition. By \hyperref[lemma:computationGenFunc]{Lemma \ref*{lemma:computationGenFunc}}, \(R(\zeromat)=1\). The expression for \(R([1-\theta]\) is given by \hyperref[lemma:rOne2One]{Lemma \ref*{lemma:rOne2One}}.
\end{proof}

\subsection{Generating function evaluated for cycles and even paths}

\begin{lemma}[Chernoff bound for a cycle]
    \label{lemma:rCycle}
    Given \(\nu\in[0,2]\), let \(m_1\) and \(m_2\) be two mappings of size \(n\) such that \(\Thetamat = \frac{\nu}{2}(\mmat_1 + \mmat_2)\) be a matrix block corresponding to a cycle of the type given in \hyperref[fig:decomposition2]{Fig. \ref*{fig:decomposition2}-I}.
    \begin{align*}
        \frac{\log R(\Thetamat)}{I_{XY}} &\leq -\frac{n}{2}\nu(2-\nu)+n\cdot\rho_{\max}^2(\nu-1)
    \end{align*}
    where \(I_{XY} = -\frac{1}{2}\sum\log\pth{1-\rho_i^2}\), mutual information between correlated features, and \(\rho_{\max} = \max |\rho_i|\), the largest correlation coefficient under the canonical form.
\end{lemma}

\begin{proof}
    To keep track of \(\mmat_1\) and \(\mmat_2\), let us write \(\Thetamat = \nu\theta\mmat_1 + \nu(1-\theta)\mmat_2\) such that the two matrices have distinct coefficients.\\
    Without loss of generality, let \(\Thetamat\) have its rows and columns arranged such that
    \begin{align*}
        \Thetamat = \nu
        \begin{bmatrix}
            \theta & 1-\theta & 0  & \cdots & 0 & 0\\
            0 & \theta & 1-\theta  & \cdots & 0 & 0\\
            0 & 0 & \theta  & \cdots & 0 & 0\\
            \vdots & \vdots  & \vdots & \ddots & \vdots & \vdots\\
            0 & 0 & 0 & \cdots & \theta & 1-\theta\\
            1-\theta & 0 & 0 & \cdots & 0 & \theta\\
        \end{bmatrix}.
    \end{align*}
    \(\Thetamat\) has \(|\calU|=|m_1|=n\) rows and \(|\calV|=|m_1|=n\) columns.\\
    Recall the definition of \(\Pmat(\Thetamat,\rho)\) as given in \hyperref[lemma:computationGenFunc]{Lemma \ref*{lemma:computationGenFunc}}:
    \begin{align*}
        \Pmat(\Thetamat,\rho) &\define \pth{1-\rho^2}\Imat+\crocMat{\rho^2\cdot\diag(\Thetamat\onevec)}{-\rho\Thetamat}{-\rho\Thetamat^\top}{\rho^2\cdot\diag(\Thetamat^\top \onevec)}
    \end{align*}
    Define the \(\Pmat_1\) and \(\Pmat_2\) to be the matrices that form the diagonal blocks of \(\Pmat(\Thetamat,\rho) = \crocMat{\Pmat_1}{-\rho\Thetamat}{-\rho\Thetamat^\top}{\Pmat_2}\):
    \begin{align*}
        \Pmat_1 &\define \pth{1-\rho^2}\Imat + \rho^2\diag(\Thetamat\vctr{1})\\
        &= \pth{1-\rho^2(1-\nu)}\Imat\\
        \Pmat_2 &\define \pth{1-\rho^2}\Imat + \rho^2\diag(\Thetamat^\top\vctr{1})\\
        &= \pth{1-\rho^2(1-\nu)}\Imat
    \end{align*}
    Then
    \begin{align*}
        \det \Pmat(\Thetamat,\rho) &= \croc{\det \Pmat_1}\croc{\det\pth{\Pmat_2 - \rho^2\Thetamat^\top \Pmat_1^{-1} \Thetamat}}\\
        &= \croc{1-\rho^2(1-\nu)}^n\det\croc{\pth{1-\rho^2(1-\nu)}\Imat - \frac{\rho^2}{1-\rho^2(1-\nu)}\Thetamat^\top\Thetamat}\\
        &= \det\croc{\pth{1-\rho^2(1-\nu)}^2\Imat - \rho^2\Thetamat^\top\Thetamat}
    \end{align*}
    where \(\Pmat_2 - \rho^2\Thetamat^\top \Pmat_1^{-1} \Thetamat = \pth{1-\rho^2(1-\nu)}\Imat - \frac{\rho^2}{1-\rho^2(1-\nu)}\Thetamat^\top\Thetamat\) is the Schur complement of \(\Pmat_1\).\\
    We can write
    \begin{equation*}
        \Thetamat^\top\Thetamat\\
        = \nu^2\pth{1-2\theta(1-\theta)}\Imat+2\nu^2\theta(1-\theta)\Qmat
    \end{equation*}
    where 
    \begin{align*}
        \Qmat &\define \frac{1}{2}
        \begin{bmatrix}
            0 & 1 & 0  & \cdots & 0 & 0 & 1\\
            1 & 0 & 1  & \cdots & 0 & 0 & 0\\
            0 & 1 & 0  & \cdots & 0 & 0 & 0\\
            \vdots & \vdots  & \vdots & \ddots & \vdots & \vdots & \vdots\\
            0 & 0 & 0 & \cdots & 0 & 1 & 0\\
            0 & 0 & 0 & \cdots & 1 & 0 & 1\\
            0 & 0 & 0 & \cdots & 0 & 1 & 0
        \end{bmatrix}
    \end{align*}
    and \(\Qmat\) has eigenvalues \(\cos\pth{\frac{2k\pi}{n}}\), \(k\in\{0,1,2,\cdots,n-1\}\). These are all between \(-1\) and \(1\). Furthermore, the sum of the eigenvalues of \(\Qmat\) is \(\tr(\Qmat)=0\)
    Thus the eigenvalues of \(\pth{1-\rho^2(1-\nu)}^2\Imat - \rho^2\Thetamat^\top\Thetamat\) are 
    \[
    \pth{1-\rho^2(1-\nu)}^2 - \rho^2\nu^2\pth{\pth{1-2\theta(1-\theta)}+2\theta(1-\theta)\cos\pth{\frac{2k\pi}{n}}}.
    \] 
    By \hyperref[lemma:minProd]{Lemma \ref*{lemma:minProd}} with \(\tau = \pth{1-\rho^2(1-\nu)}^2 - \rho^2\nu^2\pth{1-2\theta(1-\theta)}\) and \(\sigma = -2\rho^2\nu^2\theta(1-\theta)\)
    \begin{align*}
    \det \Pmat(\Thetamat,\rho) 
    &= \det\croc{\pth{1-\rho^2(1-\nu)}^2\Imat - \rho^2\Thetamat^\top\Thetamat}\\
    &\geq \croc{\pth{1-\rho^2(1-\nu)}^2 - \rho^2\nu^2 + 4\rho^2\nu^2\theta(1-\theta)}^{\frac{n}{2}}
    \croc{\pth{1-\rho^2(1-\nu)}^2 - \rho^2\nu^2}^{\frac{n}{2}}\\
    &= \croc{\pth{1-\rho^2(1-\nu)}^2 - \rho^2\nu^2(1-2\theta)^2}^{\frac{n}{2}}\croc{1-\rho^2}^{\frac{n}{2}}
    \croc{1-\rho^2(1-\nu)^2}^{\frac{n}{2}}\\
    \end{align*}
    This last expression is maximized over \(\theta\) by picking \(\theta = 1/2\). Then, the inequality becomes the inequality above becomes \(\det \Pmat(\Thetamat,\rho) \geq \croc{1-\rho^2(1-\nu)}^{n}\croc{1-\rho^2}^{\frac{n}{2}}\croc{1-\rho^2(1-\nu)^2}^{\frac{n}{2}}\).
    Since \(\nu\in[0,2]\iff(1-\nu)^2\in[0,1]\) by \hyperref[lemma:push2exp]{Lemma \ref*{lemma:push2exp}}, we have the bound \(1-\rho^2(1-\nu)^2 \geq \pth{1-\rho^2}^{(1-\nu)^2}\). Furthermore, since \(-(1-\nu)\in[0,1]\), we have the bound \(1-\rho^2(1-\nu) \geq \pth{1+\rho^2}^{-(1-\nu)}\). Then,
    \begin{align*}
        \log \det \Pmat(\Thetamat,\rho) \geq \frac{n}{2}\pth{1+(1-\nu)^2}\log\pth{1-\rho^2}-n(1-\nu)\log\pth{1+\rho^2}
    \end{align*}
    \(R(\Thetamat)\) has \(\frac{1}{2}\pth{|\calU|+|\calV|-\sum_{(u,v)} \Theta_{u,v}} = \frac{n}{2}(1+1-\nu) = \frac{n}{2}(2-\nu)\) multiplicative terms of \(\pth{1-\rho^2}\) in the numerator and our bound has \(\frac{n}{4}\pth{1+(1-\nu)^2}\) of them in the denominator, which gives, in total, \(\frac{n}{4}\nu(2-\nu) + \frac{n}{2}(1-\nu)\) such terms. Then
    \begin{align*}
        \log R(\Thetamat) &= \frac{n}{4}\nu(2-\nu)\log\pth{1-\rho^2} - \frac{n}{2}(\nu-1)\croc{\log\pth{1+\rho^2}+\log\pth{1-\rho^2}}
    \end{align*}
    By \hyperref[lemma:compareEpsilon]{Lemma \ref*{lemma:compareEpsilon}}, \(-\frac{1}{2}\sum \log\pth{1-\rho_i^2}+\log\pth{1+\rho_i^2}\) is upper bounded by \(-\frac{1}{2}\sum \rho_i^2\log\pth{1-\rho_i^2}\), which itself is upper bounded by \(\rho_{\max}^2 I_{XY}\) where \(\rho_{\max} = \max |\rho_i|\) and \(I_{XY} = -\frac{1}{2}\sum \log\pth{1-\rho_i^2}\). This gives us the claimed result.
    
\end{proof}

\begin{lemma}[Chernoff bound for an even path]
    \label{lemma:rEvenPath}
    Given \(\nu\in[1,2]\), let \(m_1\) and \(m_2\) be two mappings of size \(n\) such that \(\Thetamat = \frac{\nu}{2}(\mmat_1+\mmat_2)\) be a matrix block corresponding to a cycle of the type given in \hyperref[fig:decomposition2]{Fig. \ref*{fig:decomposition2}-II}.
    \begin{align*}
        \frac{\log R(\Thetamat)}{I_{XY}} \leq -\frac{n}{2}\cdot \nu (2-\nu) - (\nu-1)^2\nu^2 + 6n\cdot \rho_{\max}^2(\nu-1)
    \end{align*}
    where \(I_{XY} = -\frac{1}{2}\sum\log\pth{1-\rho_i^2}\), mutual information between correlated features, and \(\rho_{\max} = \max |\rho_i|\), the largest correlation coefficient under the canonical form.
\end{lemma}

\begin{proof}
    Without loss of generality, let \(\Thetamat\) have its rows and columns arranged such that
    \begin{align*}
        \Thetamat = \nu
        \begin{bmatrix}
            \theta & 1-\theta & 0  & \cdots & 0 & 0 & 0\\
            0 & \theta & 1-\theta  & \cdots & 0 & 0 & 0\\
            0 & 0 & \theta  & \cdots & 0 & 0 & 0\\
            \vdots & \vdots  & \vdots & \ddots & \vdots & \vdots & \vdots\\
            0 & 0 & 0 & \cdots &\theta & 1-\theta & 0\\
            0 & 0 & 0 & \cdots & 0 & \theta & 1-\theta\\
        \end{bmatrix}.
    \end{align*}
    \(\Thetamat\) has \(|\calU|=|m_1|=n\) rows and \(|\calV|=|m_1|+1=n+1\) columns.\\
    Define the \(\Pmat_1\) and \(\Pmat_2\) to be the matrices that form the diagonal blocks of \(\Pmat(\Thetamat,\rho) = \crocMat{\Pmat_1}{-\rho\Thetamat}{-\rho\Thetamat^\top}{\Pmat_2}\):
    \begin{align*}
        \Pmat_1 &\define \pth{1-\rho^2}\Imat + \rho^2\diag(\Thetamat\vctr{1})\\
        &= \pth{1-\rho^2(1-\nu)}\Imat\\
        \Pmat_2 &\define \pth{1-\rho^2}\Imat + \rho^2\diag(\Thetamat^\top\vctr{1})\\
        &= \pth{1-\rho^2(1-\nu)}\Imat - \rho^2\nu\diag(1-\theta,0,0,\cdots,0,\theta)
    \end{align*}
    Then
    \begin{align*}
        \det \Pmat(\Thetamat,\rho) &= \croc{1-\rho^2}^n \det\croc{\pth{1-\rho^2(1-\nu)}\Imat - \rho^2\nu\diag(1-\theta,0,0,\cdots,0,\theta) - \frac{\rho^2}{1-\rho^2(1-\nu)}\Thetamat^\top\Thetamat}
    \end{align*}
    where \(\Pmat_2 - \rho^2\Thetamat^\top \Pmat_1^{-1} \Thetamat = \pth{1-\rho^2(1-\nu)}\Imat - \rho^2\nu\diag(1-\theta,0,0,\cdots,0,\theta) - \frac{\rho^2}{1-\rho^2(1-\nu)}\Thetamat^\top\Thetamat\) is the Schur complement of \(\Pmat_1\).\\
    Define another auxilliary matrix
    \begin{align*}
        \Qmat \define \frac{1}{2\nu^2\theta(1-\theta)}\croc{\Thetamat^\top\Thetamat - \croc{1- 2\theta(1-\theta)}\nu^2\Imat + \nu\pth{1-\rho^2(1-\nu)}\diag(1-\theta,0,0,\cdots,0,\theta)}
    \end{align*}
    such that the Schur complement of \(\Pmat_1\) can be expressed as
    \begin{align*}
        &\pth{1-\rho^2(1-\nu)}\Imat - \rho^2\nu\diag(1-\theta,0,0,\cdots,0,\theta) - \frac{\rho^2}{1-\rho^2(1-\nu)}\Thetamat^\top\Thetamat\\
        &= \croc{1-\rho^2\pth{1-\nu}-\frac{\rho^2\nu^2\pth{1-2\theta(1-\theta)}}{1-\rho^2(1-\nu)}}\Imat-\frac{2\rho^2\nu^2\theta(1-\theta)}{1-\rho^2(1-\nu)}\Qmat
    \end{align*}
    It can be shown that
    \begin{align*}
        \Qmat &\define \frac{1}{2}
        \begin{bmatrix}
            1 & 1 & 0  & \cdots & 0 & 0 & 0\\
            1 & 0 & 1  & \cdots & 0 & 0 & 0\\
            0 & 1 & 0  & \cdots & 0 & 0 & 0\\
            \vdots & \vdots  & \vdots & \ddots & \vdots & \vdots & \vdots\\
            0 & 0 & 0 & \cdots & 0 & 1 & 0\\
            0 & 0 & 0 & \cdots & 1 & 0 & 1\\
            0 & 0 & 0 & \cdots & 0 & 1 & 1
        \end{bmatrix}
        + \nu(1-\nu)(1-\rho^2)\diag(1-\theta,0,0,\cdots,0,\theta).
    \end{align*}
    We have \(\nu(1-\nu)\max\{\theta,1-\theta\}\in[-1,0]\). Then, \(\Qmat\) is an irreducible non-negative square matrix with row sums at most 1. Consequently, by the Perron-Frobenius theorem, its eigenvalues are all between \(-1\) and \(1\). Then the eigenvalues of the Schur complement of \(\Pmat_1\) are between \(1-\rho^2(1-\nu)-\frac{\rho^2\nu^2}{1-\rho^2(1-\nu)}\) and \(1-\rho^2(1-\nu)-\frac{\rho^2\nu^2(1-2\theta)^2}{1-\rho^2(1-\nu)}\). Furthermore, the sum of the eigenvalues of \(\Qmat\) is \(\tr(\Qmat)=1+\nu(1-\nu)(1-\rho^2)\). Multiplying the \(n+1\) eigenvalues of \(\Qmat\) by \(-\frac{2\rho^2\nu^2\theta(1-\theta)}{1-\rho^2(1-\nu)}\) and adding \(1-\rho^2(1-\nu)-\frac{\rho^2\nu^2\pth{1-2\theta(1-\theta)}}{1-\rho^2(1-\nu)}\) gives us the \(n+1\) eigenvalues of the Schur complement of \(\Pmat_1\). Now we plug in \(\theta=1/2\). By \hyperref[lemma:minProd]{Lemma \ref*{lemma:minProd}}, the determinant of the Schur complement of \(\Pmat_1\) is lower bounded by
    \begin{align*}
        \croc{1-\rho^2(1-\nu)}^{\frac{n}{2}}\croc{1-\rho^2(1-\nu)-\frac{\rho^2\nu^2}{1-\rho^2(1-\nu)}}^{\frac{n}{2}+1+\nu(1-\nu)\pth{1-\rho^2}}.
    \end{align*}
   Then, multiplying the determinant of the Schur complement of \(\Pmat_1\) by \(\det \Pmat_1 = \croc{1-\rho^2(1-\nu)}^n\) gives us
    \begin{align*}
        \det\Pmat(\Thetamat,\rho) &\geq \croc{1-\rho^2(1-\nu)}^{n-1}\croc{\pth{1-\rho^2(1-\nu)}^2-\rho^2\nu^2}^{\frac{n}{2}+1+\nu(1-\nu)\pth{1-\rho^2}}\\
        &= \croc{1-\rho^2(1-\nu)}^{n-1-\nu(1-\nu)\pth{1-\rho^2}}\croc{\pth{1-\rho^2}\pth{1-\rho^2(1-\nu)^2}}^{\frac{n}{2}+1+\nu(1-\nu)\pth{1-\rho^2}}.
    \end{align*}
    Since \(\nu\in[0,2]\iff(1-\nu)^2\in[0,1]\) by \hyperref[lemma:push2exp]{Lemma \ref*{lemma:push2exp}}, we have the bound \(1-\rho^2(1-\nu)^2 \geq \pth{1-\rho^2}^{(1-\nu)^2}\). Furthermore, since \(-(1-\nu)\in[0,1]\), we have the bound \(1-\rho^2(1-\nu) \geq \pth{1+\rho^2}^{-(1-\nu)}\). Then,
    \begin{align*}
        \log \det \Pmat(\Thetamat,\rho) &\geq \pth{\frac{n}{2}+1+\nu(1-\nu)\pth{1-\rho^2}}\pth{1+(1-\nu)^2}\log\pth{1-\rho^2}\\
        &\pheq-(n-1-\nu(1-\nu)\pth{1-\rho^2})(1-\nu)\log\pth{1+\rho^2}.
    \end{align*}

    \(R(\Thetamat)\) has \(\frac{1}{2}\pth{|\calU|+|\calV|-\sum_{(u,v)} \Theta_{u,v}} = \frac{1}{2}(n+n+1-\nu n) = \frac{n}{2}(2-\nu)+\frac{1}{2}\) multiplicative terms of \(\pth{1-\rho^2}\) in the numerator and our bound has \(\frac{1}{2}\pth{\frac{n}{2}+1+\nu(1-\nu)\pth{1-\rho^2}}\pth{1+(1-\nu)^2}\) of them in the denominator, which gives, in total, \(\frac{n}{4}\nu(2-\nu) - \frac{n-1}{2}(\nu-1)+\frac{1}{2}\nu(\nu-1)^3\pth{1-\rho^2}-\frac{1}{2}\nu(\nu-1)\rho^2\) such terms. Then
    \begin{align*}
        \log R(\Thetamat) &\leq \frac{1}{2}\croc{\frac{n}{2}\nu(2-\nu)+\nu^2(\nu-1)^2 - \nu^2(\nu-1)^2\rho^2-\nu(\nu-1)\rho^2}\log\pth{1-\rho^2}\\
        &\pheq - \frac{\nu-1}{2}\croc{(n-1)+\nu(\nu-1)\pth{1-\rho^2}}\croc{\log\pth{1-\rho^2}+\log\pth{1+\rho^2}}
    \end{align*}
    For the multi-dimensional case, since under the canonical form all dimensions are mutually independent, we get \(\log R(\Thetamat)\) by simply summing this expression for all \(\rho_i\). This gives us
\begin{align*}
        \log R(\Thetamat) &\leq \croc{\frac{n}{2}\nu(2-\nu)+\nu^2(\nu-1)^2}\sum_{i\in\calD}\frac{1}{2}\log\pth{1-\rho_i^2}\\
        &\pheq - (\nu-1)\croc{\nu^2(\nu-1)+\nu}\sum_{i\in\calD}\frac{1}{2}\rho_i^2\log\pth{1-\rho_i^2}\\
        &\pheq -(\nu-1)(n-1)\sum_{i\in\calD}\frac{1}{2}\croc{\log\pth{1-\rho_i^2}+\log\pth{1+\rho_i^2}}\\
        &\pheq + (\nu-1)^2\nu \sum_{i\in\calD}\frac{1}{2}\rho_i^2\croc{\log\pth{1-\rho^2}+\log\pth{1+\rho^2}}
    \end{align*}
    Given \(\nu\geq 1\), the last term is non-positive, and therefore we can drop it while maintaining the inequality. By \hyperref[lemma:compareEpsilon]{Lemma \ref*{lemma:compareEpsilon}}, \(-\frac{1}{2}\sum \log\pth{1-\rho_i^2}+\log\pth{1+\rho_i^2}\) is upper bounded by \(-\frac{1}{2}\sum \rho_i^2\log\pth{1-\rho_i^2}\), which itself is upper bounded by \(\rho_{\max}^2 I_{XY}\) where \(\rho_{\max} = \max |\rho_i|\) and \(I_{XY} = -\frac{1}{2}\sum \log\pth{1-\rho_i^2}\). This gives us 
    \begin{align*}
        \frac{\log R(\Thetamat)}{I_{XY}} \leq -\frac{n}{2}\cdot \nu (2-\nu) - (\nu-1)^2\nu^2 + n\cdot \rho_{\max}^2(\nu-1) + \rho_{\max}^2(\nu-1)^2(\nu^2+1)
    \end{align*}
    Since \(n\geq 1\) and \(\nu\in[1,2]\), we have \((\nu-1)^2(\nu^2+1) \leq 5n(\nu-1)\). This gives us the claimed result.

\end{proof}

\section{Statistics of the information density matrix}
\label{app:stats}

Recall the definition of \(\Gmat\): Let \(f_{XY}\), \(f_X\) and \(f_Y\) denote the joint and marginal distributions for correlated features in \(\Avec\) and \(\Bvec\).
\begin{align*}
    G_{u,v} &= \log \frac{f_{XY}(\Avec(v),\Bvec(v))}{f_X(\Avec(u))f_Y(\Bvec(v))}
\end{align*}

\begin{lemma}
\label{lemma:stats}
The expressions for the first and second moments of \(\Gmat\) in terms of the correlation vector \(\rhovec\) as defined in \hyperref[sec:canonical]{Appendix \ref*{sec:canonical}} are given below.
\begin{compactenum}[(a)]
    \item Mean and variance of information density of a true pair:
    \begin{align*}
        \bbE\croc{G_{u,v}|u\mapped{M}v} &= \sum_{i\in\calD} - \frac{1}{2}\log\pth{1-\rho_i^2}\\
        \Var\pth{G_{u,v}|u\mapped{M}v} &= \sum_{i\in\calD}\rho_i^2
    \end{align*}
    \item Mean and variance of information density of a false pair:
    \begin{align*}
        \bbE\croc{G_{u,v'}|u\notmapped{M}v'} &= \sum_{i\in\calD} - \frac{1}{2}\log\pth{1-\rho_i^2} - \frac{\rho_i^2}{1-\rho_i^2}\\
        \Var\pth{G_{u,v'}|u\notmapped{M}v'} &= \sum_{i\in\calD} \frac{\rho_i^2(1+\rho_i^2)}{\pth{1-\rho_i^2}^2}
    \end{align*}
    
    \bcomment{
    \item Covariance between information density of a true match and a false match with a user in common:
    \begin{align*}
        \Cov\pth{G_{u,v},G_{u,v'}|u\mapped{M}v} &= \sum_{i\in\calD} -\frac{\rho_i^4}{2(1-\rho_i^2)}
    \end{align*}
    \item Covariance between information density of two false matches with a user in common:
    \begin{align*}
        \Cov\pth{G_{u,v'},G_{u,v''}|u\mapped{M}v} &= \sum_{i\in\calD}\frac{\rho_i^4}{4\pth{1-\rho_i^2}^2}
    \end{align*}
    \item Covariance between information density of two false matches that break apart two true matches:
    \begin{align*}
        \Cov\pth{G_{u_2,v_1},G_{u_1,v_2}|u_1\mapped{M}v_1,u_2\mapped{M}v_2} &= \sum_{i\in\calD}\frac{\rho_i^4}{1-\rho_i^2}
    \end{align*}
    \item Covariance between information density of two false matches that break apart one true match:
    \begin{align*}
        \Cov\pth{G_{u_2,v_1},G_{u_1,v_2}|u_1\mapped{M}v_1,u_2\notmapped{M}v_2} &= \sum_{i\in\calD} \frac{\rho_i^4(1+\rho_i^2)}{2\pth{1-\rho_i^2}^2}
    \end{align*}
    \item Covariance between information density of two matches that do not break apart any true match:
    \begin{align*}
        \Cov\pth{G_{u_2,v_1},G_{u_1,v_2}|u_1\notmapped{M}v_1,u_2\notmapped{M}v_2} &= 0
    \end{align*}
    }
\end{compactenum}
\end{lemma}

\begin{proof}
Assume \(\Avec\) and \(\Bvec\) are given in canonical form with correlation vector \(\rhovec\in[-1,1]^{\calD}\). Then
\begin{align}
    \label{eq:Gxy}
    \log \frac{f_{XY}(\xvec,\yvec)}{f_X(\xvec)f_Y(\yvec)} &= \log \frac{\prod_i f_{X_iY_i}(x_i,y_i)}{\prod_i f_{X_i}(x_i)\prod_if_{Y_i}(y_i)}\nonumber\\
    &= \sum_{i\in\calD} \log\frac{f_{X_iY_i}(x_i,y_i)}{f_{X_i}(x_i)f_{Y_i}(y_i)}\nonumber\\
    &= \sum_{i\in\calD} \log\frac{\frac{1}{2\pi\sqrt{1-\rho_i^2}}\exp\pth{-\frac{1}{2}\crocVec{x_i}{y_i}^\top\crocMat{1}{\rho}{\rho}{1}^{-1}\crocVec{x_i}{y_i}}}{\frac{1}{\sqrt{2\pi}}\exp\pth{-\frac{x_i^2}{2}}\cdot\frac{1}{\sqrt{2\pi}}\exp\pth{-\frac{y_i^2}{2}}}\nonumber\\
    &= \sum_{i\in\calD} -\frac{1}{2}\log\pth{1-\rho_i^2}-\frac{\rho_i^2(x_i^2+y_i^2)-2\rho_ix_iy_i}{2(1-\rho_i^2)}\nonumber\\
    &= \sum_{i\in\calD} -\frac{1}{2}\log\pth{1-\rho_i^2}-\frac{\rho_i^2(x_i-y_i)^2}{2(1-\rho_i^2)}+\frac{\rho_ix_iy_i}{1+\rho_i}
\end{align}

Then \(G_{u,v}\) can be written as the sum of \(|\calD|\) independent random variables, each a function of only \(A_i(u),B_i(v)\) and \(\rho_i\). For the rest of this section we assume \(|\calD|=1\) and drop all subscripts \(i\) for simplicity of notation. The means, variances or covariances for the case \(|\calD|>1\) can be found by summing over the corresponding expression for different values of \(\rho\).

For the derivations (a) to (d), assume \(u\mapped{M}v\) and let \(X\define A(u)\), \(Y\define B(v)\). Then \((X,Y)\sim\calN\pth{\zerovec,\crocMat{1}{\rho}{\rho}{1}}\). Define \(Z \define \frac{Y-\rho X}{\sqrt{1-\rho^2}}\). Then \(Y = \rho X + \sqrt{1-\rho^2}Z\). Also define \(W \define B(v')\). Then \(W\), \(X\) and \(Z\) are i.i.d. standard normal random variables.

\begin{compactenum}[(a)]
    \item \(\bbE\croc{G_{u,v}|u\mapped{M}v}\) and \(\Var\pth{G_{u,v}|u\mapped{M}v}\):
    
    We want to find the mean and variance of \(-\frac{1}{2}\log\pth{1-\rho^2}-\frac{\rho^2(X-Y)^2}{2(1-\rho^2)}+\frac{\rho XY}{1+\rho}\). The first term is constant. We find the mean, variance and covariance of the two other terms.
    
    \(X-Y = (1-\rho)X - \sqrt{1-\rho^2}Z \sim \calN(0,2(1-\rho))\), which implies that \(\frac{(X-Y)^2}{2(1-\rho)}\sim\chi^2(1)\) and therefore \(\frac{\rho^2(X-Y)^2}{2(1-\rho^2)}\) has mean \(\frac{\rho^2}{1+\rho}\) and variance \(\frac{2\rho^4}{(1+\rho)^2}\).
    
    \(XY = \rho X^2 + \sqrt{1-\rho^2}XZ\). Then \(XY\) has mean \(\rho\).
    \begin{align*}
        (XY-\rho)^2 &= \rho^2X^4 + 2\rho\sqrt{1-\rho^2}X^3Z + (1-\rho^2)X^2Z^2\\
        &\pheq - 2\rho^2 X^2 - 2\rho\sqrt{1-\rho^2}XZ +\rho^2\\
        \implies \Var(XY) &= \bbE[(XY-\rho)^2]\\
        &= 3\rho^2 + (1-\rho^2) -2\rho^2 + \rho^2\\
        &= 1+\rho^2
    \end{align*}
    So \(\frac{\rho XY}{1+\rho}\) has mean \(\frac{\rho^2}{1+\rho}\) and variance \(\frac{\rho^2(1+\rho^2)}{(1+\rho)^2}\).
    
    \begin{align*}
        \Cov((X-Y)^2,XY) &= \bbE\croc{\pth{(X-Y)^2-2(1-\rho)}\pth{XY-\rho}}\\
        \pth{(X-Y)^2-2(1-\rho)}\pth{XY-\rho} &= -2(XY-\rho)^2 -2(XY-\rho) + (X^2+Y^2)(XY-\rho)\\
        &= -2(XY-\rho)^2 -2(XY-\rho) + X^3Y+XY^3-\rho X^2 -\rho Y^2\\
        \implies \Cov((X-Y)^2,XY) &= -2\Var(XY) + 2\bbE[X^3Y] - 2\rho\bbE[X^2]\\
        &= -2(1+\rho^2) +2 \bbE\croc{\rho X^4 + \sqrt{1-\rho^2}X^3Z} - 2\rho\\
        &= -2-2\rho^2 + 6\rho - 2\rho\\
        &= -2(1-\rho)^2
    \end{align*}
    The covariance between \(\frac{\rho^2(X-Y)^2}{2(1-\rho^2)}\) and \(\frac{\rho XY}{1+\rho}\) equals \(\frac{-\rho^3(1-\rho)}{(1+\rho)^2}\).
    
    Then
    \begin{align*}
        \bbE\croc{G_{u,v}|u\mapped{M}v} &= \bbE\croc{-\frac{1}{2}\log\pth{1-\rho^2}-\frac{\rho^2(X-Y)^2}{2(1-\rho^2)}+\frac{\rho XY}{1+\rho}}\\
        &= -\frac{1}{2}\log\pth{1-\rho^2} - \frac{\rho^2}{1+\rho} + \frac{\rho^2}{1+\rho}\\
        &= -\frac{1}{2}\log\pth{1-\rho^2}\\
        \Var\pth{G_{u,v}|u\mapped{M}v} &= \Var\pth{\frac{\rho^2(X-Y)^2}{2(1-\rho^2)}} + \Var\pth{\frac{\rho XY}{1+\rho}} - 2\Cov\pth{\frac{\rho^2(X-Y)^2}{2(1-\rho^2)},\frac{\rho XY}{1+\rho}}\\
        &= \frac{2\rho^4}{(1+\rho)^2} + \frac{\rho^2(1+\rho^2)}{(1+\rho)^2} + \frac{2\rho^3(1-\rho)}{(1+\rho)^2}\\
        &= \frac{\rho^2+2\rho^3+\rho^4}{(1+\rho)^2} = \rho^2
    \end{align*}
    
    \item \(\bbE\croc{G_{u,v'}|u\notmapped{M}v'}\) and \(\Var\pth{G_{u,v'}|u\notmapped{M}v'}\):
    
    Once again, we find the mean, variance and covariance of the random terms in \(-\frac{1}{2}\log\pth{1-\rho^2}-\frac{\rho^2(X-W)^2}{2(1-\rho^2)}+\frac{\rho XW}{1+\rho}\).
    
    \(X-W\sim\calN(0,2)\), which implies that \(\frac{(X-W)^2}{2}\sim\chi^2(1)\) and therefore \(\frac{\rho^2(X-W)^2}{2(1-\rho^2)}\) has mean \(\frac{\rho^2}{1-\rho^2}\) and variance \(\frac{2\rho^4}{\pth{1-\rho^2}^2}\).
    
    \(\bbE[XW]=0\) and \(\Var(XW) = \bbE[X^2W^2] = \bbE[X^2]\bbE[W^2]=1\). Then \(\frac{\rho XW}{1+\rho}\) has mean \(0\) and variance \(\frac{\rho^2}{(1+\rho)^2}\).
    
    \begin{align*}
        \Cov((X-W)^2,XW) &= \bbE\croc{\pth{(X-W)^2-2}XW}\\
        \pth{(X-W)^2-2}XW &= X^3W + XW^3 - 2X^2W^2 - 2XW\\
        \implies \Cov((X-W)^2,XW) &= -2
    \end{align*}
    The covariance between \(\frac{\rho^2(X-W)^2}{2(1-\rho^2)}\) and \(\frac{\rho XY}{1+\rho}\) equals \(\frac{-\rho^3}{(1+\rho)(1-\rho^2)}\).
    
    Then
    \begin{align*}
        \bbE\croc{G_{u,v'}|u\notmapped{M}v'} &= \bbE\croc{-\frac{1}{2}\log\pth{1-\rho^2}-\frac{\rho^2(X-W)^2}{2(1-\rho^2)}+\frac{\rho XW}{1+\rho}}\\
        &= -\frac{1}{2}\log\pth{1-\rho^2} - \frac{\rho^2}{1-\rho^2}\\
        \Var\pth{G_{u,v'}|u\notmapped{M}v'} &= \Var\pth{\frac{\rho^2(X-W)^2}{2(1-\rho^2)}} + \Var\pth{\frac{\rho XW}{1+\rho}} - 2\Cov\pth{\frac{\rho^2(X-W)^2}{2(1-\rho^2)},\frac{\rho XW}{1+\rho}}\\
        &= \frac{2\rho^4}{\pth{1-\rho^2}^2} + \frac{\rho^2}{(1+\rho)^2} + \frac{2\rho^3}{(1+\rho)(1-\rho^2)}\\
        &= \frac{2\rho^4+\rho^2(1-\rho)^2 + 2\rho^3(1-\rho)}{\pth{1-\rho^2}^2}\\
        &= \frac{\rho^2(1+\rho^2)}{\pth{1-\rho^2}^2}
    \end{align*}

\end{compactenum}
\end{proof}

\begin{lemma}
    \label{lemma:highDimStats}
    Under \hyperref[cond:highDimensional]{Condition \ref*{cond:highDimensional}},
    \begin{compactenum}[(a)]
        \item Mean and variance of information density of a true pair:
        \begin{align*}
            \bbE\croc{G_{u,v}|u\mapped{M}v} &= I_{XY}\\
            \Var\pth{G_{u,v}|u\mapped{M}v} &= 2I_{XY}(1-o(1))
        \end{align*}
        \item Mean and variance of information density of a false pair:
        \begin{align*}
            \bbE\croc{G_{u,v'}|u\notmapped{M}v'} &= -I_{XY}(1+o(1))\\
            \Var\pth{G_{u,v'}|u\notmapped{M}v'} &= 2I_{XY}(1+o(1))
        \end{align*}
    \end{compactenum}
\end{lemma}
\begin{proof}
    Exact expressions for the statistics are given in \hyperref[lemma:stats]{Lemma \ref*{lemma:stats}}. The expression for \(\bbE\croc{G_{u,v}|u\mapped{M}v}\) is exactly equal to the expression for mutual information \(I_{XY} = -\frac{1}{2}\sum \log\pth{1-\rho_i^2}\).

    By the Taylor series expansion of \(\log\), given \(x\in(0,1)\),
    \begin{align*}
        -\frac{1}{x}\log(1-x) &= \frac{1}{x}\sum_{k=1}^\infty \frac{x^k}{k}\\
        &< \frac{1}{x}\sum_{k=1} x^k = \frac{1}{1-x}
    \end{align*}
    So \(-\log(1-x) < \frac{x}{1-x}\). Furthermore, the limit of \(-\frac{1}{x}\log(1-x)\) and the limit of \(\frac{1}{1-x}\) are both 1 as \(x\to 0\). Then, given \(x\leq o(1)\), \(\frac{x}{1-x} = -(1+o(1))\log(1-x)\). Then \(\sum \frac{\rho_i^2}{1-\rho_i^2} = 2I_{XY}(1+o(1))\), which gives us \(\bbE\croc{G_{u,v'}|u\notmapped{M}v'} = I_{XY} -  2I_{XY}(1+o(1)) = -I_{XY}(1+o(1))\).

    Finally, since \(-\log(1-x) < \frac{x}{1-x} < \frac{x(1+x)}{(1-x)^2}\) and the limits of \(-\frac{1}{x}\log(1-x)\) and \(\frac{(1+x)}{(1-x)^2}\) are both 1 as \(x\to 0\), we have \(\sum \frac{\rho_i^2\pth{1+\rho_i^2}}{\pth{1-\rho_i^2}^2} = 2I_{XY}(1+o(1))\), which gives us the last part of the claim.
    
    So \(\frac{\rho_i^2}{1-\rho_i^2} > -\log\pth{1-\rho_i^2}\).
\end{proof}

\section{Other lemmas}

\begin{lemma}[Bernoulli's inequality]
    \label{lemma:push2exp}
    Let \(a\in(-1,\infty)\). Then
    \begin{itemize}
        \item \((1+ax)>(1+a)^x\) if \(x\in(0,1)\).
        \item \((1+ax)< (1+a)^x\) if \(x\in\bbR\setminus[0,1]\).
        \item If \(|ax|\leq o(1)\),\\
        then \(\left|\frac{1+ax}{(1+a)^x}-1\right|\leq o(1)\).
    \end{itemize}
\end{lemma}
\begin{proof}
    Let \(f(x) \define (1+a)^x-(1+ax)\). \(f(0)=f(1)=0\). The first derivative is given by \(f'(x)=(1+a)^x\log(1+a)-a\). \(f'(0)\) is strictly negative for \(x=0\) since \(\log(1+a)<a\) and \(f'(1)\) is strictly positive for \(x=1\) since \(\log(1+a)\geq\frac{a}{1+a}\). Furthermore the second derivative \(f''(x) = (1+a)^x\log^2(1+a)\) is strictly positive. Then this function has a global minimum at some \(x^*\in(0,1)\), is strictly decreasing for \(x<x^*\) and strictly increasing for \(x>x^*\). Then \(f(x)<0\) for \(x\in(0,1)\) and \(f(x)>0\) for \(x\in(-\infty,0)\cap(1,\infty)\).
    
    The Taylor series expansion for \(f(x)\) at \(x=0\) is given by
    \begin{align*}
        f(x) &=-ax + \sum_{k=1}^\infty \frac{x^k\log^k(1+a)}{k!},
    \end{align*}
    which is \(-o(1) + o(1)\) if \(|ax|\leq o(1)\). Then \(\left|\frac{1+ax}{(1+a)^x}-1\right|\leq o(1)\).

    \bcomment{
    
    The Taylor series expansion for \(f(x)\) at \(x=1\) is given by
    \begin{align*}
        f(x)&=-a(x-1) + (1+a)\sum_{k=1}^\infty \frac{(x-1)^k\log^k(1+a)}{k!},
    \end{align*}
    which is \(-o(1) + (1+a)o(1)\) if \(|a(x-1)|\leq o(1)\). Then \(\left|\frac{1+ax}{(1+a)^x}-1\right|\leq o(1)\).

    }
    
\end{proof}

\begin{lemma}
    \label{lemma:minProd}
    Let  \(\tau\in(0,\infty)\), \(\sigma\in(-\tau,\tau)\), \(x_1,\cdots,x_n\in[-1,1]\) and \(s = \sum x_i\).
    Then
    \begin{align*}
        \prod_{i=1}^n (\tau+\sigma x_i) \geq \pth{\tau-\sigma}^{\frac{n-s}{2}}\pth{\tau+\sigma}^{\frac{n+s}{2}}
    \end{align*}
\end{lemma}

\begin{proof}
    Define \(\theta_i=\frac{x_i+1}{2}\in[0,1]\). By the concavity of \(\log\),
    \begin{align*}
        \log\pth{\tau+\sigma x_i} &= \log\pth{\theta_i(\tau+\sigma)+(1-\theta_i)(\tau-\sigma)}\\
        &\geq \theta_i\log(\tau+\sigma)+(1-\theta_i)\log(\tau-\sigma)\\
        &= \frac{x_i}{2}\log\pth{\frac{\tau+\sigma}{\tau-\sigma}}+\frac{1}{2}\log(\tau^2-\sigma^2)
    \end{align*}
    
    Then
    \begin{align*}
        \prod_{i=1}^n (\tau+\sigma x_i) &\geq \exp\pth{\sum_{i=1}^n \frac{x_i}{2}\log\pth{\frac{\tau+\sigma}{\tau-\sigma}}+\frac{1}{2}\log(\tau^2-\sigma^2)}\\
        &= \exp\pth{\frac{s}{2}\log\pth{\frac{\tau+\sigma}{\tau-\sigma}}+\frac{n}{2}\log(\tau^2-\sigma^2)}\\
        &= \pth{\tau+\sigma}^{\frac{n+s}{2}}\pth{\tau-\sigma}^{\frac{n-s}{2}}
    \end{align*}
\end{proof}

\begin{lemma}
    \label{lemma:compareEpsilon}
    Let \(\rho_{\max} \define \max_i |\rho_i|\) and \(I_{XY} \define -\frac{1}{2}\sum_i \log\pth{1-\rho_i^2}\). Then
    \begin{align*}
        -\sum_i \log\pth{1-\rho_i^2} + \log\pth{1+\rho_i^2} \leq \sum_i \rho_i^2\log\pth{1-\rho_i^2} \leq \rho_{\max}^2 I_{XY}
    \end{align*}
\end{lemma}
\begin{proof}
    First we show that, for any \(x\in[0,1]\), \(-x\log(1-x) \geq -\log(1-x)-\log(1+x)\):

    By the Taylor series expansion,
    \begin{align*}
        -x\log\pth{1-x} &= x\sum_{k=1}^\infty \frac{x^k}{k} = \sum_{k=2}^\infty \frac{x^{k}}{k-1}\\
        &= \sum_{k=1}^\infty \frac{x^{2k}}{2k-1} + \frac{x^{2k+1}}{2k}
    \end{align*}
    For \(x\in[0,1]\), \(\frac{x^{2k}}{2k-1} + \frac{x^{2k+1}}{2k} \geq x^{2k}\pth{\frac{1}{2k-1} + \frac{1}{2k}} = \frac{x^{2k}}{k}\pth{1+\frac{1}{2(2k-1)}}\), which is strictly greater than \(\frac{x^{2k}}{k}\) for any \(k>1/2\). Then
    \begin{align*}
        -x\log\pth{1-x} &= \sum_{k=1}^\infty \frac{x^{2k}}{2k-1} + \frac{x^{2k+1}}{2k}\\
        &> \sum_{k=1}^\infty \frac{x^{2k}}{k}\\
        &\eql{(*)} -\log\pth{1-x^2}\\
        &= -\log\pth{1-x} -\log \pth{1+x}
    \end{align*}
    where \((*)\) follows from the Taylor series expansion. Consequently, we have
    \begin{align*}
        -\rho^2\log\pth{1-\rho^2} \geq - \log\pth{1-\rho^2} - \log\pth{1+\rho^2}
    \end{align*}
    for any \(\rho\in[-1,1]\).

    The result follows from the fact that \(\sum \rho_i^2\log\pth{1-\rho_i^2} \leq \max_i \rho_i^2 \log\pth{1-\rho_i^2}\).
\end{proof}

\end{document}